\documentclass[11pt]{article}
\usepackage[margin=1in]{geometry}
\usepackage{authblk}
\PassOptionsToPackage{final}{microtype}
\PassOptionsToPackage{backend=biber,style=ieee,citestyle=numeric-comp,maxcitenames=1,mincitenames=1,natbib=true}{biblatex}
\PassOptionsToPackage{colorlinks}{hyperref}
\input{preamble.tex}

\def\tmp#1#2#3{%
  \definecolor{Hy#1color}{#2}{#3}%
  \hypersetup{#1color=Hy#1color}}
\tmp{link}{HTML}{800006}
\tmp{cite}{HTML}{2E7E2A}
\tmp{file}{HTML}{131877}
\tmp{url} {HTML}{8A0087}
\tmp{menu}{HTML}{727500}
\tmp{run} {HTML}{137776}
\def\tmp#1#2{%
  \colorlet{Hy#1bordercolor}{Hy#1color#2}%
  \hypersetup{#1bordercolor=Hy#1bordercolor}}
\tmp{link}{!60!white}
\tmp{cite}{!60!white}
\tmp{file}{!60!white}
\tmp{url} {!60!white}
\tmp{menu}{!60!white}
\tmp{run} {!60!white}

\title{Toward Circuit-Aware Mutation Testing for Quantum Programs:\\Exact Visibility and Execution Planning}

\author[1]{Andriy Miranskyy}
\author[2]{Jos{\'e} Campos}
\affil[1]{%
\textit{Department of Computer Science, Toronto Metropolitan University}\\
Toronto, Canada\\
\texttt{avm@torontomu.ca}}
\affil[2]{%
\textit{Faculdade de Engenharia, Universidade do Porto}\\
\textit{LASIGE, Faculdade de Ci\^encias, Universidade de Lisboa}\\
Porto, Portugal\\
\texttt{jcmc@fe.up.pt}}

\date{}
\begin{document}
\maketitle
\begin{abstract}
Mutation testing assesses test suites by introducing program changes and checking whether tests detect them.
In quantum programs, a surviving mutant may leave program behavior unchanged, produce a change hidden by the chosen input or measurement, or remain undetected because statistical evidence is insufficient.
We develop a circuit-aware framework to distinguish these causes and guide execution planning and test changes.
The framework covers ideal, noiseless programs and mutations that replace one gate or subcircuit.
The method traces how a mutation can affect measured outputs and combines established mathematical checks with statistical estimates.
For supported cases, we establish conditions under which a test cannot reveal a mutation and exact formulas for its effect when the required state information is available.
Under stated assumptions, we bound the mutation score while retaining surviving and unresolved mutants, and identify a restricted circuit family whose mutation effects admit efficient classical classification.
The deliverables include a decision procedure, mathematical guarantees, worked examples of test changes and their implementation requirements, and a small-circuit reference implementation for differential expectation tests.
These results help testers decide what to execute, what to change, and what remains unknown.
The evidence is limited to proofs and executable examples under the ideal model.

\end{abstract}
\section{Introduction}
Mutation testing evaluates whether a test suite can detect plausible faults.
A mutation is a program change introduced to model such a fault, and a mutant is the modified program that results.
A surviving mutant may preserve the declared program semantics or may change behavior without violating the suite's assertions.
Whether such survival reveals a testing weakness depends on how the program contract and the oracle treat the behavioral difference.
In quantum programs, measurement, finite shots, and the potentially enormous state space further complicate this distinction.
A \emph{shot} is one execution from input preparation through measurement; repeated shots estimate the program's output probabilities or expected measurement outcome.

Consider a one-qubit program that starts in state $\ket{0}$ (the quantum representation of bit 0), applies a rotation $R_y(\pi/4)$, and measures the output bit.
A sign-error mutant replaces that rotation by $R_y(-\pi/4)$.
For this input, both programs return bit 1 with probability $\sin^2(\pi/8)\approx0.1464$, so an assertion over those computational-basis frequencies cannot distinguish them, regardless of the shot count.
The output states nevertheless differ: an $X$-basis measurement, labeling the $-1$ outcome as bit 1, gives probabilities approximately $0.1464$ and $0.8536$, respectively.
If the program's entire interface is this fixed input and computational-basis output bit, the two implementations are observationally equivalent at that interface.
If the output state is available to subsequent quantum operations, changing the measurement can expose a meaningful state difference.

\subsection{Scope and Claims}
\label{sec:scope-claims}
Interpreting mutant survival requires separating semantic difference, visibility to a selected test, and evidence obtainable within an execution budget.
Insufficient evidence establishes neither equivalence nor invisibility.
Throughout the formal development, the circuit is a state-transforming program whose output may in principle be examined by different measurements, and the scalar oracle compares its expectation with the original program's expectation.
An internal state difference need not violate a classical output contract, and a differential expectation gap need not violate an independently specified acceptance assertion.

Earlier work shows why classical testing and debugging do not transfer wholesale: intermediate measurements can disturb program state, finite-shot outputs are probabilistic, and classical simulation does not scale \cite{miranskyy2019testing,miranskyy2020bugfree,miranskyy2021testingdebugging}.
The recent quantum-software-engineering roadmap therefore identifies efficient test oracles, scalable test optimization, mutant subsumption, and fault representativeness as open problems \cite{murillo2025roadmap}; our circuit-local analysis contributes exact visibility certificates and mutation-guided measurement design to the first two directions.
A recent perspective similarly links quantum test adequacy to fault-based sensitivity and statistical power at a fixed shot budget \cite{abreu2026testing}, motivating explicit treatment of both mutation-induced observable gaps and the resources needed to detect them.
This paper focuses first on the ideal, noiseless setting.
We model a quantum program as a sequence of reversible quantum gates on a fixed register, followed by a measurement used by the test oracle.
The ideal model lets us separate a change to program behavior from the ability of a particular test to observe that change.
This separation guides execution planning.
Practical application on quantum hardware will require device-specific noise and drift models.
The ideal, noiseless analysis in this paper provides a baseline for understanding which mutation effects the declared tests can detect, against which changes caused by hardware noise can subsequently be assessed.

Quantum circuits with mixed states can represent non-unitary gates, including intermediate measurements and noise \cite{aharonov1998mixed}.
Mid-circuit measurements, resets, classical control, non-unitary channels, leakage out of a qubit's two-dimensional computational space, and device noise are outside the formal model.
On noisy hardware, a mutation that is semantically equivalent in the ideal circuit model may still change execution through added gates, calibration effects, routing, unintended interactions between qubits (crosstalk), or other device-level artifacts.

The analysis starts from declared mutants.
The space of possible quantum mutations is infinite: one may perturb angles by arbitrary amounts, replace gates in many ways, alter controls, permute wires, or modify subcircuits.
When analyzing these mutants and planning their execution, we keep \emph{realism}, \emph{non-equivalence}, \emph{detectability}, and \emph{non-redundancy} distinct.
Here, \emph{realistic} means resembling plausible programming, modeling, or compilation mistakes; \emph{non-equivalent} means changing the ideal output state for at least one input in the declared domain; \emph{detectable} means reaching a selected test's declared detection threshold; and \emph{non-redundant} means not duplicating the effect of another already selected mutant.
These properties serve different purposes and must not be collapsed into a single adequacy claim.

For test-suite evaluation, the mutant population or its sampling design must be fixed independently of the suite being evaluated; otherwise, selecting mutants because that same suite produces a large gap makes a high mutation score partly true by construction.
For execution planning with a fixed suite, by contrast, test-conditional information may legitimately be used to avoid executions that are provably irrelevant or unaffordable, provided that the resulting detection yield is not reported as an unbiased adequacy score.

The local-unitary model treats one gate or subcircuit replacement at a time.
Higher-order mutants are outside its direct scope, although the exact relative-unitary identity applies to any mutated subcircuit represented by one unitary replacement.
A quantum compiler, or \emph{transpiler}, rewrites a circuit into gates and qubit mappings for a target device.
A source-level edit may expand into multiple gate-level changes after transpilation; such a case must be represented as a subcircuit replacement or analyzed at the source-level semantic boundary.

\subsection{Objectives and Contributions}
Our \textbf{goal} is to provide a mathematically justified basis for acting on declared quantum mutants under explicit program semantics and test-oracle rules.
The practical decisions are whether to execute a mutant against the current suite, whether a permitted change of test can reveal it, and how to report the result when evidence is incomplete.

The contribution is a testing procedure that combines established quantum response and statistical results into three kinds of decisions.
\begin{enumerate}
\item \textbf{Diagnose surviving mutants.}
We distinguish semantic equivalence from input-dependent invisibility, measurement-dependent invisibility, finite-strength cancellation, and insufficient evidence.
For fixed-axis rotations, two state expectations determine every strength with zero gap and the maximum possible gap for the selected test.
For pure stabilizer states, Pauli mutation axes and observables, and a Clifford suffix, the response has one of three shapes that can be determined in polynomial time.
\item \textbf{Plan executions with sound reporting.}
A constant-cost angle check can certify that a rotation is below a differential expectation or total-variation threshold before any state calculation.
For native interval assertions, a bound smaller than the original value's acceptance margin certifies survival; larger changes require evaluating the actual acceptance rule.
The planner combines valid certificates and uncertainty bounds into a mutation-score interval that retains the weights of survivors and unresolved mutants.
\item \textbf{Guide test-oracle improvement.}
State-side and observable-side certificates identify when the input or measurement must change.
The Helstrom observable supplies the unrestricted optimum for a fixed input, against which permitted measurements and their implementation costs can be assessed.
\end{enumerate}
The reference implementation \cite{quantummutantplanner} and worked examples make the decisions reproducible.

\Cref{sec:framework} introduces the quantum model through inputs, assertions, and mutation effects before stating the formal results.
\Cref{sec:core-workflow} explains the execution decisions and score accounting, and \cref{sec:core-examples} develops the examples and practitioner guidance.
\Cref{sec:core-validation} states the validation requirements and limitations.
\Cref{sec:related-work} situates the framework within mutation-testing research and identifies the established mathematical tools it reuses, and \cref{sec:conclusion} summarizes the engineering contribution.
The appendices retain the detailed derivations, algorithms, and estimator comparisons for readers implementing or checking the framework.

\section{Analyzing Mutants and Tests}
\label{sec:framework}
This section states the results needed to understand and apply the framework.
The detailed statistical argument, exact structured derivations, expanded discussion, and selection pseudocode are given in \cref{app:detailed-framework}, while the general perturbation proof is given in \cref{app:local-perturbation-derivation}.

\subsection{Reading the Quantum Model as a Software Test}
\label{sec:quantum-reading-guide}

Our reference-based expectation-value test prepares the same input for the original and mutant circuits, measures the same output quantity, and compares their expected values.
We call this a \emph{differential expectation-value test} because detection depends on the difference between those values, using the original program as the reference and a declared threshold to decide which differences count.
The quantum notation describes what is prepared and measured; the test oracle supplies the rule for judging the result.
We assume familiarity with matrices and basic probability and introduce the quantum concepts used by the analysis here.

A pure state $\ket{\psi}$ is a unit-length complex vector, and $\bra{\psi}$ is its conjugate transpose.
For $n$ qubits the vector has $2^n$ entries, called \emph{amplitudes}, indexed by output bit strings; measuring in the \emph{computational basis} returns a bit string with probability equal to the squared magnitude of its entry.
The states $\ket{0}$ and $\ket{1}$ are the two computational-basis states of one qubit, while $\ket{+}=(\ket{0}+\ket{1})/\sqrt{2}$ is a superposition.
A \emph{density operator} $\rho$ represents either a pure state, $\rho=\ket{\psi}\bra{\psi}$, or a probabilistic mixture of states; it is a positive semidefinite matrix with trace one.
The circuit state mathematically describes the system; a test obtains information about it through measurement.

A gate is represented by a \emph{unitary} matrix $G$: $G^\dagger G=GG^\dagger=I$, where $\dagger$ denotes conjugate transpose and $I$ is the identity.
It updates a density operator to $G\rho G^\dagger$; composing gates gives the circuit matrix $U$.
A \emph{channel} describes the input-to-output transformation of states; in our unitary model this is $\rho\mapsto U\rho U^\dagger$.
Multiplying $U$ by a scalar global phase $e^{i\chi}$ leaves this transformation unchanged, so different circuit matrices can transform every input state in the same way.
A \emph{relative phase} is a phase difference between basis amplitudes; changing it can alter later measurement probabilities.

A joint state is a \emph{product state} if it factors into states of its subsystems; it is \emph{entangled} if it cannot be expressed even as a probabilistic mixture of such product states.
An entangling gate can create entanglement from a product input.
For example, the controlled-NOT gate $\mathrm{CNOT}_{c,t}$ flips target bit $t$ when control bit $c$ is 1 and leaves it unchanged when $c$ is 0, extending linearly to superpositions.

An \emph{observable} $O$ specifies a measurement and its numerical outcome labels and is represented by a Hermitian matrix, $O=O^\dagger$.
The expectation $\Tr(O\rho_{\mathrm{out}})$ is the mean outcome on output state $\rho_{\mathrm{out}}$, where the trace $\Tr$ is the sum of diagonal matrix entries.
For example, measuring $Z$ labels computational-basis outcomes 0 and 1 by $+1$ and $-1$, respectively, so its expectation is $p(0)-p(1)$.
Measuring the projector $\ket{1}\bra{1}$ instead uses labels 0 and 1, so its expectation is simply $p(1)$.
The Pauli observables $X$, $Y$, and $Z$ are the three standard one-qubit measurement axes, each with outcomes $\pm1$; changing the axis generally requires additional gates before computational-basis readout.
The gates $R_x$, $R_y$, and $R_z$ rotate a qubit about these respective axes by the supplied angle.
A Pauli string is a tensor product of $I$, $X$, $Y$, and $Z$ across qubits, measuring the product of the corresponding local outcomes; for example, $Z_1Z_2$ measures a two-qubit correlation.
The tensor-product symbol $\otimes$ combines subsystem operators or states, and subscripts identify the qubits on which an operator acts.

Changing the measurement is therefore part of changing the test, and it can expose a mutant that survives another measurement on the same input.
The sign-reversal example in the introduction requires precisely this change, provided that the declared interface permits it.
Throughout, \emph{population-level} expectation or detection refers to the exact output probabilities before finite sampling; a \emph{mutant population} separately denotes the declared set of program variants.
The formal model uses a finite-dimensional register, bounded Hermitian observables, and original and mutated unitary circuits acting on that same register.

\subsection{Program Semantics and Permitted Observations}
\label{sec:semantic-contract}

Whether a mutant is equivalent to the original depends on the allowed inputs and on what counts as observable behavior.
Write $U$ and $U_m$ for the original and mutant circuit unitaries, and $\mathcal D_{\mathrm{in}}$ for the declared set of allowed input states.
Under \emph{state-transforming semantics}, equivalence means that the two circuits produce the same output state for every allowed input:
\begin{equation}
\label{eq:domain-state-equivalence}
U\rho U^\dagger=U_m\rho U_m^\dagger\qquad\text{for every }\rho\in\mathcal D_{\mathrm{in}}.
\end{equation}
When the domain contains every density operator, this is channel equivalence and, for the unitary circuits considered here, is equivalent to $U_m=e^{i\phi}U$.
On a restricted input domain, equal output states do not require globally equivalent channels.

As with a classical API, equivalence depends on what callers are allowed to observe.
Let $\mathcal F_{\mathrm{int}}$ denote the set of output measurements permitted by the program's interface.
For example, an interface that returns only a bit measured in the computational basis does not let callers choose an $X$-basis measurement instead.
Two programs are \emph{observationally equivalent at that interface} if, for every allowed input in $\mathcal D_{\mathrm{in}}$ and every permitted measurement in $\mathcal F_{\mathrm{int}}$, they assign the same probability to each possible outcome.
Equality of the complete output distributions allows individual runs to return different results.
For separately prepared invocations, this equality also persists under independent repetition.
Coherent control of the program, in which a quantum control selects whether it runs while preserving superpositions, or entangled inputs across multiple invocations require a separately declared interface contract.
Equal output states guarantee this agreement for any measurement.
If all measurements are permitted, agreement of their outcome distributions also guarantees equal output states; with restricted measurements, different states can produce identical permitted observations.
The introductory sign-reversal example illustrates the latter case: on its fixed input, the original and mutant have different output states but the same computational-basis bit probabilities.

Comparing only an average can hide a further difference.
For a measurement with possible outcomes $-1$, $0$, and $+1$, always returning $0$ and returning $-1$ or $+1$ with equal probability both give mean zero, despite having different distributions.
For an observable with only two distinct outcomes, however, the mean determines both outcome probabilities.
Matching output distributions is also different from satisfying the same acceptance condition.
For example, a contract requiring bit 1 with probability at least $0.9$ accepts both $0.92$ and $0.98$, although those probabilities describe different behavior.

The mathematical analysis below counts a mutant as non-equivalent when it changes the output state for at least one allowed input.
This state-transforming semantics is appropriate for a quantum subroutine whose output qubits feed later quantum operations, because those operations may expose a state difference.
For a top-level program whose interface exposes only classical bits, compare the declared bit-string distributions instead.
Under that interface, the introductory sign-reversal mutant is equivalent on its fixed input and should not be counted as a fault that the interface's tests ought to detect.

The interface defines which inputs and observations are allowed; the current test suite chooses particular cases and assertions within those limits.
A study using interface-relative equivalence must therefore fix the input domain and permitted measurements independently of the suite being evaluated.
It must explain which mutants it excludes as equivalent and report the resulting mutant population separately from the population defined by state differences.
Otherwise, a study could exclude a mutant merely because the current tests cannot observe it, concealing a testing weakness.
Adding a test within the permitted interface can address such a weakness.
Allowing a previously forbidden measurement, such as an $X$-basis measurement at an interface that exposes only computational-basis bits, extends the evaluation contract and evaluates behavior beyond the original interface.

\subsection{Mutation Signal and Measurement Cost}

For a fixed input and measurement, the mutation signal is the change in the expected result.
Repeated measurements estimate this change, with uncertainty that depends on the number of shots and the range of possible outcomes.
Write the ideal circuit around mutation site $j$ as
$$
U
=
U_{>j}G_jU_{<j},
$$
where $U_{<j}$ contains the gates before $G_j$, $U_{>j}$ contains the gates after it, and the rightmost operation acts first.
For input state $\rho$ and observable $O$, define the original and mutant expectations by
$$
\mu(\rho,O)
=
\Tr\!\left(OU\rho U^\dagger\right),
\qquad
\mu_m(\rho,O)
=
\Tr\!\left(OU_m\rho U_m^\dagger\right).
$$
The population-level signal available to an expectation-value test is
\begin{equation}
\label{eq:core-observable-gap}
\Delta_m(\rho,O)
=
|\mu(\rho,O)-\mu_m(\rho,O)|.
\end{equation}
The gap $\Delta_m$ measures the change for this particular input and observable; equivalence requires agreement across the declared input and observation families.
With more than two possible measurement outcomes, equal expectations can still hide different outcome distributions; the distribution assertions below retain their full decision rule.
The spectrum of a Hermitian observable $O$ is the set of its eigenvalues, which are the possible outcomes of an ideal direct measurement of $O$.
Its spectral range
$$
R_O
=
\lambda_{\max}(O)-\lambda_{\min}(O)
\geq0
$$
is therefore the width from the smallest to the largest possible measurement outcome.
For example, the Pauli observable $Z$ has eigenvalues $-1$ and $+1$, so $R_Z=2$, whereas the projector $\ket{1}\bra{1}$ has eigenvalues $0$ and $1$, so its spectral range is one.
We require $R_O>0$: if $R_O=0$, then $O$ is proportional to the identity and no unitary mutant can change its expectation.
Let $C(\rho,O,m)>0$ denote the cost of one paired original--mutant comparison: one independent shot from each circuit, including preparation and measurement.
Hoeffding's bound gives a sufficient sampling cost for estimating a nonzero gap to a fixed fraction of its size; \cref{sec:estimation-efficiency} derives the bound and an optional signal-to-cost ranking score.
Deciding whether a gap crosses a detection threshold depends on its distance $|\Delta_m-\gamma|$ from that threshold.
Execution planning must also account for preprocessing, pilot measurements, and the cost of implementing each measurement.

\subsection{Test Collections, Detection, and Adequacy}
\label{sec:test-oracle}

A differential test counts a change as detection when the expectation gap reaches a threshold chosen before inspecting the mutant's effect.
Write a suite of these differential expectation-value tests as $\mathcal T=\{t_r=(\rho_r,O_r,\gamma_r)\}_{r=1}^{N_{\mathrm{test}}}$.
Each test specifies an allowed input $\rho_r\in\mathcal D_{\mathrm{in}}$, an observable $O_r$, and a positive threshold $\gamma_r$.
The mutant passes this population-level rule exactly when its expectation differs from the original by strictly less than $\gamma_r$.
The threshold specifies the smallest mean difference counted as detection; the probability of making a wrong statistical decision is a separate quantity.
Using the original as a reference does not establish that the original program is correct.

An existing test may judge outputs by a different rule, such as requiring an output probability to lie in a specified interval.
We call that supplied pass/fail condition a \emph{native assertion}.
If its acceptance region is $\mathcal A_r$, an exact native assertion kill requires the original value to lie inside $\mathcal A_r$ and the mutant value to lie outside it.
For example, expectations $0$ and $0.15$ both satisfy $\mathcal A_r=[-0.2,0.2]$, although their differential gap exceeds $\gamma_r=0.1$.
Conversely, expectations $0.19$ and $0.21$ cross that acceptance boundary with a gap below $0.1$.
The triple $(\rho_r,O_r,\gamma_r)$ represents a native assertion only when the two acceptance rules agree.
Otherwise, retain the native rule, including how it treats boundary values and the requirement that the original passes.
The gap formulas can still certify some native outcomes, as follows, but a native acceptance rule must not be replaced by a differential threshold.

\begin{proposition}[Certificates for native expectation assertions]
\label[proposition]{prop:native-assertion-certificates}
Fix a native acceptance region $\mathcal A_r$ and establish that the original expectation $\mu_r=\mu(\rho_r,O_r)$ lies in it.
If $\mathcal A_r=(\ell_r,u_r)$, define the positive acceptance margin
\begin{equation}
\label{eq:native-acceptance-margin}
h_r=\min\{\mu_r-\ell_r,u_r-\mu_r\}.
\end{equation}
Any valid upper gap bound $\overline\Delta_{m,r}$ gives the sufficient survival certificate
\begin{equation}
\label{eq:native-margin-skip}
\Delta_m(\rho_r,O_r)\leq\overline\Delta_{m,r}<h_r
\quad\Longrightarrow\quad
\mu_m(\rho_r,O_r)\in\mathcal A_r.
\end{equation}
More generally, suppose the original expectation and signed shift $d_{m,r}=\mu_r-\mu_m(\rho_r,O_r)$ have jointly valid bounds
\begin{equation*}
\mu_r\in[\mu_r^-,\mu_r^+],
\qquad d_{m,r}\in[d_{m,r}^-,d_{m,r}^+].
\end{equation*}
Then the mutant expectation lies in
\begin{equation}
\label{eq:native-mutant-interval}
J_{m,r}=[\mu_r^- - d_{m,r}^+,\mu_r^+ - d_{m,r}^-].
\end{equation}
The inclusion $J_{m,r}\subseteq\mathcal A_r$ certifies survival, and $J_{m,r}\cap\mathcal A_r=\varnothing$ certifies a native kill; otherwise these bounds leave the outcome unresolved.
These set tests use the native assertion's own open or closed boundary conventions.
\end{proposition}
\begin{proof}
A change smaller than $h_r$ cannot reach either endpoint of the open acceptance interval.
The general enclosure follows from $\mu_m=\mu_r-d_{m,r}$ by interval subtraction, after which inclusion or disjointness settles the native predicate.
\end{proof}

Reaching the nearest boundary distance does not establish a kill: for $\mathcal A_r=(0,1)$, the original mean $0.2$ and mutant mean $0.7$ both pass, although their gap $0.5$ exceeds $h_r=0.2$.
If only an upper absolute-gap bound is available, use $d_{m,r}\in[-\overline\Delta_{m,r},\overline\Delta_{m,r}]$ in \cref{eq:native-mutant-interval}; a signed enclosure can give a sharper decision.
Statistical bounds must cover the original-pass condition and the mutant enclosure simultaneously, with the error allocation and data-dependent execution safeguards used below.
In particular, uncertainty in the original mean widens the mutant enclosure; it does not justify increasing the acceptance margin.
These certificates concern the exact native predicate and remain distinct from the realized outcome of a sampled assertion, whose significance, sample size, and acceptance procedure must be retained.

With finitely many shots, a test estimates the expectations from samples and needs a statistical rule for deciding whether the gap reaches the threshold.
The fixed-shot rule in \cref{sec:finite-shot-threshold} distinguishes detected, below-threshold, and unresolved outcomes.
It allocates error probabilities across the declared comparisons to control the chance of any wrong classification in that collection.
Before execution, proofs, deterministic bounds, or intervals estimated from pilot measurements can help decide which pairs to run.
An observed test outcome instead comes from the original and mutant samples under the prespecified statistical rule.
The planning label \emph{execute--predicted above threshold} therefore records a prediction, with an observed kill requiring the sampled test outcome described above.

The suite detects a mutant if at least one of its tests detects it.
For each mutant $m$, collect the indices of those tests in
\begin{equation}
\label{eq:suite-detection-set}
K_{\mathcal T}(m)
=
\{r:\Delta_m(\rho_r,O_r)\geq\gamma_r\}.
\end{equation}
An empty set $K_{\mathcal T}(m)=\varnothing$ means that the mutant survives every test in the suite.
For a prespecified nonempty population $\mathcal M^\star$ of non-equivalent mutants, the population-level mutation score is the fraction detected by the suite:
\begin{equation}
\label{eq:suite-mutation-score}
\mathrm{MS}(\mathcal T;\mathcal M^\star)
=
\frac{1}{|\mathcal M^\star|}
\sum_{m\in\mathcal M^\star}
\bm{1}\!\left\{K_{\mathcal T}(m)\neq\varnothing\right\}.
\end{equation}
Each mutant counts once if at least one test detects it, regardless of how many other tests also detect it.
For example, detecting eight of ten mutants gives a score of $0.8$.
A score of one means that every mutant in this population is detected by at least one test at its declared threshold.
Individual tests may still miss mutants in that population, and faults outside it may still escape the suite.

To use this score to evaluate a suite's adequacy, fix the mutant population $\mathcal M^\star$ independently of the suite $\mathcal T$, including any sampling strata and exclusions.
Selecting mutants because they have large gaps against the same suite would favor faults that the suite already detects and could inflate its apparent adequacy.
Non-equivalent mutants that no test detects must therefore remain in the population and count as survivors.
Exact analysis may identify these survivors without executing them, but they still belong in the mutation-score denominator.

Choose each detection threshold $\gamma_r$ to represent the smallest mean change meaningful to the declared program contract, before inspecting mutant effects.
The threshold is expressed in the observable's units, so its probability interpretation depends on the observable.
For a $\pm1$ observable, $\gamma=0.1$ corresponds to a change of $0.05$ in the probability of either outcome.
For a projector, the same threshold corresponds to a probability change of $0.1$.
These examples explain the units; the threshold for a meaningful fault remains specific to the declared program contract.

The shot budget determines how precisely that mean change can be estimated.
For a fixed mutant--test pair, \cref{sec:finite-shot-threshold} gives the gap-estimation radius $r_m=R_O\sqrt{\log(2/\alpha_{\mathrm{test},m})/n_{\mathrm{shots}}}$.
Here $\alpha_{\mathrm{test},m}$ is the failure probability allocated to the comparison, and $n_{\mathrm{shots}}$ is the number of independent shots taken from each circuit.
To make this radius at most a chosen error tolerance $\epsilon$, allocate $n_{\mathrm{shots}}\geq R_O^2\log(2/\alpha_{\mathrm{test},m})/\epsilon^2$ shots per circuit.
Even with this precision, a gap close to $\gamma_r$ may remain unresolved because its confidence interval crosses the detection threshold.
If the budget is too small to resolve such a case, retain the unresolved outcome and the declared threshold for a meaningful fault.

These thresholds and uncertainty bounds help determine which mutant--test pairs to run.
The planner below handles one pair at a time; optimizing a schedule for the whole suite, sharing original-circuit samples, and reusing measurement bases are outside its scope.
Ranking mutants by large gaps can support smoke testing and fault localization.
\Cref{tab:decision-notation} collects the notation used in these decisions.

We use $\mathcal{O}$ for computational complexity and $O$ for observables.

\begin{table}[!htb]
\caption{Principal notation; differently indexed quantities have distinct roles.}
\label{tab:decision-notation}
\centering
\small
\begin{tabularx}{\linewidth}{@{}l>{\raggedright\arraybackslash}X@{}}
\toprule
\textbf{Symbol} & \textbf{Meaning} \\
\midrule
$\mathcal M$, $M=|\mathcal M|$ & Candidate set and its count; $\mathcal M^\star$ is the declared non-equivalent population. \\
$N_{\mathrm{test}}$, $N_{\mathrm{axis}}$ & Number of tests and of fixed rotation axes. \\
$\tau_j$, $O_j$, $W_m$ & Mutation-site state, backward-propagated observable, and relative unitary. \\
$\Delta_m$, $\delta_m=\Delta_m/R_O$ & Absolute and range-normalized expectation gaps. \\
$R_O$, $C_m$ & Outcome range and paired-comparison cost. \\
$S_m^{\mathrm{state}}$, $S_m^{\mathrm{op}}$ & State-dependent first-order sensitivity and operator-level proxy. \\
$K_{\mathcal T}(m)$ & Set of test indices detecting $m$. \\
$K_{\mathrm{ret}}$, $K_{\mathrm{bucket}}$ & Retained Pauli-term cap and representative cap per bucket. \\
$n_{\mathrm{shots}}$, $N_{\mathrm{copy}}$ & Shots per known circuit and copies in the separate discrimination problem. \\
$F$ & Squared fidelity between quantum states. \\
$\delta_\theta$ & Offset of an Euler angle. \\
$\gamma_r$ & Positive population threshold for test $r$; detection includes equality. \\
$\mathcal A_r$, $h_r$ & Native acceptance region and the original mean's margin inside an open interval. \\
$\delta_{\mathrm{disc}}$ & Error target for equal-prior discrimination of an unknown original or mutant output source. \\
$\alpha_{\mathrm{pilot}}$ & Total failure budget for simultaneous pilot gap intervals. \\
$\alpha_{\mathrm{test}}$ & Total failure budget for simultaneous observed-test gap intervals. \\
\bottomrule
\end{tabularx}
\end{table}

\subsection{From Analysis to Execution Decisions}
\label{sec:workflow-overview}

Analysis can sometimes settle a test's population-level outcome before the original and mutant circuits are executed.
Given the original circuit, candidate mutations, allowed inputs and observations, and a fixed test suite, an execution planner decides which pairs to skip or run and records the evidence for each decision.
For an adequacy evaluation, the mutant population remains the one fixed independently of the suite.
Pairs selected for execution also need a sampling budget and a statistical decision rule, such as the one in \cref{sec:finite-shot-threshold}.

First, try to prove that the mutant and original are equivalent for every allowed input under the semantics declared in \cref{sec:semantic-contract}.
A \emph{semantic-equivalence certificate} is such a proof: it establishes equal output states under state-transforming semantics, or equal complete outcome distributions for every permitted measurement under interface-relative semantics.
For example, applying the same deterministic compilation process to both programs can prove equivalence if it preserves program behavior and produces the same complete circuit representation.
\Cref{app:transpiler-equivalence,app:circuit-equivalence} give the assumptions needed for this check and describe other equivalence checks.
Without a proof of equivalence, this step cannot exclude the mutant.

For each remaining pair, check whether the mutation has no effect on the tested expectation or an effect smaller than the detection threshold.
For example, a changed operation that cannot influence the tested observable through the remaining gates has zero effect on that expectation.
The backward-light-cone argument in \cref{sec:backward-light-cone-proof} proves this structural rule for the selected observable.
If these checks do not settle the pair, evaluate $\Delta_m(\rho_r,O_r)$ exactly when feasible, or obtain an interval that accounts for approximation and estimation errors.
Failure to obtain usable gap evidence leaves the threshold outcome unresolved.

A \emph{certified zero} is a gap proved to be exactly zero under the stated circuit and test assumptions.
The proof can follow from circuit structure, as in the light-cone example above, or from exact mathematical identities.
An \emph{algebraic zero certificate} is a proof of the latter kind: calculations with exact quantities, performed by hand or by symbolic software, establish that the gap equals zero.
The use of symbolic software alone is insufficient; the proof must apply to the declared quantities with their exact values preserved.
A value that rounds to zero or falls below a numerical tolerance remains only numerically compatible with zero, even when evaluated at high precision.

For the differential rule, a pair can be skipped when the exact gap, or the upper end of its interval, is strictly below $\gamma_r$.
If the interval comes from statistical estimation, the skip retains its stated confidence qualification.
A certified zero is below every positive detection threshold.
If the exact gap, or the lower end of its interval, is at least $\gamma_r$, the analysis predicts detection; execution can then provide observed evidence.
A pair whose outcome remains unresolved is also a candidate for execution.
If a resource limit prevents that execution, record the limit and keep the outcome unresolved.
Native assertions follow their own acceptance rules when deciding whether the available evidence supports a skip or predicts detection.
\Cref{sec:core-workflow} gives the detailed planning procedure and its guarantees, while \cref{tab:classification-dimensions} separates what must be recorded about the program, test, execution decision, and evidence.

\begin{table}[!htb]
\caption{Information recorded separately for program semantics, each mutant--test outcome, execution planning, and supporting evidence.}
\label{tab:classification-dimensions}
\centering
\small
\renewcommand{\arraystretch}{1.12}
\begin{tabularx}{\linewidth}{@{}>{\raggedright\arraybackslash}p{0.20\linewidth}>{\raggedright\arraybackslash}X@{}}
\toprule
\textbf{Dimension} & \textbf{Recorded information and meaning} \\
\midrule
Program semantics & Equivalent, non-equivalent, or unresolved under the declared input domain and permitted observations.
Only a proof of equivalence under the declared semantics justifies excluding the mutant from the corresponding non-equivalent population. \\
Test outcome & For differential test $t_r$, the mutant survives that test if $\Delta_m(\rho_r,O_r)<\gamma_r$ and is detected at the population level if $\Delta_m(\rho_r,O_r)\geq\gamma_r$.
Record \emph{unresolved} if the available evidence cannot establish either relation.
A proven zero gap establishes survival under this differential rule.
Native assertions use their own pass/fail rule.
Record survival for the full suite only when survival has been established for every test in the suite. \\
Execution decision & Skip or execute the mutant--test pair, with the reason supplied by the planning policy.
Skipping because of a budget limit does not by itself establish whether the mutant passes the test, although separate analysis may settle that outcome. \\
Evidence & Record the support for the claim: algebra, validated bounds, numerical evaluation, pilot intervals, or observed samples, with any confidence qualification.
Analytical detection and observed kills retain their respective provenance; insufficient evidence leaves the claim unresolved. \\
\bottomrule
\end{tabularx}
\end{table}

For example, a mutant with a certified zero gap on one test can be skipped for that test but detected by another.
It is a suite survivor only when every test is below its threshold; if no test establishes detection and some test outcomes remain unresolved, its suite status remains unresolved.
Both survivors and unresolved mutants remain in the declared mutant population.
Zero expectation gap also need not imply equal output distributions, so a certificate for one mean must not be extended to a distribution assertion without further justification.
Several zero-effect conditions may hold for the same pair, so a mutant can have more than one explanation for surviving a test.

\subsection{Computational-Basis Distribution Assertions}
\label{sec:distribution-oracles}

The same expectation-value framework also applies to tests that check a distribution of computational-basis output bit strings.
For each bit string $x$, the projector $P_x=\ket{x}\bra{x}$ represents an indicator that is one when $x$ occurs and zero otherwise, so its expectation is exactly the probability of that outcome.
The analysis can therefore determine how a mutation changes individual outcome probabilities, while the distribution assertion retains its own statistical procedure and pass/fail rule for judging them together.
Skipping the test on the grounds of zero effect requires a proof covering every probability capable of changing that decision.

Checking a distribution also affects which structural shortcuts are available.
When a test checks outcomes on only a subset of qubits, the backward light cone can exclude mutations that cannot influence those outputs through the remaining gates.
For the full distribution over every qubit, however, each outcome projector acts on the entire register, so its structural backward light cone cannot exclude any mutation on that register.
Other exact checks remain available: for example, a mutation that leaves the output state unchanged for a fixed input also leaves its output distribution unchanged.
For an oracle that declares detection when the total-variation distance reaches a positive threshold, \cref{cor:distribution-rotation-ceiling} also gives an angle-only skip rule that does not require propagating the outcome projectors.
That certificate applies to the exact distance; sampled histogram and chi-square assertions retain their own decision procedures.
The GHZ example in \cref{sec:ghz-distribution} illustrates how changing from a single expectation to a full-distribution assertion changes both the detectable mutants and the available filtering opportunities.

\subsection{Analyzing a Mutant--Test Pair}
\label{sec:core-mutant-test-analysis}

For a chosen input and observable, the output mutation effect can be expressed and calculated at the changed gate.
The calculation resembles backward program analysis: express the final observable at the mutation site, then compare its expectation before and after the local change.
This mathematical rewriting leaves the circuit execution and its output measurement unchanged.
Readers primarily interested in execution planning may proceed to \cref{sec:core-workflow} on a first reading and return here for the conditions supporting its decisions.

To place the state, observable, and mutation at the same point in the circuit, define
$$
\tau_j
=
G_jU_{<j}\rho U_{<j}^\dagger G_j^\dagger,
\qquad
O_j
=
U_{>j}^\dagger O U_{>j},
\qquad
W_m
=
\widetilde G_jG_j^\dagger.
$$
Here $\tau_j$ is the state immediately after the original gate $G_j$, including the preceding circuit prefix $U_{<j}$.
The propagated observable $O_j$ expresses the final measurement at that same point by accounting for the remaining gates in the suffix $U_{>j}$.
The transformation $O\mapsto U_{>j}^\dagger O U_{>j}$ is called \emph{unitary conjugation}.
The \emph{relative unitary} $W_m$ isolates the gate change: $G_j^\dagger$ undoes the original gate and $\widetilde G_j$ supplies its replacement.
Thus $W_m$ acts on the post-gate state $\tau_j$, and the left-multiplication convention is $\widetilde G_j=W_mG_j$.
$W_m$ is not generally the replacement gate itself.
At this common point, the original expectation is $\Tr(O_j\tau_j)$ and the mutant expectation is $\Tr(W_m^\dagger O_jW_m\tau_j)$.
The analyses below compare these two expectations.
Two operators \emph{commute} when exchanging their order leaves their product unchanged; their commutator $[A,B]=AB-BA$ is then zero.

Equality for one input and observable does not establish equality for other tests.
\Cref{tab:equivalence-hierarchy} distinguishes the conditions for one fixed test, every input, every observable, and both unrestricted families.
The unrestricted rows quantify over all density operators at the mutation site, since the unitary prefix through $G_j$ maps the unrestricted circuit input domain onto that set.
For a restricted domain or measurement interface, use the equivalence definition in \cref{sec:semantic-contract}; the unrestricted conditions require proofs over their full input and measurement domains.

\begin{table}[!htb]
\caption{Expectation and state equivalence for one unitary local replacement.
Moving across columns enlarges the observable family; moving down rows enlarges the input family.}
\label{tab:equivalence-hierarchy}
\centering
\small
\renewcommand{\arraystretch}{1.25}
\begin{tabularx}{\linewidth}{@{}>{\raggedright\arraybackslash}p{0.15\linewidth}>{\raggedright\arraybackslash}X>{\raggedright\arraybackslash}X@{}}
\toprule
 & \textbf{Fixed observable $O_j$} & \textbf{Every observable} \\
\midrule
\textbf{Fixed input $\tau_j$}
& Expectation-test equivalence: $\Delta_m(\rho,O)=0$.
& Output-state equivalence: $W_m\tau_jW_m^\dagger=\tau_j$, or $[W_m,\tau_j]=0$. \\
\textbf{Every input}
& Observable equivalence: $[W_m,O_j]=0$.
& Unitary equivalence up to global phase: $W_m=e^{i\phi}I$. \\
\bottomrule
\end{tabularx}
\end{table}

The local and state-side propositions below justify these conditions.
The lower-right condition follows because a unitary commuting with every Hermitian operator is a scalar phase.

Circuit connectivity can sometimes prove zero effect without computing either $O_j$ or $\tau_j$; \cref{sec:core-result-structural-filtering} gives this structural filter.
When it does not settle the pair, \cref{sec:core-result-exact-local} gives the exact gap and the conditions under which exposing a mutation requires changing the input or observable.
\Cref{sec:core-result-exact-structured} specializes the gap to common Pauli-axis mutations, fixed one-qubit replacements, and Clifford mutations.
For other small mutations whose exact effect cannot be evaluated tractably, \cref{sec:core-perturbative-fallback} gives a first-order approximation with an explicit error bound.
\Cref{sec:claim-hierarchy} explains what each kind of evidence permits the planner to conclude.

\FloatBarrier

\subsubsection{Sound Structural Filtering}
\label{sec:core-result-structural-filtering}

A mutation cannot change the tested expectation if the remaining gates provide no path from the changed operation to the measured qubits.
To check for such a path, trace backward from the measured output through the remaining gates, as in a conservative backward slice from a software assertion.
The \emph{support} of an operator is the set of qubits on which its action differs from the identity.
Start with $\supp(O)$ and traverse the suffix backward, adding every wire of each gate that touches the current set.
The resulting set, the \emph{structural backward light cone} $\mathrm{BLC}_j(O)$, uses only circuit connectivity and exact gate supports.
It can therefore be computed before propagating the full observable or estimating the mutation-site state.
The exact algebraic support $\supp(O_j)$ can be expensive to compute, but it is contained in this structural set:
$$
\supp(O_j)
\subseteq
\mathrm{BLC}_j(O).
$$
Therefore,
\begin{equation}
\label{eq:core-light-cone-discard}
\supp(W_m)\cap\mathrm{BLC}_j(O)=\varnothing
\quad\Longrightarrow\quad
\Delta_m(\rho,O)=0.
\end{equation}
If the mutation acts entirely outside the cone, the pair can therefore be skipped during fixed-suite execution planning.
The support-containment and zero-effect implications are proved in \cref{sec:backward-light-cone-proof}.
The zero-effect conclusion is exact under the unitary circuit model even when the structural cone contains more wires than the true operator support.
Overlap with the cone only leaves open the possibility of an effect: commutation or state-dependent cancellation can still make the gap zero.
\Cref{fig:backward-light-cone} illustrates the backward structural growth and distinguishes an in-cone mutation from operations on an excluded wire.

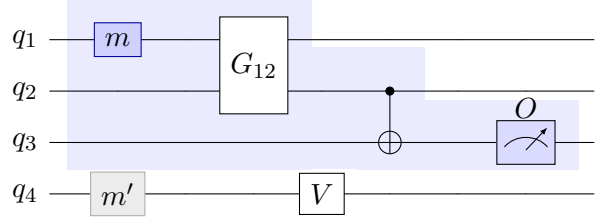
\begin{figure}[!htb]
\centering
\begin{quantikz}[
    row sep={0.68cm,between origins},
    column sep={0.90cm,between origins},
    wire types={q,q,q,q},
    thin lines
]
\lstick{$q_1$} & \gate[style={fill=blue!18,draw=blue!60!black}]{m}\gategroup[3,steps=3,style={draw=none,fill=blue!8,inner xsep=5pt,inner ysep=1pt,yshift=3pt},background]{} & & \gate[2]{G_{12}} & & & & & \\
\lstick{$q_2$} & & & \gategroup[2,steps=3,style={draw=none,fill=blue!8,inner xsep=5pt,inner ysep=1pt,yshift=3pt},background]{} & & \ctrl{1} & & & \\
\lstick{$q_3$} & & & & & \targ{}\gategroup[1,steps=3,style={draw=none,fill=blue!8,inner xsep=5pt,inner ysep=1pt,yshift=3pt},background]{} & & \meter[style={fill=blue!14},label style={yshift=-3pt}]{O} & \\
\lstick{$q_4$} & \gate[style={fill=gray!15,draw=gray!65}]{m'} & & & \gate{V} & & & &
\end{quantikz}
\caption{Four wires $q_1$--$q_4$ carry gates from left to right.
The blue backward cone starts at $O$ on $q_3$, adds $q_2$ at the $q_2$-control/$q_3$-target CNOT and $q_1$ at $G_{12}$, and excludes $q_4$.
Thus, $m$ needs analysis, whereas $m'$ and operations on $q_4$ such as $V$ cannot affect $O$ in the ideal unitary model.}
\label{fig:backward-light-cone}
\Description{A four-wire quantum circuit with a shaded backward cone from a measurement on the third wire through gates connecting the first three wires.
Mutation m lies inside this cone; mutation m-prime and gate V on the fourth wire lie outside it.}
\end{figure}

\subsubsection{Exact Local Mutation Effect}
\label{sec:core-result-exact-local}

When connectivity alone does not settle a pair, the local gap formula determines whether the replacement changes the selected expectation.
It also identifies when changing the input alone or the observable alone cannot expose the mutation.
Evaluating the formula requires the propagated observable $O_j$ and the mutation-site state $\tau_j$, or the relevant expectations in that state.
Obtaining those quantities can require exponential work; \cref{thm:stabilizer-visibility} identifies a particular family with polynomial-time classical analysis.
Evaluating an exact formula in ordinary floating-point arithmetic also differs from certifying its value at a decision boundary.

The first proposition holds the observable fixed and asks whether any input can give a nonzero gap.
The second holds the input fixed and asks whether any observable can do so.
These are applications of trace cyclicity and norm duality to the two choices available when augmenting a test.

\begin{proposition}[Exact local replacement]
\label{thm:core-exact-local}
For any unitary replacement acting on a common subsystem, or with both operations embedded into the union of their supports,
\begin{equation}
\label{eq:core-exact-local}
\Delta_m(\rho,O)
=
\left|
\Tr\!\left[
\left(O_j-W_m^\dagger O_jW_m\right)\tau_j
\right]
\right|.
\end{equation}
The maximum of this gap over all density operators at the mutation site is
$$
\sup_{\tau}
\left|
\Tr\!\left[
\left(O_j-W_m^\dagger O_jW_m\right)\tau
\right]
\right|
=
\left\|O_j-W_m^\dagger O_jW_m\right\|_\infty.
$$
The mutation is invisible to $O_j$ for every state if and only if
$$
W_m^\dagger O_jW_m=O_j,
\qquad
\text{equivalently}
\qquad
[W_m,O_j]=0.
$$
The prefix through $G_j$ is a bijection on density operators, so the maximization also ranges over all circuit inputs when arbitrary inputs are allowed.
\end{proposition}

The proof, including attainment of the operator-norm maximum, is given in \cref{sec:exact-local-replacement}.
Here $\|D\|_\infty$ is the operator norm; for a Hermitian difference $D$, it is the largest absolute eigenvalue and gives the largest possible absolute expectation over input states.
The next result asks the complementary testing question: with the input fixed, how much could changing the measurement expose?
\begin{proposition}[State-side dual and optimal observable]
\label[proposition]{prop:state-side-dual}
For fixed $\tau_j$ and $W_m$, define the traceless Hermitian state difference
\begin{equation}
\label{eq:state-side-difference}
\Delta\tau_j
=
\tau_j-W_m\tau_jW_m^\dagger.
\end{equation}
Then
\begin{equation}
\label{eq:state-side-gap}
\Delta_m(\rho,O)
=
\left|\Tr(O_j\Delta\tau_j)\right|,
\end{equation}
and the maximum over Hermitian propagated observables with $\|O_j\|_\infty\leq1$ is
\begin{equation}
\label{eq:state-side-trace-norm}
\max_{\substack{O_j=O_j^\dagger\\\|O_j\|_\infty\leq1}}
\Delta_m(\rho,O)
=
\|\Delta\tau_j\|_1,
\end{equation}
attained by $O_j^\star=\operatorname{sgn}(\Delta\tau_j)$, with $\operatorname{sgn}(0)=0$ on the zero eigenspace.
The mutation is invisible to every observable for this input if and only if $\Delta\tau_j=0$, equivalently $[W_m,\tau_j]=0$.
\end{proposition}

The specialization of standard trace-norm duality to the mutation site is given in \cref{sec:state-side-dual-proof}.
For the Hermitian difference $\Delta\tau_j$, the trace norm $\|\Delta\tau_j\|_1$ is the sum of its absolute eigenvalues.
The sign observable assigns $+1$, $-1$, or $0$ to the positive, negative, or zero eigenvalues of that difference, respectively.
It supplies a best possible measurement benchmark for the fixed input, whose construction and implementation may still be expensive.
Dividing the gap by the observable's outcome range connects this measurement benchmark to how well the two output states can be distinguished.
Let $\rho_{\mathrm{out}}=U\rho U^\dagger$ and $\rho_{\mathrm{out}}^{(m)}=U_m\rho U_m^\dagger$, and write $D_{\mathrm{tr}}(\sigma,\omega)=\|\sigma-\omega\|_1/2$.
This quantity is the \emph{trace distance} between the states, and it bounds the normalized gap:
\begin{equation}
\label{eq:normalized-gap-trace-distance}
\delta_m(\rho,O)
=
\frac{\Delta_m(\rho,O)}{R_O}
\leq
D_{\mathrm{tr}}\!\left(\rho_{\mathrm{out}},\rho_{\mathrm{out}}^{(m)}\right)
=
\frac{1}{2}\|\Delta\tau_j\|_1.
\end{equation}
For nonzero $\Delta\tau_j$, equality is attained by the Helstrom observable $O_j^\star=\operatorname{sgn}(\Delta\tau_j)$, propagated to the output.
Its minimum and maximum eigenvalues are $-1$ and $+1$, so its spectral range is two; a nontrivial kernel adds the intermediate outcome zero without changing that range.
If $\Delta\tau_j=0$, every nonconstant observable attains the zero normalized gap and trace distance.
If an unknown output is equally likely to come from either circuit, the Holevo--Helstrom theorem gives the best single-shot probability of identifying its source as $(1+D_{\mathrm{tr}})/2$ \cites[Sec.~2.1, especially Eqs.~(2.5)--(2.9)]{helstrom1971quantumdetection}[Theorem~3.4 and Eq.~(3.11)]{watrous2018theory}.
This optimum assumes unrestricted measurements; the permitted interface may prevent a test from attaining it.
The derivation and the resulting observable-independent efficiency ceiling are given in \cref{sec:state-side-dual-proof}.

\Cref{thm:core-exact-local} therefore requires two things for a nonzero gap: the mutation must change the propagated observable, and that change must have nonzero expectation in $\tau_j$.
If $[W_m,O_j]=0$, changing the input cannot expose the mutation through this observable.
If $[W_m,\tau_j]=0$, changing the observable cannot expose it on this input.
If neither commutator vanishes, the selected expectation can still be unchanged, so determining the gap requires evaluating that expectation difference.

\subsubsection{Exact Effects for Common Gate Changes}
\label{sec:core-result-exact-structured}

Common gate changes let us evaluate the local effect without assuming that the change is small.
For a Pauli-axis rotation, two state expectations determine the gap at every angle; fixed one-qubit and Clifford replacements also admit exact constructions below.
These formulas simplify the observable change, although obtaining the required state expectations can still be expensive.

First consider a relative mutation that rotates about an axis $A$ by an angle $\eta$ in radians:
$$
W_m
=
e^{-i\eta A/2},
\qquad
A=A^\dagger,
\qquad
A^2=I.
$$
The condition $A^2=I$ means that $A$ is its own inverse, or an \emph{involution}.
Every Pauli string satisfies it, so the axis may involve several qubits, as in joint rotations $R_{XX}$ and $R_{ZZ}$.
The formulas apply to every Hermitian involution.
The \emph{Bloch sphere} represents pure one-qubit states as points on a unit sphere whose coordinates are their $X$, $Y$, and $Z$ expectations.
On one qubit, these involutions include every Bloch-sphere axis $A=\bm{n}\cdot\bm{\sigma}$ with $\|\bm{n}\|_2=1$, where $\bm{\sigma}=(X,Y,Z)$ and $\bm{n}$ specifies the axis direction.
Decompose the propagated observable into the parts that commute and anticommute with $A$:
\begin{equation*}
\begin{aligned}
O_{j,+}
&=
\frac{O_j+AO_jA}{2},
&
O_{j,-}
&=
\frac{O_j-AO_jA}{2}.
\end{aligned}
\end{equation*}
Anticommutation means $AO_{j,-}=-O_{j,-}A$.
The commuting part is unchanged by the rotation, so only the anticommuting part can contribute to the mutation effect.
Because $A^2=I$, its exponential reduces to sine and cosine terms, giving the conjugation identity below \cite[Sec.~III.A, Theorem~1, and Eq.~(12)]{schuld2019evaluating}.
Parameter-shift methods use this same sinusoidal structure to recover expectation derivatives from evaluations at shifted gate angles, while analytic optimization reconstructs and optimizes the response from three evaluations \cite{ostaszewski2021structure,nakanishi2020sequential}.
The corollary uses the same response to find angles that leave the original expectation unchanged and angles that produce the largest gap.

\begin{proposition}[Exact involutory-axis rotation effect]
\label{thm:core-exact-pauli-rotation}
For every mutation strength $\eta$,
\begin{equation}
\label{eq:core-exact-pauli-conjugation}
e^{i\eta A/2}O_je^{-i\eta A/2}
=
O_{j,+}
+
\cos(\eta)O_{j,-}
-
i\sin(\eta)O_{j,-}A.
\end{equation}
Consequently,
\begin{equation}
\label{eq:core-exact-pauli-gap}
\Delta_m(\rho,O)
=
\left|
\Tr\!\left(
\left[
(1-\cos\eta)O_{j,-}
+
i\sin\eta\,O_{j,-}A
\right]
\tau_j
\right)
\right|.
\end{equation}
\end{proposition}

The proof is given in \cref{sec:exact-structured-mutations}.
The identity holds for any angle, including cases where the linear term near zero vanishes but higher-order terms produce an effect.
For the common same-axis replacement $R_A(\theta)\rightarrow R_A(\theta')$, the relative mutation has $\eta=\theta'-\theta$.
Thus, the same theorem exactly covers rotation-gate removal with $\eta=-\theta$, a sign flip with $\eta=-2\theta$, and fixed large-angle changes such as $\eta=\pi/2$.
The size of the angle alone does not require an approximation or full-circuit simulation.

For a fixed site, axis, and test, all dependence on the input state enters through two expectation values.
They multiply the cosine and sine contributions to the signed change; define them as
\begin{equation}
\label{eq:exact-pauli-two-coefficients}
a_m
=
\Tr(O_{j,-}\tau_j),
\qquad
b_m
=
i\Tr(O_{j,-}A\tau_j).
\end{equation}
Both coefficients are real because $O_{j,-}$ and $iO_{j,-}A$ are Hermitian.

\begin{corollary}[Finite-strength visibility and optimal strength]
\label[corollary]{cor:finite-strength-visibility}
For a fixed-axis involutory relative rotation,
\begin{equation}
\label{eq:finite-strength-gap}
\Delta_m(\eta)
=
\left|(1-\cos\eta)a_m+\sin\eta\,b_m\right|.
\end{equation}
Consequently,
\begin{equation}
\label{eq:finite-strength-maximum}
\max_{\eta\in\mathbb R}\Delta_m(\eta)
=
|a_m|+\sqrt{a_m^2+b_m^2},
\end{equation}
and the mutation is invisible for every strength if and only if $a_m=b_m=0$.
If $(a_m,b_m)\neq(0,0)$, its complete zero set modulo $2\pi$ is
\begin{equation}
\label{eq:finite-strength-zero-set}
\eta
\equiv
0
\quad\text{or}\quad
\eta
\equiv
\eta_{0,m}
:=
2\operatorname{atan2}(-b_m,a_m)
\pmod{2\pi}.
\end{equation}
The two residues are distinct exactly when $b_m\neq0$; in particular, $a_m=0$ and $b_m\neq0$ gives the nontrivial zero $\eta_{0,m}\equiv\pi$.
The phase form, maximizing strengths, restricted-domain optimization, and proof of the zero set are given in \cref{sec:exact-structured-mutations}.
\end{corollary}

Once $a_m$ and $b_m$ have been computed, a testing tool can reuse them for every declared angle change at the same site and axis under the same test.
Their joint vanishing means that no angle changes this expectation; otherwise, the zero set identifies the particular angles at which the effect cancels.
The maximizing angle can help with smoke testing or fault localization when the strength is controllable, but maximizing the gap does not generally model a plausible fault.
For an adequacy study, the mutation operator must fix $\eta$ independently of the suite.
The formulas then diagnose whether that declared strength lies near $\eta_{0,m}$ or whether $a_m=b_m=0$ makes every strength invisible.

Even without the state coefficients, an upper bound can sometimes establish that the gap is below the test threshold.
Write the Hermitian observable difference as
\begin{equation}
\label{eq:structured-observable-difference}
D_{A,j}(\eta)=(1-\cos\eta)O_{j,-}+i\sin\eta\,O_{j,-}A.
\end{equation}
The exact norm identity proved in \cref{sec:bounded-amplification} bounds its expectation for every state:
\begin{equation}
\label{eq:exact-rotation-gap-envelope}
\Delta_m(\eta)
\leq
2|\sin(\eta/2)|\,\|O_{j,-}\|_\infty.
\end{equation}
Some state at the mutation site attains this bound, although the selected test state may give a smaller gap or exact cancellation.
Let $c_O=(\lambda_{\max}(O)+\lambda_{\min}(O))/2$.
Because the anticommuting component is unchanged by the scalar shift $O_j\mapsto O_j-c_OI$, it satisfies $\|O_{j,-}\|_\infty\leq R_O/2$, and \cref{eq:exact-rotation-gap-envelope} immediately yields $\delta_m\leq|\sin(\eta/2)|$.
Thus the angle and outcome range alone can certify a below-threshold gap.
\begin{proposition}[Rotation threshold gate]
\label[proposition]{prop:rotation-threshold-gate}
For $W_m=e^{-i\eta A/2}$ with $A=A^\dagger$ and $A^2=I$, an observable of range $R_O>0$, and a differential threshold $\gamma>0$,
\begin{equation}
\label{eq:rotation-threshold-gate}
R_O|\sin(\eta/2)|<\gamma
\quad\Longrightarrow\quad
\Delta_m<\gamma.
\end{equation}
The check takes $\mathcal{O}(1)$ scalar operations when the angle and outcome range are supplied, with exact or validated arithmetic required to certify the strict inequality.
\end{proposition}
If the upper bound equals $\gamma$, this rule alone leaves the threshold decision unresolved.
The rule requires no execution budget: it bounds the population gap itself.
The independent-copy lower bound below addresses a different question---how many outputs are needed to identify an unknown original or mutant source at a prescribed error probability.
Both scalar rules use the periodic quantity $|\sin(\eta/2)|$, or the canonical effective angle in $[0,\pi]$; retain the original signed offset when evaluating the fixed-axis gap formulas.
For $H_m=A/2$, the operator proxy introduced below satisfies $S_m^{\mathrm{op}}=\|O_{j,-}\|_\infty$.
When a Clifford suffix propagates a unit-norm Pauli observable to a single Pauli string, this proxy is just zero or one according to commutation, so it cannot rank sites beyond that distinction.

To determine what changing the measurement could achieve on the same input, apply the rotation analysis to the state difference instead.
This also supplies the trace-distance bound, an observable attaining it, and the basis for the independent-copy lower bound.

\begin{corollary}[State-side involutory rotation]
\label[corollary]{cor:rotation-state-side}
If $W_m=e^{-i\eta A/2}$ with $A=A^\dagger$ and $A^2=I$, define
\begin{equation*}
\tau_{j,-}^{(A)}
=
\frac{\tau_j-A\tau_jA}{2}.
\end{equation*}
Then
\begin{equation}
\label{eq:state-side-pauli-difference}
\Delta\tau_j
=
(1-\cos\eta)\tau_{j,-}^{(A)}
-
i\sin\eta\,\tau_{j,-}^{(A)}A.
\end{equation}
The opposite sine sign relative to the observable-side expression in \cref{eq:core-exact-pauli-gap} reflects the two conjugations: the state by $W_m$ and the observable by $W_m^\dagger$.
Its trace norm is exactly
\begin{equation}
\label{eq:state-side-pauli-trace-norm}
\|\Delta\tau_j\|_1
=
2|\sin(\eta/2)|\,\|\tau_{j,-}^{(A)}\|_1.
\end{equation}
Apart from the unchanged strengths $\eta\equiv0\pmod{2\pi}$, the rotation is therefore invisible to every observable for this input if and only if $\tau_{j,-}^{(A)}=0$, equivalently $[A,\tau_j]=0$.
\end{corollary}
The derivation is given in \cref{sec:state-side-dual-proof}.
The two zero conditions answer complementary testing questions.
If $O_{j,-}=0$, the chosen observable has zero gap for every input and strength.
If $\tau_{j,-}^{(A)}=0$, the chosen input has zero gap for every observable and strength.

\begin{corollary}[Test-independent involutory-rotation ceiling]
\label[corollary]{cor:test-independent-rotation-ceiling}
Let $A=A^\dagger$ satisfy $A^2=I$ and have nonzero $+1$ and $-1$ eigenspaces.
For Pauli strings this excludes only $A=\pm I$.
For the relative rotation $W_m=e^{-i\eta A/2}$, every input and nonconstant observable satisfy
\begin{equation}
\label{eq:test-independent-rotation-ceiling}
D_{\mathrm{tr}}\!\left(\rho_{\mathrm{out}},\rho_{\mathrm{out}}^{(m)}\right)
\leq
|\sin(\eta/2)|,
\qquad
\delta_m(\rho,O)
\leq
|\sin(\eta/2)|.
\end{equation}
Both ceilings are tight.
For the discrimination problem, consider $N_{\mathrm{copy}}$ independent executions with inputs fixed before any outcomes are observed.
The inputs may differ across executions, and each may be entangled with its own untouched reference system, but the joint input--reference states must be products across executions.
Even with an arbitrary joint measurement on all output--reference copies, attaining any fixed nonzero advantage over random guessing under equal priors requires $N_{\mathrm{copy}}=\Omega(1/\sin^2(\eta/2))$ as $\eta\to0$.
This lower bound holds for every such predetermined product-input strategy.
\end{corollary}

The tightness construction, fidelity proof, and boundary with entangled multi-use channel discrimination are given in \cref{sec:state-side-dual-proof}.
The tightness construction permits arbitrary input-state preparation; a tool restricted to computational-basis or simple product inputs can realize a strictly smaller ceiling.
For example, $R_z(\eta)$ acting directly on either computational-basis state changes no output state, so the attainable trace distance on those inputs is zero even when $|\sin(\eta/2)|>0$.

\begin{corollary}[Distribution rotation ceiling]
\label[corollary]{cor:distribution-rotation-ceiling}
Let $W_m=e^{-i\eta A/2}$ with $A=A^\dagger$ and $A^2=I$, and apply the same fixed output measurement to the original and mutant circuits.
Write its effects as $E_x\geq0$ with $\sum_xE_x=I$, and its outcome probabilities as $p(x)=\Tr(E_x\rho_{\mathrm{out}})$ and $p_m(x)=\Tr(E_x\rho_{\mathrm{out}}^{(m)})$.
Each \emph{measurement effect} $E_x$ is the positive semidefinite operator specifying the probability of outcome $x$; this allows measurements more general than basis projectors.
Then
\begin{equation}
\label{eq:distribution-rotation-ceiling}
D_{\mathrm{TV}}(p,p_m)=\frac12\sum_x|p(x)-p_m(x)|
\leq D_{\mathrm{tr}}(\rho_{\mathrm{out}},\rho_{\mathrm{out}}^{(m)})
\leq|\sin(\eta/2)|.
\end{equation}
Consequently, for a differential distribution test with detection rule $D_{\mathrm{TV}}\geq\gamma_{\mathrm{TV}}>0$,
\begin{equation}
\label{eq:tv-threshold-gate}
|\sin(\eta/2)|<\gamma_{\mathrm{TV}}
\quad\Longrightarrow\quad
D_{\mathrm{TV}}(p,p_m)<\gamma_{\mathrm{TV}}.
\end{equation}
\end{corollary}
\begin{proof}
Let $E_+$ be the sum of effects for outcomes with $p(x)>p_m(x)$.
Since the probabilities each sum to one, $D_{\mathrm{TV}}=\Tr[E_+(\rho_{\mathrm{out}}-\rho_{\mathrm{out}}^{(m)})]$.
The inequality $0\leq E_+\leq I$ bounds this trace by half the trace norm of the traceless state difference, and \cref{cor:test-independent-rotation-ceiling} supplies the rotation ceiling.
For $A=\pm I$, both distances vanish directly.
\end{proof}

This rule covers computational-basis distributions and their marginals without enumerating outcomes.
It does not imply that a sampled histogram or chi-square test will accept: a population distance and a finite-sample rejection are different quantities.
Nor does a small TV bound uniformly control Pearson chi-square divergence when reference probabilities can approach zero.
For $0<\epsilon<1$, the distributions $p=(\epsilon^3,1-\epsilon^3)$ and $q=(\epsilon,1-\epsilon)$ have $D_{\mathrm{TV}}=\epsilon-\epsilon^3\to0$ as $\epsilon\to0$, while $\sum_x(q(x)-p(x))^2/p(x)\sim1/\epsilon$.
More generally, a specified statistic satisfying $|T(\sigma)-T(\omega)|\leq L_TD_{\mathrm{tr}}(\sigma,\omega)$ admits the bound $L_T|\sin(\eta/2)|$ on its change; deciding acceptance still requires that statistic's threshold or acceptance region.
For a native expectation interval, combining \cref{eq:rotation-threshold-gate,eq:native-margin-skip} gives the sufficient condition $R_O|\sin(\eta/2)|<h_r$ for survival, with no converse detection rule.

The same rotation formulas also cover any fixed one-qubit unitary replacement, even when the original and replacement gates rotate about different axes.
The relative unitary determines the effective axis and angle; a global phase has no effect on the output state.
\begin{corollary}[Fixed one-qubit unitary replacement]
\label[corollary]{cor:fixed-one-qubit-replacement}
Let $G_j$ and $\widetilde G_j$ be one-qubit unitaries and let $W_m=\widetilde G_jG_j^\dagger$.
If $W_m=e^{i\chi}I$, set $\eta_{\mathrm{eff}}=0$ and leave $A_{\mathrm{eff}}$ undefined; the replacement is channel-equivalent to the original operation.
Otherwise, there exist a global phase $\chi$, an effective angle $\eta_{\mathrm{eff}}\in(0,\pi]$, and a unit vector $\bm{n}\in\mathbb R^3$ such that
\begin{equation}
\label{eq:fixed-one-qubit-axis-angle}
W_m
=
e^{i\chi}e^{-i\eta_{\mathrm{eff}}A_{\mathrm{eff}}/2},
\qquad
A_{\mathrm{eff}}
=
n_xX+n_yY+n_zZ,
\qquad
A_{\mathrm{eff}}^2=I.
\end{equation}
The global phase cancels from conjugation, so \cref{eq:core-exact-pauli-conjugation,eq:core-exact-pauli-gap,eq:exact-rotation-gap-envelope} apply exactly to this fixed replacement with $A=A_{\mathrm{eff}}$ and $\eta=\eta_{\mathrm{eff}}$.
If the replacements form a declared fixed-axis family, \cref{cor:finite-strength-visibility} also gives that family's complete zero set and maximum.
An arbitrary multiparameter path need not retain one effective axis, so each fixed replacement on such a path must instead be reduced separately.
\end{corollary}

The constructive axis--angle reduction is given in \cref{sec:exact-structured-mutations}.

Clifford replacements allow a different exact construction: update each Pauli term of the propagated observable separately.
A Clifford operation maps each Pauli string to another Pauli string, possibly with a sign change, under conjugation.
This property lets the tool transform a list of Pauli terms without creating a larger list at that step.
If $O_j=\sum_P c_PP$ and the relative mutation $W_m$ is Clifford, the Clifford normalizer property \cite[Sec.~3 and Table~1]{gottesman1999heisenberg} implies that for each Pauli string there are a Pauli string $P'_P$ and a sign $s_P\in\{-1,1\}$ such that
$$
W_m^\dagger P W_m
=
s_PP'_P.
$$
Therefore,
$$
O_j-W_m^\dagger O_jW_m
=
\sum_P c_P(P-s_PP'_P)
$$
can be constructed exactly term by term without a perturbative expansion.
If $W_m$ is a Pauli up to global phase, this expression reduces to twice the part of $O_j$ that anticommutes with that Pauli.
Evaluating the realized gap from \cref{eq:core-exact-pauli-gap} or the Clifford replacement identity still requires the corresponding expectations in $\tau_j$.
Exact formulas are also available when the mutation generator can be decomposed into its eigenvalues and eigenspace projectors, including every Hermitian generator with two distinct eigenvalues.
Computing those projectors and the required state expectations may still be expensive.
The Clifford derivation and spectral formula are given in \cref{sec:exact-structured-mutations}.

\subsubsection{Bounded Approximation for Other Small Mutations}
\label{sec:core-perturbative-fallback}

Small non-Pauli gate replacements ordinarily have a tractable local spectral decomposition and should use the exact branch in \cref{sec:core-result-exact-structured}.
The fallback below extends the analysis to more general subcircuit changes when exact evaluation is unavailable but the required commutators, state expectations, and remainder bound remain tractable.
Many-qubit generators with more than two eigenvalues are candidate cases; no scalable implementation or measured benefit for this fallback is established here.

The generator $H_m$ is the Hermitian matrix specifying the direction of the gate change, while the real parameter $\eta$ controls its strength.
The first-order approximation keeps the term proportional to $|\eta|$ in the gap and bounds the error from the remaining terms.
For a left-multiplied perturbation,
$$
\widetilde G_j
=
e^{-i\eta H_m}G_j,
$$
define the sensitivity of the selected input and observable by
\begin{equation}
\label{eq:core-state-sensitivity}
S_m^{\mathrm{state}}(\rho,O)
=
\left|
i\Tr\!\left([H_m,O_j]\tau_j\right)
\right|
\end{equation}
and the operator-level proxy, which does not use the selected input, by
\begin{equation}
\label{eq:core-operator-sensitivity}
S_m^{\mathrm{op}}(O)
=
\|[H_m,O_j]\|_\infty.
\end{equation}
The state-dependent quantity $S_m^{\mathrm{state}}$ is the actual first-order coefficient for the selected test; $S_m^{\mathrm{op}}$ bounds that coefficient over all states.
At the declared strength, the finite-change gap $\Delta_m$ may also include higher-order contributions.
The expectation derivative underlying the fallback is the standard parameterized-circuit derivative \cite[Sec.~III, especially Eqs.~(3), (5), and (6)]{schuld2019evaluating}.
The quantitative remainder bound is
\begin{equation}
\label{eq:core-perturbative-bound}
\left|
\Delta_m(\rho,O)
-
|\eta|S_m^{\mathrm{state}}(\rho,O)
\right|
\leq
\frac{\eta^2}{2}
\|[H_m,[H_m,O_j]]\|_\infty.
\end{equation}
The Taylor derivation and the error bound, valid at finite strength, are given in \cref{app:local-perturbation-derivation}.
Let $s_m=|\eta|S_m^{\mathrm{state}}$ be the first-order gap estimate and $B_m=\eta^2\|[H_m,[H_m,O_j]]\|_\infty/2$ its error bound.
The exact gap must then lie in
\begin{equation}
\label{eq:perturbative-gap-interval}
\Delta_m\in I_m^{\mathrm{det}}=[\max(0,s_m-B_m),\;\min(R_O,s_m+B_m)].
\end{equation}
The clipping uses the universal bound $0\leq\Delta_m\leq R_O$.
If the upper endpoint is strictly below $\gamma$, the gap is certified below threshold.
If the lower endpoint is at least $\gamma$, the analysis establishes population-level detection; otherwise the decision remains unresolved.
This interval gives a deterministic bound on approximation error, with its guarantee independent of statistical confidence levels.
Its guarantee assumes exact values of the quantities used to compute it or validated bounds on their numerical errors.
Ordinary floating-point evaluation must retain its numerical provenance.

The interval can settle a threshold decision even when the first-order estimate alone is unreliable.
Use $S_m^{\mathrm{op}}$ only to bound the worst-case first-order response, and use $|\eta|S_m^{\mathrm{state}}$ for ranking only when the first-order term is nonzero and dominates the remainder.
If $S_m^{\mathrm{state}}=0$, the result supplies only an $\mathcal{O}(\eta^2)$ upper bound, and a second-order calculation is needed to predict a nonzero leading effect.
If $[H_m,O_j]=0$, the effect is exactly zero for every $\eta$.

\subsubsection{Evidence Supporting the Decisions}
\label{sec:claim-hierarchy}

The same execution decision can rest on a proof, a bounded approximation, or a statistical estimate, each supporting a different strength of claim.
The framework therefore records the execution-decision status separately from the evidence type.
For fixed-suite planning, the statuses in \cref{alg:efficient-mutant-selection} distinguish semantic-equivalence, certified-zero, and below-threshold skips (plus an optional discrimination-policy skip) from predicted-above-threshold and threshold-unresolved execution decisions.
The six evidence categories below record what was actually obtained.

The \emph{exact} category means mathematically exact under the stated assumptions; computing such a result may still be expensive.
The local identities and commutation conditions in \cref{thm:core-exact-local,prop:state-side-dual}, the rotation ceiling in \cref{cor:test-independent-rotation-ceiling}, and the structural rule in \cref{eq:core-light-cone-discard} require neither a small-mutation approximation nor pilot measurements.
In particular, the commutation and structural rules can certify that a mutant is invisible under their stated scope.
An exact record must also state how the result was obtained: for example, equality after transpilation with pinned settings, a strength bound, a proof of zero gap using circuit structure or exact identities, a structured formula, or simulation.
The distinction between an exact formula and its numerical evaluation still applies; an ordinary floating-point result is not by itself a validated certificate.

The structured analysis in \cref{sec:core-result-exact-structured} exactly describes how Pauli rotations, fixed one-qubit replacements, Clifford replacements, and other spectrally tractable mutations change the propagated observable.
When the corresponding state expectations in $\tau_j$ are unavailable, this result receives the \emph{exact observable-side update only} evidence type and does not by itself determine the realized gap or an above- or below-threshold decision.

The \emph{state-dependent perturbative} category uses the first-order coefficient $S_m^{\mathrm{state}}$ for the selected test when exact evaluation is computationally unavailable, as in \cref{sec:core-perturbative-fallback}.
The remainder bound in \cref{eq:core-perturbative-bound} supplies the deterministic interval in \cref{eq:perturbative-gap-interval}, which can decide a threshold even when a first-order ranking is unreliable.

When the propagated state cannot be computed classically, the pilot measurements in \cref{sec:scalable} provide \emph{pilot-estimated} gaps with simultaneous error bars that include the declared propagation and sampling errors.
Finite pilot data can show that an effect is compatible with a declared visibility tolerance, but they cannot prove that the exact effect is zero.

The \emph{operator proxy} evidence type records $S_m^{\mathrm{op}}$ only as a state-independent worst-case bound.
A large value permits a large first-order response for some state; it establishes neither that response for the selected test nor a large finite-strength gap.
The \emph{unavailable} evidence type records that none of the supported gap, interval, perturbative, or proxy quantities was obtained, leaving the size of the mutation effect unresolved.

\section{Circuit-Aware Execution Planning}
\label{sec:core-workflow}

The clearest opportunities to avoid executions are reusable equivalence proofs, the constant-cost rotation ceiling, structural exclusion, commutation when the propagated observable is cheap, and the exact stabilizer family.
These checks can settle outcomes without constructing a general mutation-site state.
The rule-by-rule scalability conditions are collected in \cref{sec:rule-scalability}.
Their value depends on the oracle and on whether analysis costs less than the execution it replaces; a full-distribution oracle, for example, loses the structural shortcut but retains the TV rotation ceiling.
General dense state evaluation can cost more than simulating the mutant directly, and pilot estimation offers a further, as yet unvalidated, opportunity when shared measurements can amortize its cost.

An execution planner records a skip or execute decision, its evidence, and any unresolved outcome for every declared mutant--test pair.
\Cref{fig:planning-workflow} shows the flow, and \cref{tab:planning-cost-order} gives the reference costs.
The scalar pseudocode in \cref{alg:efficient-mutant-selection} uses the differential expectation rule; \cref{tab:oracle-planning} specifies which certificates transfer to other oracles.
For simultaneous Euler-angle changes, each fixed candidate is reduced separately unless the entire family is known to retain one effective axis.
Euler angles parameterize a one-qubit gate by three ordered rotations about specified axes, up to global phase.
The filter order is a reproducible baseline, not a cost-optimal schedule: transpiler checks are reusable across tests but can cost more than the scalar and structural checks.
Any reordered planner must record its order and preserve the distinction between semantic exclusions and test-relative skips.

For small circuits, the relevant classical baseline is direct statevector evaluation of the same mutant outcomes, including any reusable original-circuit calculation.
A fair comparison charges propagation, equivalence checking, state evaluation, and memory to the planner; it does not compare that work only with the number of shots avoided.
When exact analysis settles an outcome, further execution is needed only if observed evidence is required.
Otherwise, the execute decision schedules known original and mutant circuits under the declared sampling rule, whether the samples come from an ideal simulator or a suitably specified execution model.
Hardware costs and noise require the additional model described in \cref{sec:core-validation}.

\begin{figure}[!htbp]
\centering
\begin{tikzpicture}[font=\small,box/.style={draw,rounded corners,align=center,text width=0.57\linewidth,inner sep=5pt},record/.style={draw,align=center,text width=0.23\linewidth,inner sep=5pt}]
\node[box] (inputs) at (0,0) {Declare semantics, mutant population, oracle, thresholds, and cost model};
\node[box] (equivalence) at (0,-1.25) {Check semantic equivalence};
\node[box] (cheap) at (0,-2.5) {Apply oracle-specific angle, structural, and commutation certificates};
\node[box] (effect) at (0,-3.9) {Evaluate the remaining effects or obtain valid bounds};
\node[box] (execute) at (0,-5.3) {Execute when observed evidence is needed; retain unresolved outcomes if evidence is insufficient};
\node[record] (excluded) at (0.47\linewidth,-1.25) {Record certified semantic exclusions};
\node[record,minimum height=2.8cm] (account) at (0.47\linewidth,-4.3) {Record detected, survivor, and unresolved weights; report score bounds and evidence provenance};
\draw[-latex] (inputs) -- (equivalence);
\draw[-latex] (equivalence) -- (cheap);
\draw[-latex] (cheap) -- (effect);
\draw[-latex] (effect) -- (execute);
\draw[-latex] (equivalence) -- (excluded);
\draw[-latex] (cheap.east) -| (account.north);
\draw[-latex] (effect.east) -- (account.west |- effect.east);
\draw[-latex] (execute.east) -- (account.west |- execute.east);
\end{tikzpicture}
\caption{Execution planning after the contract and population are fixed.
Pairs settled by a certificate can bypass later analysis or execution and enter accounting with their evidence.
Only certified semantic equivalence supports exclusion from the corresponding non-equivalent population; test-relative survivors remain in its denominator.}
\Description{A flowchart from declared inputs through equivalence, cheap certificates, effect evaluation, and execution, with separate records for semantic exclusions and retained detected, survivor, or unresolved outcomes.}
\label{fig:planning-workflow}
\end{figure}
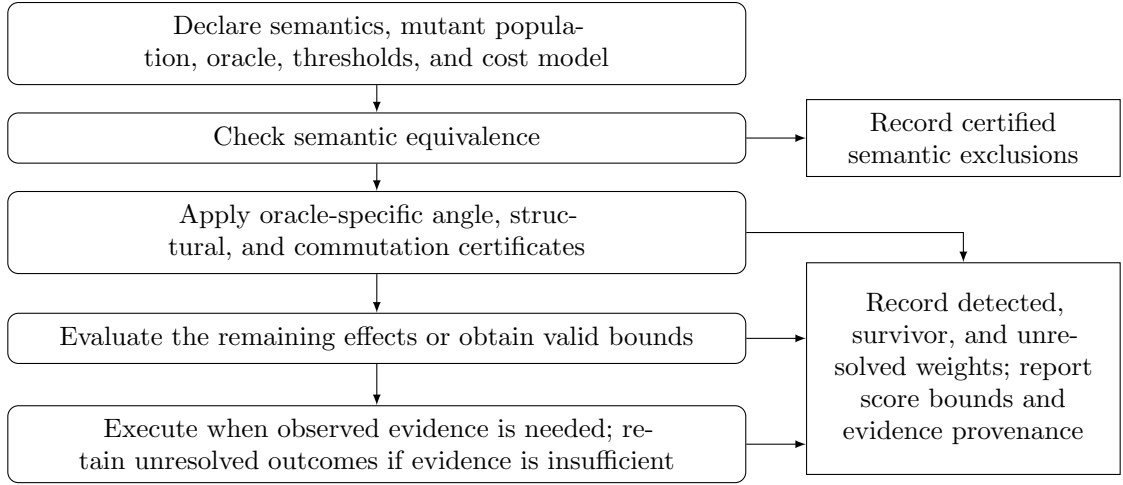

\begin{table}[!htb]
\caption{Reference fixed-test execution-planning stages and their costs.
The decisions shown use the differential expectation rule; \cref{tab:oracle-planning} gives the oracle-specific extensions and \cref{sec:planner-growth} details representation growth.}
\label{tab:planning-cost-order}
\centering
\small
\renewcommand{\arraystretch}{1.12}
\begin{tabularx}{\linewidth}{@{}>{\raggedright\arraybackslash}p{0.20\linewidth}>{\raggedright\arraybackslash}X>{\raggedright\arraybackslash}p{0.34\linewidth}@{}}
\toprule
\textbf{Stage} & \textbf{Decision} & \textbf{Required work} \\
\midrule
Compiler comparison and exact proof
& Remove certified semantic matches before test-conditional planning.
& Compare circuits compiled under identical pinned settings; verify matches on the changed support after exact context cancellation, or by bounded full-register exact algebra. Results are reusable across tests. \\
Rotation threshold gate
& If $R_O|\sin(\eta/2)|<\gamma$, label the pair below threshold.
& $\mathcal{O}(1)$ with known strength; a fixed one-qubit replacement also permits local axis--angle reduction. \\
Structural light cone
& Label a site outside the cone as zero-effect for this test.
& Linear in the size of the circuit dependency graph; no state data. \\
Propagation and commutation
& Form $O_j$ and test $[W_m,O_j]=0$.
& A Clifford suffix maps each input Pauli term to one term; retaining at most $K_{\mathrm{ret}}$ terms through $L_{>j}$ suffix gates requires $\mathcal{O}(K_{\mathrm{ret}}L_{>j})$ local term updates. \\
Exact state evaluation
& Evaluate $(a_m,b_m)$ or the exact replacement expectation.
& Numerical statevectors store $2^n$ amplitudes; density matrices store $4^n$ entries. Pauli expectations do not require a dense observable. \\
Prefix pilot
& Estimate the Pauli expectations needed by the remaining candidates.
& Repeated prefix executions; shared-shadow sampling cost grows exponentially in Pauli weight (number of nonidentity factors), as detailed in \cref{sec:scalable}. \\
Mutant execution
& Run candidates predicted above the threshold or unresolved at the declared confidence.
& Prespecified shots from each known circuit, plus test-specific overhead; this allocation differs from the unknown-source copy budget. \\
\bottomrule
\end{tabularx}
\end{table}

\begin{table}[!htbp]
\caption{Certificates available under different oracle rules.
All rows retain the declared semantics, numerical-certification requirements, and evidence provenance.}
\label{tab:oracle-planning}
\centering\small
\begin{tabularx}{\linewidth}{@{}>{\raggedright\arraybackslash}p{0.20\linewidth}XX@{}}
\toprule
\textbf{Oracle} & \textbf{Zero-effect checks} & \textbf{Nonzero-effect decision} \\
\midrule
Differential expectation
& Light cone and commutation for $O_j$; state invariance or an exact zero expectation.
& Compare the gap bounds with $\gamma_r$; the angle ceiling is $R_O|\sin(\eta/2)|$. \\
Differential TV distance
& State invariance or collective equality of outcome probabilities; structural exclusion only for restricted measured outputs.
& Compare TV bounds with $\gamma_{\mathrm{TV}}$; the angle ceiling is $|\sin(\eta/2)|$. \\
Native expectation interval
& The expectation checks establish unchanged acceptance when the original passes.
& Use the acceptance margin for sufficient survival, or test the mutant-mean enclosure against $\mathcal A_r$ (\cref{prop:native-assertion-certificates}). \\
Sampled histogram or chi-square
& Distribution equality preserves the oracle's outcome law under the same sampling procedure, not each realized outcome.
& Retain the statistic, shot count, significance, and acceptance rule; the TV ceiling alone gives no acceptance certificate. \\
\bottomrule
\end{tabularx}
\end{table}

A \emph{prefix pilot} repeatedly executes the original circuit only through the mutation site and measures quantities needed to estimate several candidate effects.
The local-Pauli version of \emph{classical shadows} measures each qubit along an independently and uniformly chosen $X$, $Y$, or $Z$ axis on each pilot shot.
It stores the chosen axes and outcomes as classical data and reuses them to estimate many Pauli expectations.
Three-point interpolation instead reconstructs a rotation response from three circuit configurations, while targeted measurements directly estimate the needed Pauli terms.
Shared shadows can amortize their locality-dependent cost over many target strings, but strength reuse also benefits three-point interpolation and targeted measurements.
\Cref{sec:pilot-cost,sec:measurement-baselines} compare these alternatives and delimit the simpler paired-execution break-even calculation.

An executable Python reference implementation, tests, and a worked notebook are available in the implementation repository \cite{quantummutantplanner}.
The package implements the planning procedures in \cref{alg:efficient-mutant-selection,alg:candidate-effect-evaluation} and a separate fixed-shot classifier.
The package's scalar planner and fixed-shot classifier use differential expectation rules; the TV and native-assertion certificates above are mathematical extensions that they do not implement.
The implementation repository's README documents installation and reproduction of the worked examples.
\Cref{sec:core-validation} details the implementation's coverage, limitations, and remaining validation needs.

\FloatBarrier

\subsection{Preserving Classifications and Mutation Scores}
\label{sec:planner-guarantee}
Skipping executions must preserve what the report says about the declared mutant population.
A certified below-threshold skip cannot hide a population-level detection for that pair, whereas a budget-based skip can leave its outcome unresolved.
The following proposition shows how valid pair-level evidence supports correct classifications and mutation-score bounds when execution is incomplete.

\begin{proposition}[Planning soundness and mutation-score bounds]
\label[proposition]{thm:planner-soundness}
Fix a differential test collection $\mathcal T$ as in \cref{sec:test-oracle}, with prespecified positive thresholds and a nonempty mutant population $\mathcal M^\star$ chosen independently of the evaluated suite.
Assume sound semantic exclusions leaving only non-equivalent mutants, valid deterministic gap bounds and zero certificates, and simultaneous statistical coverage with pilot and observed-test failure budgets $\alpha_{\mathrm{pilot}}$ and $\alpha_{\mathrm{test}}$, preserved under data-dependent execution.
If representatives are used, require certified identical per-test population detection indicators within each class and retain its original multiplicity, as detailed in \cref{sec:planner-soundness-proof}.
For each mutant or representative class, declare detection when some test has a certified lower gap bound at least its threshold, survival when every test has a certified upper gap bound strictly below its threshold, and otherwise leave it unresolved.
Budget-floor skips and unexecuted pairs supply no classification by themselves.
Then, with probability at least $1-\alpha_{\mathrm{pilot}}-\alpha_{\mathrm{test}}$, all resolved classifications are correct, no certified below-threshold skip omits a population-level detection for its pair, and
\begin{equation}
\label{eq:planner-score-bounds}
\frac{W_{\mathrm D}}{W}
\leq
\mathrm{MS}(\mathcal T;\mathcal M^\star)
\leq
\frac{W_{\mathrm D}+W_{\mathrm U}}{W},
\qquad
W=W_{\mathrm D}+W_{\mathrm S}+W_{\mathrm U}=|\mathcal M^\star|,
\end{equation}
where $W_{\mathrm D},W_{\mathrm S},W_{\mathrm U}$ are the sums of class multiplicities in the detected, survivor, and unresolved categories.
For deterministic evidence alone the corresponding claims are unconditional.
\end{proposition}

The complete assumptions and proof are given in \cref{sec:planner-soundness-proof}.
For example, 60 certified detected mutants, 25 certified survivors, and 15 unresolved mutants in a population of 100 support a mutation score between $60/100$ and $75/100$.
Dropping the unresolved mutants would change the denominator and answer a different question.
The score interval is a point exactly when no unresolved population weight remains.

The proposition establishes that combining valid evidence preserves classifications and population accounting.
Pilot-supported or analytically resolved population detections must retain that provenance and must not be relabelled as observed finite-shot kills.

The same accounting argument applies to other fixed oracle rules once their own outcomes are certified.
For a TV-threshold test, use valid TV intervals in place of expectation-gap intervals.
For a native assertion, first establish that the original passes, then use \cref{prop:native-assertion-certificates} or another valid certificate for that acceptance predicate.
Declare a suite kill when some native test certifies a kill, survival only when every native test certifies survival, and otherwise leave the mutant unresolved.
Under the same semantic, multiplicity, and simultaneous-coverage assumptions, \cref{eq:planner-score-bounds} then bounds the score defined by those oracle-specific detection indicators.
The proof uses only the correctness of the resolved indicators: detected weight is certainly killed, survivor weight certainly survives, and unresolved weight can contribute between zero and its full value.
This extension changes the per-test classifiers, not the accounting; in particular, setting $\gamma_r=h_r$ in the scalar planner would not implement native detection correctly.

\subsection{GHZ Example: Decisions and Population Accounting}
\label{sec:ghz-planning-summary}

Prepare the three-qubit Greenberger--Horne--Zeilinger (GHZ) state $(\ket{000}+\ket{111})/\sqrt2$ from $\ket{000}$ using $R_y^{(0)}(\pi/2)$, $\mathrm{CNOT}_{0,1}$, and $\mathrm{CNOT}_{0,2}$.
This entangled state is an equal superposition of all-zero and all-one bit strings, so computational-basis measurements return three perfectly correlated bits.
\Cref{tab:ghz-planning-summary} compares two tests of the same eight author-defined mutants: first the expectation of $Z_1$, then the entire computational-basis output distribution.
The circuit, input, and mutant population stay fixed while the test changes.
All eight mutants change the output state on this input, so all eight remain non-equivalent under the declared state-transforming semantics.
The detailed mutation strengths and certificates are given in \cref{sec:ghz-planning-walkthrough}.

\subsubsection{Single-Qubit Expectation Test}
\label{sec:ghz-expectation}

The first test uses $O=Z_1$ and detects a mutant when its expectation gap reaches $\gamma=0.1$.
This expectation is the probability that qubit $q_1$ is zero minus the probability that it is one; the original GHZ state gives equal probabilities and hence mean zero.
A mutant that preserves this balance survives the test even if it changes the other bits or their correlations with $q_1$.
The $Z_1$ column of \cref{tab:ghz-planning-summary} records why the planner skips or executes each mutant for this test.

The exact gaps are zero for $m_1$--$m_5$, $|\sin(0.01)|<0.1$ for $m_6$, and one for $m_7$ and $m_8$.
Thus $W_{\mathrm D}=2$, $W_{\mathrm S}=6$, and $W_{\mathrm U}=0$, giving the population mutation score $2/8$ before any observed mutant execution.
The planner schedules two candidates; removing the six skipped survivors would instead give the misleading score $2/2$.

\subsubsection{Full Output-Distribution Test}
\label{sec:ghz-distribution}
Now replace the single-qubit expectation test by a comparison of the joint probabilities of all three output bits, keeping the same input and mutants.
Both tests use computational-basis measurements, but the new test retains information about which combinations of bits occur that the $Z_1$ mean discards.
Order bits as $q_0q_1q_2$ and write $p$ and $p_m$ for the original and mutant distributions.
The original assigns probability $1/2$ to each of $000$ and $111$, and zero to every other bit string.
Measure the distribution change by total-variation (TV) distance and declare detection under the separate population rule
\begin{equation}
\label{eq:ghz-tvd}
D_{\mathrm{TV}}(p,p_m)=\frac12\sum_x|p(x)-p_m(x)|\geq\gamma_{\mathrm{TV}},
\qquad \gamma_{\mathrm{TV}}=0.1.
\end{equation}
Because both distributions sum to one, probability lost from some outcomes is gained by others, so the sum of absolute differences counts the redistributed probability mass twice.
The factor $1/2$ corrects this double counting and makes the distance range from zero to one.
This illustrative differential distribution contract is separate from any supplied benchmark's native assertion, which retains its own acceptance rule.
QuMu's normalized histogram discrepancy is the empirical analogue of this distance \cite[Sec.~4.2]{polo2025mutation_based}.
The two thresholds apply to different quantities: $\gamma=0.1$ for $Z_1$ corresponds to a probability change of $0.05$ for either value of $q_1$, whereas $\gamma_{\mathrm{TV}}=0.1$ applies to the joint distribution.
\Cref{tab:ghz-planning-summary} places the two tests side by side; write $s=\sin(0.01)$.

\begin{table}[!htbp]
\caption{The same eight GHZ mutants under the $Z_1$ expectation test ($\gamma=0.1$) and the full-distribution test ($\gamma_{\mathrm{TV}}=0.1$).
Every $Z_1$ skip is a survivor; every execute entry predicts detection.
For distributions, $x,y:a$ assigns probability $a$ to each listed bit string in $q_0q_1q_2$ order; all unlisted probabilities are zero.
All detections are established analytically; the example includes no finite-shot test outcomes.}
\label{tab:ghz-planning-summary}
\centering
\small
\setlength{\tabcolsep}{4pt}
\renewcommand{\arraystretch}{1.12}
\begin{tabularx}{\linewidth}{@{}l>{\raggedright\arraybackslash}p{0.20\linewidth}>{\raggedright\arraybackslash}X>{\raggedright\arraybackslash}p{0.26\linewidth}>{\raggedright\arraybackslash}p{0.14\linewidth}@{}}
\toprule
 & & \textbf{$Z_1$ test (\cref{sec:ghz-expectation})} & \multicolumn{2}{c}{\textbf{Full-distribution test (\cref{sec:ghz-distribution})}} \\
\cmidrule(lr){3-3}\cmidrule(l){4-5}
\textbf{ID} & \textbf{Mutation} & \textbf{Planning decision} & \textbf{Nonzero probabilities} & \textbf{$D_{\mathrm{TV}}$; outcome} \\
\midrule
$m_1$ & Delete $\mathrm{CNOT}_{0,2}$ & Skip: outside light cone & $000,110:1/2$ & $1/2$; detected \\
$m_2$ & Append $X_2$ & Skip: outside light cone & $001,110:1/2$ & $1$; detected \\
$m_3$ & Append $Z_1$ & Skip: observable commutation & $000,111:1/2$ & $0$; survivor \\
$m_4$ & Append $X_1$ & Skip: state-dependent zero & $010,101:1/2$ & $1$; detected \\
$m_5$ & Reverse the $R_y$ angle & Skip: finite-strength zero & $000,111:1/2$ & $0$; survivor \\
$m_6$ & Offset the $R_y$ angle by $0.01$ & Skip: below threshold & $000:(1-s)/2$\newline $111:(1+s)/2$ & $s/2$; survivor \\
$m_7$ & Offset the $R_y$ angle by $\pi/2$ & Execute: predicted detection & $111:1$ & $1/2$; detected \\
$m_8$ & Delete $\mathrm{CNOT}_{0,1}$ & Execute: predicted detection & $000,101:1/2$ & $1/2$; detected \\
\bottomrule
\end{tabularx}
\end{table}

The additional detections are $m_1$, $m_2$, and $m_4$: each preserves the equal probabilities of zero and one on $q_1$ but moves probability to bit strings absent from the original output.
For example, appending $X_1$ in $m_4$ replaces $000$ and $111$ by $010$ and $101$, so the $Z_1$ mean remains zero while the full distributions have disjoint support and total-variation distance one.
The structural shortcut also changes: each full-register outcome projector acts on every qubit, so no mutation site can be excluded by the structural light-cone rule for these projectors.

Equal computational-basis probabilities can still conceal a changed quantum state.
For $m_3$, the diagonal final $Z_1$ commutes with every output projector and therefore certifies the entire distribution unchanged.
Its output is nevertheless $(\ket{000}-\ket{111})/\sqrt2$, so the state-side invariance certificate fails.
State invariance guarantees distribution equality, which can also hold for different states.
The sign reversal $m_5$ has the same probabilities through cancellation for the specified input.
The remaining survivor, $m_6$, changes the probabilities, but its distance $s/2$ stays below $\gamma_{\mathrm{TV}}$.
The TV rotation ceiling already settles this case without calculating those probabilities: $|\sin(0.01/2)|<0.1$.

Thus changing the test raises the population mutation score from $2/8$ to $5/8$ while keeping the state-transforming population fixed.
Changing the program's declared semantics would be a separate step: under an interface restricted to this fixed input and computational readout, $m_3$ and $m_5$ would be interface-equivalent and require a separately declared score denominator.
A sampled histogram or chi-square oracle must still use its own sample-size, significance, and acceptance rules; handling zero expected bins must avoid division by zero.
Neither the analytical outcomes here nor term-wise projector gaps constitute observed finite-shot rejections.

\FloatBarrier
\section{Practical Implications and Representative Examples}
\label{sec:core-examples}

\subsection{One-Qubit Pauli-Axis Rotation}
\label{sec:one-qubit-example}

A changed output state can leave a measurement unchanged, as the GHZ survivors illustrate.
For a single rotation, the exact gap shows how this depends on both the original angle and the mutation strength.
Consider the circuit $U=R_y(\theta)$ with input $\rho=\ket{0}\bra{0}$ and observable $O=Z$.
The original expectation is $\cos\theta$; replacing the gate by $R_y(\theta+\eta)$ changes it to $\cos(\theta+\eta)$.
Thus \cref{thm:core-exact-pauli-rotation} with $A=Y$ gives the exact gap
$$
\Delta_m
=
|\cos\theta-\cos(\theta+\eta)|
$$
and expanding \cref{eq:finite-strength-gap} for small $\eta$ gives the first-order expression
$$
\Delta_m
\approx
|\eta\sin\theta|.
$$
A zero first-order coefficient, as at $\theta=0$ or $\pi$, does not establish a zero exact gap.
The exact expression must still be evaluated at the declared mutation strength.

The outcome range is $R_Z=2$, so dividing the gap by two gives the normalized signal and its upper bound from \cref{cor:test-independent-rotation-ceiling}:
\begin{equation*}
\delta_m
=
\frac{\Delta_m}{2}
=
\left|\sin(\eta/2)\sin(\theta+\eta/2)\right|
\leq
|\sin(\eta/2)|.
\end{equation*}
Equality holds when $\theta+\eta/2\equiv\pi/2\pmod{\pi}$, so a suitable original angle attains the ceiling.
The same small-strength limit also constrains the separate task of identifying which circuit produced an unknown output.
Under the independent-execution assumptions of \cref{cor:test-independent-rotation-ceiling}, any fixed nonzero advantage over random guessing with equal priors requires $\Omega(1/\eta^2)$ executions as $\eta\to0$.

For a fixed original angle, the coefficients in \cref{eq:exact-pauli-two-coefficients} are $a_m=\cos\theta$ and $b_m=\sin\theta$.
The largest gap over all mutation strengths, from \cref{eq:finite-strength-maximum}, is therefore
$$
\max_\eta\Delta_m(\eta)
=
1+|\cos\theta|.
$$
The zero set in \cref{eq:finite-strength-zero-set} includes $\eta_0\equiv-2\theta\pmod{2\pi}$, which is distinct from the unchanged angle when $\sin\theta\neq0$.
This strength reverses the sign of the original angle: $R_y(\theta)\rightarrow R_y(-\theta)$.
The replacement preserves the $Z$ expectation on input $\ket{0}$ and is therefore test-equivalent for this pair, although the two circuit unitaries are generally not equivalent.
\Cref{fig:finite-strength-gap} shows this cancellation and the maximizing strength for the representative choice $\theta=\pi/4$.

\begin{figure}[!htb]
\centering
\begin{tikzpicture}
\begin{axis}[
    width=0.78\linewidth,
    height=5.3cm,
    domain=-pi:pi,
    samples=241,
    axis lines=left,
    xlabel={mutation strength $\eta$},
    ylabel={$\Delta_m(\eta)$},
    xmin=-pi,
    xmax=pi,
    ymin=0,
    ymax=2.05,
    xtick={-pi,-pi/2,0,pi/2,pi},
    xticklabels={$-\pi$,$-\pi/2$,$0$,$\pi/2$,$\pi$},
    ytick={0,0.5,1,1.5},
    grid=major,
    grid style={gray!20},
    tick label style={font=\small}
]
\addplot[blue!70!black,very thick] {abs(cos(deg(pi/4))-cos(deg(pi/4+x)))};
\addplot[only marks,mark=o,mark size=2.4pt,blue!70!black] coordinates {(-pi/2,0)};
\node[anchor=south east,font=\small] at (axis cs:-pi/2,0.03) {$\eta_0=-\pi/2$};
\addplot[only marks,mark=*,mark size=2.2pt,red!70!black] coordinates {(3*pi/4,{1+sqrt(2)/2})};
\node[anchor=south east,xshift=-2pt,yshift=2pt,font=\small] at (axis cs:3*pi/4,{1+sqrt(2)/2}) {$\eta_\star=3\pi/4$};
\end{axis}
\end{tikzpicture}
\caption{Exact finite-strength gap for $R_y(\pi/4)\rightarrow R_y(\pi/4+\eta)$ with input $\ket{0}$ and observable $Z$.
Besides the unchanged angle $\eta=0$, the sign-reversal strength $\eta_0=-\pi/2$ is exactly test-equivalent, while the gap reaches $1+1/\sqrt{2}$ at $\eta_\star=3\pi/4$.}
\label{fig:finite-strength-gap}
\Description{The observable gap as a function of mutation strength eta, from minus pi to pi.
The curve has zeros at eta equals minus pi over two and zero, and reaches one plus one over the square root of two at eta equals three pi over four.}
\end{figure}
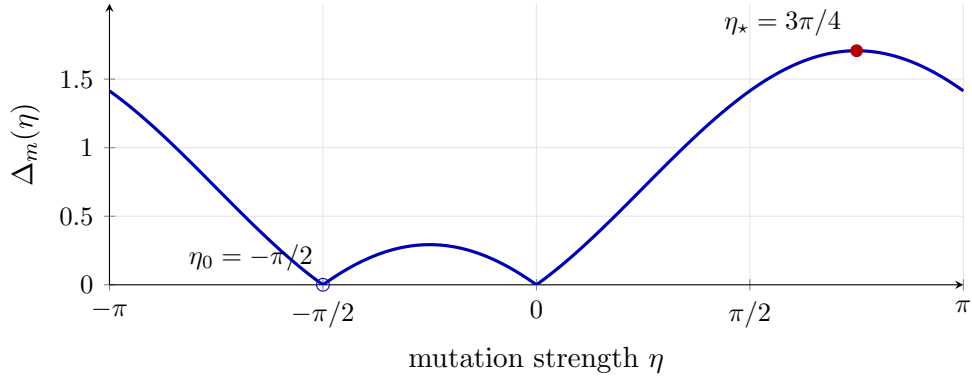

The same mechanism applies to other Pauli-axis rotations covered by \cref{thm:core-exact-pauli-rotation}.
Under the same $(\ket{0},Z)$ test, $R_x(\theta)\rightarrow R_x(\theta+\eta)$ gives the identical gap $|\cos\theta-\cos(\theta+\eta)|$.
Changing the measurement to $X$ gives gap $2|\sin\theta|$ for the $R_y$ sign reversal; changing it to $Y$ gives the same gap for the $R_x$ sign reversal.

Changing the measurement cannot help if the mutation leaves the state itself unchanged.
For example, $R_z(\theta)\ket{0}=e^{-i\theta/2}\ket{0}$ changes only a global phase.
Since $[Z,\ket{0}\bra{0}]=0$, the state-side criterion in \cref{prop:state-side-dual} certifies that every $R_z$ angle mutation is invisible to every observable on this input.
Changing the input to $\ket{+}$ removes that obstruction: an $X$ measurement then gives $|\cos\theta-\cos(\theta+\eta)|$, and a $Y$ measurement gives gap $2|\sin\theta|$ for the sign reversal.

\subsection{Fixed Cross-Axis One-Qubit Replacement}

The rotation formula also applies when a replacement changes the gate's axis.
Consider a one-gate circuit with $G_j=R_x(\pi/2)$, replacement $\widetilde G_j=R_z(\pi/2)$, input $\rho=\ket{0}\bra{0}$, and observable $O=Z$.
To express the replacement as a change applied after the original gate, compute the relative unitary $W_m=\widetilde G_jG_j^\dagger$ from \cref{cor:fixed-one-qubit-replacement}:
\begin{equation*}
\begin{aligned}
W_m
&=
R_z(\pi/2)R_x(-\pi/2)
\\
&=
\frac{1}{2}\left[I+i(X+Y-Z)\right]
\\
&=
e^{-i(2\pi/3)A_{\mathrm{eff}}/2},
\end{aligned}
\end{equation*}
where
\begin{equation*}
A_{\mathrm{eff}}
=
\frac{-X-Y+Z}{\sqrt{3}},
\qquad
\eta_{\mathrm{eff}}
=
\frac{2\pi}{3},
\qquad
\chi
=
0.
\end{equation*}
The replacement is therefore one exact rotation about an axis combining $X$, $Y$, and $Z$.
Adding separate $X$- and $Z$-axis effects would not describe this relative operation.
To evaluate its gap, first compute the state after the original gate:
\begin{equation*}
\tau_j
=
R_x(\pi/2)\ket{0}\bra{0}R_x(-\pi/2)
=
\frac{I-Y}{2}
\end{equation*}
The part of the propagated observable that anticommutes with the effective axis is
\begin{equation*}
O_{j,-}
=
\frac{Z-A_{\mathrm{eff}}ZA_{\mathrm{eff}}}{2}
=
\frac{X+Y+2Z}{3}.
\end{equation*}
These expressions give $a_m=-1/3$ and $b_m=-1/\sqrt{3}$ in \cref{eq:exact-pauli-two-coefficients}.
Substituting the coefficients and effective angle into \cref{eq:finite-strength-gap} gives
\begin{equation*}
\Delta_m
=
\left|
\left(1-\cos\frac{2\pi}{3}\right)\left(-\frac{1}{3}\right)
+
\sin\frac{2\pi}{3}\left(-\frac{1}{\sqrt{3}}\right)
\right|
=
1.
\end{equation*}
The circuit outputs provide a direct check: the original $Z$ expectation is zero, while $R_z(\pi/2)\ket{0}=e^{-i\pi/4}\ket{0}$ gives mutant expectation one.
If a replacement instead differs from the original only by a global phase, then $\widetilde G_j=e^{i\chi}G_j$ and $W_m=e^{i\chi}I$.
In that branch of \cref{cor:fixed-one-qubit-replacement}, every test has zero gap and no effective axis is defined.

\subsection{Mutation-Guided Test Augmentation}
\label{sec:core-augmentation}

The rotation examples distinguish two ways to expose a surviving mutant: change the input or change the measured output quantity.
To assess the same program contract, the new test must stay within its declared input domain and permitted measurement interface.
A measurement outside that interface can show that the states differ, but does not by itself reveal a weakness in the native assertions.

For a fixed mutant strength $\eta_m$, the two coefficients in \cref{eq:exact-pauli-two-coefficients} allow candidate inputs and observables to be compared.
For each candidate pair $t=(\rho,O)$, \cref{eq:finite-strength-gap} gives the signed expectation change
\begin{equation*}
d_{m,t}(\eta_m)
=
(1-\cos\eta_m)a_{m,t}+\sin\eta_m\,b_{m,t}.
\end{equation*}
By \cref{eq:suite-detection-set}, the mutant survives the current suite when $|d_{m,t_r}(\eta_m)|<\gamma_r$ for every existing test $t_r$.
A new test $t'$ detects it under the differential rule exactly when $|d_{m,t'}(\eta_m)|$ reaches that test's declared threshold.
A zero gap is the stronger, special case of test-relative invisibility.
For a fixed site and axis, \cref{cor:finite-strength-visibility} makes every strength invisible to the suite exactly when $a_{m,t}=b_{m,t}=0$ for every test $t$.
Even when these coefficients do not all vanish, one strength can be invisible to the entire suite if it lies in every test's zero set from \cref{eq:finite-strength-zero-set}.

For a fixed input, the state difference determines the largest gap among Hermitian observables whose outcomes lie in $[-1,1]$, equivalently whose operator norm is at most one.
Form $\Delta\tau_j$ from \cref{eq:state-side-difference} and choose $O_j^\star=\operatorname{sgn}(\Delta\tau_j)$, as in \cref{prop:state-side-dual}.
This observable assigns $+1$, $-1$, or $0$ to the positive, negative, or zero eigenspaces of the state difference.
Propagating it forward gives the output measurement $O^\star=U_{>j}O_j^\star U_{>j}^\dagger$.
By \cref{eq:state-side-trace-norm}, its gap is exactly $\|\Delta\tau_j\|_1$.
If $\Delta\tau_j=0$, no change of observable can expose the mutant on this input.

Attaining this mathematical optimum can be expensive.
Constructing $O_j^\star$ requires the full mutation-site state and may be intractable in a wide circuit.
Measuring $O^\star$ requires a circuit that changes to its eigenbasis, or an equivalent measurement construction.
Its depth, routing, calibration, and noise can dominate the cost.
The optimum therefore remains an information-theoretic benchmark unless the measurement is permitted, its construction is available, and its cost is included in the test budget.
Even an attainable nonzero gap must reach the declared threshold to count as differential detection; a native acceptance assertion must still be judged by its own rule.

When measurements are restricted, declare an admissible output family $\mathcal F_{\mathrm{out}}$ with $\|O\|_\infty\leq1$ and an implementation cost for each candidate.
Screen each $O\in\mathcal F_{\mathrm{out}}$ using its gap $|\Tr(U_{>j}^\dagger O U_{>j}\Delta\tau_j)|$ and that cost.
Restricting the output to Pauli observables differs from restricting the mutation-site observable to Paulis unless the suffix maps the relevant families into one another.
The best gap in the restricted family is a lower bound on the unrestricted optimum and can be exponentially smaller.
The evaluation identity and this limitation are given in \cref{sec:restricted-pauli-augmentation}.
Both the unrestricted construction and the restricted screen can use original-prefix information and the known relative unitary without first executing the mutant, provided that the needed expectations and observable representations are accessible.

For the sign-reversal mutant $R_y(\theta)\rightarrow R_y(-\theta)$, changing the observable is enough to obtain a nonzero gap when $\sin\theta\neq0$.
Substituting $\eta=-2\theta$ into \cref{eq:finite-strength-gap} gives zero for $(\ket{0},Z)$, while $(\ket{0},X)$ gives
\begin{equation*}
\Delta_m
=
|\sin\theta-(-\sin\theta)|
=
2|\sin\theta|.
\end{equation*}
Measuring $Y$ would still give zero gap because both states have real amplitudes and zero $Y$ expectation.
To measure $X$, add a Hadamard gate $H$ before the usual computational-basis readout, using $H^\dagger ZH=X$.
For the introductory example $\theta=\pi/4$, that readout returns bit 1 with probabilities $(1-1/\sqrt{2})/2$ for the original and $(1+1/\sqrt{2})/2$ for the mutant.
Their difference is $1/\sqrt{2}$, corresponding to the $X$-expectation gap $\sqrt{2}$ under the $\pm1$ outcome labels.

This $X$ test also attains the unrestricted optimum.
At the mutation site, \cref{eq:state-side-difference} gives $\Delta\tau_j=\sin\theta\,X$.
Hence \cref{prop:state-side-dual,eq:state-side-trace-norm} give $O_j^\star=\operatorname{sgn}(\sin\theta)X$ and $\|\Delta\tau_j\|_1=2|\sin\theta|$.
The state difference thus identifies the optimal observable without a candidate search.
To show how the gap changes between the original and optimal measurements, consider the unit-norm axes
\begin{equation*}
O(\varphi)
=
\cos\varphi\,Z+\sin\varphi\,X,
\qquad
0\leq\varphi\leq\frac{\pi}{2}.
\end{equation*}
The sign-reversal gap along this family is
\begin{equation*}
\Delta_m(\varphi)
=
2|\sin\theta\sin\varphi|.
\end{equation*}
For $\theta=\pi/4$, \cref{fig:observable-augmentation-gap} shows the gap increasing continuously from zero for $Z$ to the optimum for $X$ supplied by \cref{prop:state-side-dual}.

\begin{figure}[!htb]
\centering
\begin{tikzpicture}
\begin{axis}[
    width=0.76\linewidth,
    height=5.1cm,
    domain=0:pi/2,
    samples=181,
    axis lines=left,
    xlabel={observable direction $O(\varphi)$},
    ylabel={$\Delta_m(\varphi)$},
    xmin=0,
    xmax=pi/2,
    ymin=0,
    ymax=1.55,
    xtick={0,pi/4,pi/2},
    xticklabels={$Z$,$(Z+X)/\sqrt{2}$,$X$},
    ytick={0,0.5,1,1.41421356},
    yticklabels={$0$,$0.5$,$1$,$\sqrt{2}$},
    grid=major,
    grid style={gray!20},
    tick label style={font=\small}
]
\addplot[blue!70!black,very thick] {sqrt(2)*sin(deg(x))};
\addplot[only marks,mark=o,mark size=2.4pt,blue!70!black] coordinates {(0,0)};
\addplot[only marks,mark=*,mark size=2.2pt,red!70!black] coordinates {(pi/2,{sqrt(2)})};
\end{axis}
\end{tikzpicture}
\caption{Observable-direction gap for the sign reversal $R_y(\pi/4)\rightarrow R_y(-\pi/4)$ under input $\ket{0}$.
The existing $Z$ test is invisible, intermediate axes provide partial separation, and the state-side optimum $X$ attains the trace-norm gap $\sqrt{2}$.}
\label{fig:observable-augmentation-gap}
\Description{The observable gap increases along the measurement direction from zero at the Z axis to the square root of two at the X axis, following the square root of two times sine of the direction angle.}
\end{figure}
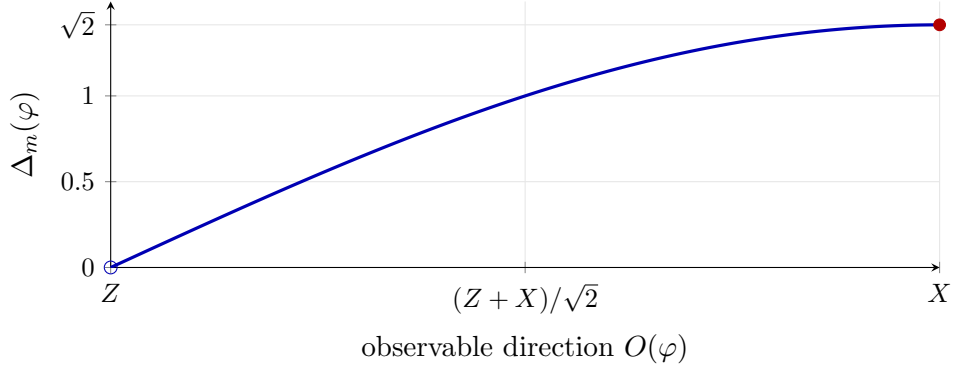

\paragraph{Two-qubit measurement construction.}
The optimal observable need not be a single Pauli string, even for two qubits.
Prepare $\ket{\psi}=(\ket{00}+\ket{11})/\sqrt2$ and let the mutant append the phase gate $S=\operatorname{diag}(1,i)$ on the first qubit, giving $\ket{\phi}=(\ket{00}+i\ket{11})/\sqrt2$.
Computational-basis readout gives the same probabilities for both states, although their difference is
\begin{equation}
\label{eq:bell-phase-difference}
D=\ket{\psi}\bra{\psi}-\ket{\phi}\bra{\phi}=\frac{XX-YY-XY-YX}{4}.
\end{equation}
The difference has nonzero eigenvalues $\pm1/\sqrt2$ on the subspace spanned by $\ket{00}$ and $\ket{11}$.
Its sign observable therefore gives the optimum:
\begin{equation}
\label{eq:bell-phase-optimum}
O^\star=\operatorname{sgn}(D)=\frac{XX-YY-XY-YX}{2\sqrt2},
\qquad |\Tr(O^\star D)|=\sqrt2.
\end{equation}
Each individual Pauli string has gap at most one in this example, below the optimum $\sqrt2$.
To measure $O^\star$, use bit order $q_0q_1$ and apply $\mathrm{CNOT}_{0,1}$, then $T=\operatorname{diag}(1,e^{i\pi/4})$ and $H$ on $q_0$, followed by computational readout.
For $B=(HT\otimes I)\mathrm{CNOT}_{0,1}$, direct multiplication gives $BO^\star B^\dagger=\operatorname{diag}(1,0,-1,0)$.
Assign outcomes $00,01,10,11$ the values $+1,0,-1,0$, respectively; the two output means are $+1/\sqrt2$ and $-1/\sqrt2$.
This realizes the optimal gap with one additional CNOT and two one-qubit gates per measured circuit, plus both readouts and classical relabeling.
The zero labels correspond to the zero eigenspace of the state difference; neither ideal state has weight there.
Whether this construction costs less than separate Pauli measurements depends on basis synthesis, connectivity, and execution costs, which this illustration does not measure.
The measurement is admissible only when the program contract permits this basis change.

These tests are constructed from the very mutations they expose.
An adequacy claim for the augmented suite therefore requires evaluation on held-out mutants, representative faults, or another population not used to construct the tests.

\subsection{A Local Mutation Exposed Through Correlation}

A measurement on one qubit can reveal a mutation on another when the suffix couples them.
Consider a mutation site followed by the suffix $\mathrm{CNOT}_{q,r}$ and final observable $O=Z_r$.
Propagating the observable backward as in \cref{thm:core-exact-local} gives
$$
O_j
=
\mathrm{CNOT}_{q,r}^\dagger Z_r\mathrm{CNOT}_{q,r}
=
Z_qZ_r.
$$
Thus the final $Z_r$ expectation reads the two-qubit correlation $Z_qZ_r$ at the mutation site.
For an $R_y$ angle mutation on qubit $q$, \cref{eq:core-exact-pauli-gap} gives
$$
\Delta_m
=
\left|
(1-\cos\eta)\Tr(Z_qZ_r\tau_j)
+
\sin\eta\Tr(X_qZ_r\tau_j)
\right|.
$$
The first-order response depends on the $X_qZ_r$ correlation: the state-dependent sensitivity in \cref{eq:core-state-sensitivity} is $|\Tr(X_qZ_r\tau_j)|$.
If this coefficient vanishes, the $Z_qZ_r$ term can still produce a gap at finite strength.
For example, take $\tau_j=\ket{00}\bra{00}$.
Then $\Tr(Z_qZ_r\tau_j)=1$ and $\Tr(X_qZ_r\tau_j)=0$, giving
\begin{equation*}
\Delta_m(\eta)
=
1-\cos\eta,
\end{equation*}
even though the first-order approximation about $\eta=0$ is zero.
As \cref{fig:correlation-finite-strength-gap} shows, the gap grows quadratically near the unchanged circuit and reaches the full outcome range at a half-turn mutation.

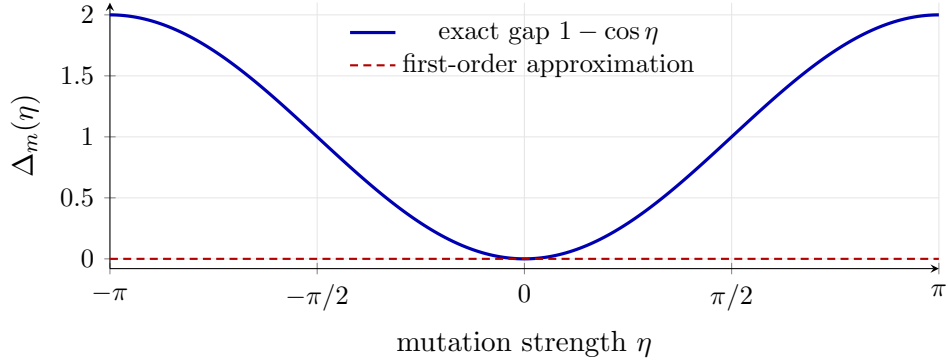
\begin{figure}[!htb]
\centering
\begin{tikzpicture}
\begin{axis}[
    width=0.76\linewidth,
    height=5.1cm,
    domain=-pi:pi,
    samples=241,
    axis lines=left,
    xlabel={mutation strength $\eta$},
    ylabel={$\Delta_m(\eta)$},
    xmin=-pi,
    xmax=pi,
    ymin=-0.08,
    ymax=2.1,
    xtick={-pi,-pi/2,0,pi/2,pi},
    xticklabels={$-\pi$,$-\pi/2$,$0$,$\pi/2$,$\pi$},
    ytick={0,0.5,1,1.5,2},
    grid=major,
    grid style={gray!20},
    tick label style={font=\small},
    legend style={font=\small,draw=none,fill=none,at={(0.5,0.97)},anchor=north}
]
\addplot[blue!70!black,very thick] {1-cos(deg(x))};
\addlegendentry{exact gap $1-\cos\eta$}
\addplot[red!70!black,thick,densely dashed] {0};
\addlegendentry{first-order approximation}
\end{axis}
\end{tikzpicture}
\caption{Exact correlation-mediated gap for $\tau_j=\ket{00}\bra{00}$, an $R_y(\eta)$ mutation on control qubit $q$, a $\mathrm{CNOT}_{q,r}$ suffix, and final observable $Z_r$.
The first-order coefficient vanishes, but the exact gap grows as $1-\cos\eta$ and reaches the full observable range at $|\eta|=\pi$.}
\label{fig:correlation-finite-strength-gap}
\Description{The exact gap, one minus cosine eta, rises from zero at eta equals zero to two at eta equals plus or minus pi.
The first-order approximation is a horizontal line at zero throughout.}
\end{figure}

For $\eta\not\equiv0\pmod{2\pi}$, the mutation lies inside the structural light cone and its relative unitary does not commute with $O_j$.
Neither zero-effect certificate in \cref{eq:core-light-cone-discard,thm:core-exact-local} therefore applies, but these facts alone do not determine the gap.
To see the role of the state, keep the same site, axis, and observable and take $\tau_j=\ket{++}\bra{++}$.
Both $\Tr(Z_qZ_r\tau_j)$ and $\Tr(X_qZ_r\tau_j)$ now vanish, so the exact gap is zero at every strength.
At nontrivial strengths the relative unitary still commutes with neither $O_j$ nor $\tau_j$.
Thus the absence of a zero-effect certificate does not guarantee a nonzero signal for the selected test.

\subsection{Efficient Analysis for Stabilizer Circuits}
\label{sec:stabilizer-family}
The correlation example needs only two state expectations, but computing even those quantities can be expensive for a general circuit.
For pure stabilizer states and a Clifford suffix, they can be obtained efficiently on a classical computer.
A pure \emph{stabilizer state} is specified by $n$ independent commuting Pauli constraints that leave it unchanged.
Products of these signed Pauli constraints form its \emph{stabilizer group}.
A signed stabilizer tableau stores those constraints as a compact binary matrix with sign information, avoiding a vector of $2^n$ amplitudes.
Standard stabilizer simulation supplies the required state information \cite{gottesman1999heisenberg,aaronson2004stabilizer}.
A Clifford suffix also maps a Pauli output observable backward to a single signed Pauli string.
When the mutation axis anticommutes with the propagated observable, the two Pauli observables that determine the response also anticommute.
They cannot both have nonzero expectations in a pure stabilizer state.
This leaves exactly three possible gap curves, even when the mutation itself introduces a non-Clifford rotation.

\begin{theorem}[Three-shape stabilizer visibility]
\label{thm:stabilizer-visibility}
Suppose $\tau_j$ is a pure stabilizer state supplied by a signed stabilizer tableau, the suffix is Clifford, the output observable is a unit-norm Hermitian Pauli string, and $W_m=e^{-i\eta A/2}$ for a Hermitian Pauli string $A$.
Write $P=O_j$, including the sign introduced by Clifford propagation.
If $[P,A]=0$, the gap vanishes for all strengths.
Otherwise $P$ and $Q=iPA$ are anticommuting Hermitian Pauli observables, and their expectations $a_m=\Tr(P\tau_j)$ and $b_m=\Tr(Q\tau_j)$ give the exhaustive classification
\begin{equation}
\label{eq:stabilizer-gap-shapes}
\Delta_m(\eta)=
\begin{cases}
0, & (a_m,b_m)=(0,0),\\
|\sin\eta|, & a_m=0,\quad b_m\in\{-1,1\},\\
1-\cos\eta, & a_m\in\{-1,1\},\quad b_m=0.
\end{cases}
\end{equation}
For an original circuit of $L$ one- and two-qubit Clifford gates on a supplied pure stabilizer input, a fixed Pauli output observable, and $M$ explicitly listed site--axis candidates, all $M$ response curves can be obtained with an $\mathcal{O}(Ln+Mn^3)$ binary-operation bound and no quantum shots.
This bound concerns the symbolic curves; certified comparison with numerical thresholds additionally requires exact or validated angle arithmetic.
\end{theorem}

\begin{proof}
Clifford propagation maps the output string to a signed Pauli $P$.
If $P$ commutes with $A$, then $O_{j,-}=0$ and the gap is zero.
Otherwise $O_{j,-}=P$, $Q=iPA$ is Hermitian with $Q^2=I$, and $PQ=-QP$.
A Pauli expectation in a pure stabilizer state is $1$ or $-1$ when the appropriately signed Pauli belongs to the stabilizer group, and zero otherwise.
The two expectations cannot both be nonzero, because that would make the state a simultaneous eigenvector of anticommuting observables.
Substitution in \cref{eq:finite-strength-gap} gives \cref{eq:stabilizer-gap-shapes}.

Group candidates by site and propagate the output Pauli backward once, storing its signed value at each candidate site.
Then propagate the supplied input tableau forward once, evaluating all candidates at each site before advancing.
These shared passes cost $\mathcal{O}(Ln)$: elementary Clifford tableau updates take $\mathcal{O}(n)$ operations \cite[Sec.~III]{aaronson2004stabilizer}, and propagating and storing each signed Pauli takes at most $\mathcal{O}(n)$ operations per gate.
For each candidate, Pauli commutation and multiplication take $\mathcal{O}(n)$ operations, and at most two signed stabilizer-membership tests determine the required expectations.
Each signed stabilizer-membership test can be implemented by Gaussian elimination over $\mathbb F_2$ in $\mathcal{O}(n^3)$ operations.
The $M$ candidate checks therefore cost $\mathcal{O}(Mn^3)$ in total, giving the overall bound $\mathcal{O}(Ln+Mn^3)$.
This bound establishes polynomial cost and allows for faster update and membership-test implementations.
No enumeration of the $2^n$ stabilizer elements or construction of a dense state vector is needed.
\end{proof}

A faster implementation uses a full Aaronson--Gottesman tableau, which stores $n$ destabilizer generators in addition to the $n$ signed stabilizer generators \cite[Sec.~III]{aaronson2004stabilizer}.
Each destabilizer anticommutes with its paired stabilizer generator and commutes with the other stabilizer generators.
For a queried Pauli, anticommutation with any stabilizer generator implies zero expectation; otherwise, its commutation relations with the destabilizers identify the stabilizer product whose sign relative to the queried Pauli determines the expectation.
Scanning these rows and multiplying the selected stabilizer generators takes $\mathcal{O}(n^2)$ binary operations per query without changing the state tableau.
Maintaining the full tableau preserves the $\mathcal{O}(Ln)$ shared-pass cost, so a supplied full input tableau gives an overall bound of $\mathcal{O}(Ln+Mn^2)$.
If only the signed stabilizer generators are supplied, Gaussian elimination and symplectic completion can construct the additional rows once in $\mathcal{O}(n^3)$ operations, giving $\mathcal{O}(n^3+Ln+Mn^2)$ including preprocessing.

For a concrete family, prepare $\ket{\mathrm{GHZ}_n}=(\ket{0}^{\otimes n}+\ket{1}^{\otimes n})/\sqrt{2}$ for any $n\geq2$, and append an identity mutation site after the preparation.
At that site, the pairs $(A,O)=(Z_1,X_1\cdots X_n)$, $(Z_1,Y_1X_2\cdots X_n)$, and $(X_1,Z_1)$ give respectively $1-\cos\eta$, $|\sin\eta|$, and zero.
In the last case the axis $X_1$ acts inside the measured qubit's structural cone and anticommutes with its observable, yet both coefficient expectations vanish.
This gives a state-dependent zero that neither structural exclusion nor observable commutation explains, for arbitrarily large registers.
The same coefficient pairs can be reused for any number of declared strengths without further state analysis or quantum shots.
Evaluating or reporting each strength still incurs scalar arithmetic or output cost.

The state representation is an essential assumption.
A Clifford suffix alone does not supply the mutation-site tableau, and the cost bound does not include recognizing whether an arbitrary preparation yields a stabilizer state.
Even one non-Clifford preparation gate can invalidate the three-shape conclusion: $R_y(\pi/4)\ket{0}$ has nonzero $X$ and $Z$ expectations.
Decomposing a state into stabilizer components may still permit efficient expectation calculations for restricted circuits, but that is a separate computational claim.

The theorem also excludes arbitrary convex mixtures of pure stabilizer states.
For $\tau_j=(\ket{0}\bra{0}+\ket{+}\bra{+})/2$, $P=Z$, and $A=Y$, the coefficients are $a_m=b_m=1/2$ and the gap is $|1-\cos\eta+\sin\eta|/2$.
An arbitrary non-Clifford suffix or a sum of Pauli output terms can likewise produce other gap curves.
Even within the theorem's assumptions, a threshold decision at a boundary requires a symbolic comparison or validated bounds; without them it remains unresolved under \cref{thm:planner-soundness}.

\subsection{Computational Scope and Scalability}
\label{sec:rule-scalability}
The stabilizer theorem shows why register width alone does not determine the cost of these rules.
The local response identities hold for any finite number of qubits; obtaining the state expectations and operator representations needed to use them can nevertheless be expensive.
The reference implementation's four-qubit default applies to dense symbolic proofs, or to the changed support in a local channel proof, rather than to the circuit width admitted by every rule.
Structural exclusion and sparse Pauli commutation can operate on much wider circuits without constructing a full-register matrix.
An exact identity therefore has to be distinguished from an efficient algorithm for evaluating all of its inputs.

\Cref{tab:scalable-operator-rules,tab:scalable-state-rules} separate these costs for the rules developed in this paper.
Here $n$ is the register width, $L$ the relevant circuit length, $k$ the mutation or gate support size, $K$ the represented Pauli-term count, and $M$ the number of declared candidates.
Arithmetic bounds assume fixed precision unless exact symbolic expressions are explicitly involved; their simplification and bit lengths require additional accounting.
The tables describe available mathematical routes, while \cref{tab:evidence-implementation} identifies which routes the reference code implements.
\Cref{sec:planner-growth} gives the representation sizes and configured fallback limits.

\begin{table}[!htbp]
\caption{Scaling of semantic, structural, and observable-side rules.
A cheap operator calculation does not by itself supply the state-dependent gap.}
\label{tab:scalable-operator-rules}
\centering\small
\renewcommand{\arraystretch}{1.12}
\begin{tabularx}{\linewidth}{@{}>{\raggedright\arraybackslash}p{0.21\linewidth}*{2}{>{\raggedright\arraybackslash}X}@{}}
\toprule
\textbf{Rule} & \textbf{Route for larger circuits} & \textbf{Limiting quantity} \\
\midrule
Semantic equivalence and duplicate channels
& Compare compiled representations at the cost of compilation; with verified common context, prove a scalar relative phase on $k$ changed qubits.
& A local dense matrix stores $4^k$ entries.
Dense full-register proofs store $4^n$ entries and symbolic expressions can grow further; compiler equality alone is insufficient for the implemented certificate. \\
Rotation ceiling and optional copy-budget floor
& Scalar arithmetic once the strength, observable range, and budget are known; the TV ceiling needs no full output distribution.
& Deriving the spectral range of an arbitrary observable can itself be expensive.
A fixed one-qubit axis--angle reduction remains local. \\
Structural light cone
& Traverse gate operands and exclude disjoint mutation supports without state simulation.
& A large cone reduces the number of exclusions.
The implementation stores boundary supports with up to $\mathcal O(Ln)$ memory, although a streaming check need not store every boundary. \\
Observable commutation and Clifford updates
& Pauli-axis commutation costs $\mathcal O(Kn)$ using explicit strings.
A Clifford suffix preserves $K$, requiring $\mathcal O(KL)$ local term updates.
& A general suffix can generate $4^n$ terms.
Truncation can control work but cannot prove an exact zero from the retained terms alone. \\
Exact local and spectral replacement rules
& A $k$-qubit relative action expands each Pauli term into at most $4^k$ terms; local eigendecomposition uses conventional $\mathcal O(8^k)$ arithmetic.
& Small $k$ controls the local action, not the cost of representing $O_j$ or obtaining its state expectations. \\
Perturbative interval and operator proxy
& Sparse commutators and coefficient-norm bounds can avoid full matrices when both the generator and $O_j$ have manageable expansions.
& Exact spectral norms or dense nested commutators can be exponential in $n$.
A loose remainder can leave the threshold unresolved, and a proxy is not a realized gap. \\
\bottomrule
\end{tabularx}
\end{table}

\begin{table}[!htbp]
\caption{Scaling of state-dependent decisions, measurements, and accounting.
Sampling bounds concern data acquisition separately from classical simulation.}
\label{tab:scalable-state-rules}
\centering\small
\renewcommand{\arraystretch}{1.12}
\begin{tabularx}{\linewidth}{@{}>{\raggedright\arraybackslash}p{0.21\linewidth}*{2}{>{\raggedright\arraybackslash}X}@{}}
\toprule
\textbf{Rule} & \textbf{Route for larger circuits} & \textbf{Limiting quantity} \\
\midrule
Input-state invariance
& Use a proved state symmetry; for a pure $\tau_j$, state invariance is equivalent to $|\Tr(W_m\tau_j)|=1$.
& Access to that expectation or symmetry is required.
Invariance of a local reduced state alone does not establish invariance of an entangled register. \\
Cancellation, all-strength zeros, and optimal strength
& Obtain $(a_m,b_m)$ once, then evaluate each strength with constant scalar work.
The stabilizer theorem supplies polynomial exact expectation queries under its assumptions.
& Arbitrary state preparation remains expensive.
Numerical or sampled coefficients do not certify an exactly vanishing gap or an exact zero set. \\
Signed pilots and threshold intervals
& For $M_{\mathrm P}$ target strings of weight at most $w$, shared shadows use $\mathcal O(3^w\log(M_{\mathrm P}/\beta)/\zeta^2)$ shots at additive error $\zeta$ and failure budget $\beta$.
& Large Pauli weight, stringent precision, prefix execution cost, or truncation error can dominate.
The built-in simulated pilot also pays classical state-simulation costs. \\
TV and native-assertion rules
& Apply the scalar distance ceiling or acceptance-margin enclosure; restrict explicit distributions to the $r$ measured qubits when the oracle permits it.
& A complete distribution has $2^r$ entries.
Full-register distribution comparison can be exponential even when expectation tests are cheap; a native baseline margin must also be obtained. \\
Helstrom and restricted test augmentation
& Screen a declared finite list of measurable Paulis using the same propagation and expectation machinery.
& All Paulis number $4^n$; a general dense Helstrom construction uses $4^n$ matrix entries.
Pure-state rank-two constructions can use $2^n$ amplitudes instead, still exponential.
Measurement synthesis adds its own cost. \\
Finite-shot classification, ranking, and score bounds
& Once samples and certified indicators are supplied, threshold arithmetic and weighted counting are scalar per record; ranking $M$ candidates costs $\mathcal O(M\log M)$.
& Reading samples, enumerating candidates and tests, and obtaining their certificates must be charged separately.
Polynomial post-processing does not remove those costs. \\
\bottomrule
\end{tabularx}
\end{table}

There are three distinct ways to make state-dependent rules practical beyond small registers.
A supplied product or stabilizer description can make the required expectations tractable, with the explicit stabilizer bounds given in \cref{thm:stabilizer-visibility}.
Alternatively, the union of the original and mutant backward causal cones may contain only $r\ll n$ qubits.
When the corresponding reduced input is supplied, a pure state uses $2^r$ amplitudes and a mixed state uses a density matrix with $4^r$ entries.
For fixed-width gates and polynomial circuit length, $r=\mathcal O(\log n)$ then gives polynomial simulation work; a small mutation support $k$ alone does not imply a small causal cone or a cheaply available reduced input.
Finally, executing the relevant prefixes and measuring the required observables can replace classical state simulation, subject to the sampling, circuit-depth, and measurement costs in \cref{sec:measurement-baselines,sec:scalable}.
The present planner does not automatically choose a tableau backend or reduce state evolution to the causal cone.

State-invariance checks particularly require care about correlations.
For a Bell pair, the reduced state of either qubit is $I/2$ and commutes with a local $Z$, yet that operation changes the global pure state to an orthogonal Bell state.
Thus local reduced-state commutation cannot replace the global condition $[W_m,\tau_j]=0$ for a mixed or entangled register.
The pure-state expectation criterion in \cref{tab:scalable-state-rules} follows because two pure output states coincide up to phase exactly when their overlap has unit modulus.
It requires that overlap, rather than merely unchanged local marginals.

Sparse coefficient bounds also offer a scalable, potentially conservative perturbative route.
For complete Pauli expansions $H_m=\sum_a h_a A_a$ and $O_j=\sum_b c_bP_b$, define $h_1=\sum_a|h_a|$ and $o_1=\sum_b|c_b|$.
Submultiplicativity and the commutator inequality give
\begin{equation}
\label{eq:sparse-perturbative-remainder}
B_m=\frac{\eta^2}{2}\|[H_m,[H_m,O_j]]\|_\infty\leq 2\eta^2 h_1^2o_1.
\end{equation}
This replaces a dense norm calculation in \cref{eq:perturbative-gap-interval} by coefficient sums, but the state-dependent first-order expectation is still needed and the bound can be too loose to decide the threshold.
Using truncated expansions additionally requires allowances for the omitted terms.
The current general-generator fallback uses bounded dense products; the coefficient bound states an available extension rather than an implemented performance result.

Population-wide redundancy and subsumption add costs in the number of candidates and tests.
Identical compiled representations can be grouped by hashing followed by equality checks, whereas unrestricted pairwise equivalence checking can require $\mathcal O(M^2)$ proof attempts.
Once exact detection indicators for $T$ tests are supplied, a direct all-pairs subsumption comparison costs $\mathcal O(M^2T)$ Boolean operations; obtaining those indicators remains a separate analysis or execution cost.
Grouping by syntactic metadata can be inexpensive, but effect-based buckets first require the corresponding profiles and provide no semantic certificate.

Scalability therefore depends on register width, causal-cone size, gate locality, represented term count, state structure, Pauli weight, required precision, and the number of candidates and tests.
A Clifford suffix can keep one Pauli term compact while spreading its weight to $n$, making classical propagation cheap but the stated shadow bound exponential.
Conversely, a wide circuit with a small relevant cone or a tractable state can admit informative exact decisions well beyond the symbolic fallback limit.
A resource-limited or statistically unresolved case retains its population weight; inability to evaluate a rule is not evidence of equivalence or survival.

\subsection{Practitioner Checklist: Do, Do Not, and Report}
\label{sec:practitioner-checklist}

A planning decision must be reported together with its test, supporting evidence, and cost.
Apply the conditions in the order given by \cref{tab:planning-cost-order}, following the evidence rules in \cref{sec:claim-hierarchy}.
The implementation steps appear in \cref{alg:efficient-mutant-selection,alg:candidate-effect-evaluation}.
For each analysis, record:
\begin{itemize}
\item The declared semantics, mutant population, test, threshold, original and replacement operations, union of the qubits they act on, and left-multiplied relative unitary.
Keep a family's signed angle offsets distinct from a fixed replacement's canonical axis and angle; a scalar relative unitary has no defined axis.
\item The decision, evidence type, provenance, uncertainty, and computational limits.
For a fixed-axis analysis intended for reuse, retain $(a_m,b_m)$, the signed response, and its zero set.
\item The measurement implementation and complete preprocessing, pilot, and execution costs, using the alternatives in \cref{sec:measurement-baselines}.
Report population weights, including the weight of unresolved mutants, using \cref{thm:planner-soundness}.
Evaluate proposed test improvements on a population independent of the mutants used to construct them.
\end{itemize}
Use the threshold ceiling as an initial check before more expensive analysis.
For $R_O=2$ and $\gamma=0.1$, canonical angles below $2\arcsin(0.05)\approx0.1000$ radians cannot be detected by that test.
Pilot measurements need a separate cost and error analysis: their error allocations and retained terms must satisfy the sampling and truncation inequalities in \cref{sec:scalable}.
Compare pilot settings with direct execution and reconstruction baselines using total preprocessing, sampling, and execution costs at matched precision and confidence.

\section{Validation Requirements and Limitations}
\label{sec:core-validation}

\subsection{Evidence and Implementation Scope}
\label{sec:implementation-scope}

The evidence comprises proofs, hand-checkable examples, executable demonstrations, and bounded synthetic feasibility checks under the stated model; \cref{tab:evidence-implementation} records the implementation limits.
No corpus-effect results, filter hit rates, measured savings, or held-out adequacy gains are reported here.

During the preparation of this work, the authors used OpenAI Codex \cite{openaicodex} and Anthropic Claude Code \cite{anthropicclaudecode} to assist with manuscript editing and to obtain additional feedback on theorem statements and proofs.
These tools also supported the review and revision of the reference implementation, which was initially written by the authors, with human review of code produced or modified through this process.
Throughout, AI feedback served as an aid to review and revision, rather than as evidence of mathematical or implementation correctness.
The authors retain responsibility for the final manuscript, proofs, and implementation.

\begin{table}[!htb]
\caption{Scope of mathematical results and reference implementation.}
\label{tab:evidence-implementation}
\centering
\small
\renewcommand{\arraystretch}{1.12}
\begin{tabularx}{\linewidth}{@{}>{\raggedright\arraybackslash}p{0.20\linewidth}>{\raggedright\arraybackslash}X>{\raggedright\arraybackslash}X@{}}
\toprule
\textbf{Capability} & \textbf{Established or implemented} & \textbf{Conditions and implementation limits} \\
\midrule
Exact local analysis & Exact identities and structural rules; supported fixed-strength calculations and algebraic zero checks.
& Numerical gaps and bounds use ordinary floating point, without validated numerical threshold certification; reports concern fixed strengths rather than reusable response families. \\
Oracle extensions & TV rotation ceiling and native acceptance-margin and interval certificates.
& Mathematical results only; the implementation's decision rules remain differential expectation rules. \\
Symbolic certificates & Independent channel-equivalence and full-circuit zero proofs for supported explicitly exact data.
& Default symbolic limit: four changed qubits for local channel proofs, or four register qubits for full-circuit proofs; unsupported cases retain fallback provenance. \\
Tractable circuit family & Polynomial-time response classification under the stabilizer theorem's assumptions.
& Sparse Pauli preprocessing is implemented; supplied states and simulated pilots remain dense. The full tableau-based classification algorithm is not implemented. \\
Pilot estimation & Simulated mutation-site sampling with reported error allocation.
& Separate pilots per candidate estimate only the declared fixed-strength response. \\
Outcome accounting & Separate fixed-shot classifier; planning soundness proposition with score bounds; worked examples.
& Soundness covers classifications and score bounds under the proposition's assumptions; execution savings require separate evaluation. \\
\bottomrule
\end{tabularx}
\end{table}

\subsection{Bounded Synthetic Feasibility Probes}
\label{sec:synthetic-feasibility}

The implementation repository \cite{quantummutantplanner} contains \path{scripts/probe_scalability.py}, which exercises fixed constructed families with known answers.
The grid, environment, and per-case measurements from one batch are recorded in \path{scalability-results/reference.json}; \path{scalability-results/README.md} documents how to reproduce it.
These paths are relative to the repository root.
These checks measure implementation feasibility separately from the proposed population and execution-planning studies.
The measurements used an Apple M1 Max, Python 3.13.15, Qiskit 2.5.1, NumPy 2.5.2, and SymPy 1.14.0, with one requested linear-algebra thread and fresh worker processes.
Each completed case reports the median of three analysis-and-check timings, excluding imports and case construction; peak resident memory includes the entire worker.
Symbolic caches were cleared before each repetition.
Each worker had a 20-second total timeout and a polled resident set size (RSS) watchdog with a 768\,MiB stopping threshold.
\Cref{tab:scalability-probes} summarizes the 39 attempted cases: 33 completed, 4 crossed the RSS threshold, and 2 timed out.
A completed case finished all repetitions and result checks; the compiler filter can complete while declining to certify equivalence.
The structural, Clifford-propagation, local-certificate, and compiler-filter families share the register-width grid $n\in\{64,256,1024,4096,8192\}$.
The non-Clifford family instead fixes $n=64$ and varies the independent-rotation count $r\in\{4,8,12,24\}$; the statevector and dense families retain smaller width grids because they allocate exponentially sized representations.
These grids measure probe coverage, not maximum supported widths, and timings on different tasks do not establish execution savings.

The circuit families are deliberately simple, with qubits numbered from $1$ to $n$.
Let $C_n$ denote the ordered chain of $n-1$ gates $\mathrm{CX}_{1,2},\mathrm{CX}_{2,3},\ldots,\mathrm{CX}_{n-1,n}$, listed in execution order.
The structural-cone and Clifford-propagation rows use $C_n$ with output observable $Z_n$, whose backward image is $Z_1\cdots Z_n$.
The local-proof, compiler-filter, and dense-symbolic rows compare the same equivalent pair: a common prefix consisting of $C_n$ followed by $H_1$, then either $\mathrm{SWAP}_{1,n}$ or its three-gate implementation $\mathrm{CX}_{1,n},\mathrm{CX}_{n,1},\mathrm{CX}_{1,n}$.
These three rows time different computations on that pair: a proof after cancelling the common context, the compiler-plus-proof filter, and a full-register symbolic comparison without context cancellation.

The non-Clifford family applies $R_z(\pi/4)$ independently to the first $r$ qubits of a 64-qubit register and propagates $X_1\cdots X_r$ backward, with identity on the remaining qubits.
Its exact expansion contains $2^r$ terms, so this sweep measures term growth at fixed register width.
The statevector-planner row prepares $\ket{\mathrm{GHZ}_n}$ from $\ket{0}^{\otimes n}$ using $H_1$ followed by $C_n$, inserts $R_z(0.4)$ on qubit 1 after preparation, and evaluates the mutant against $X_1\cdots X_n$.
The dense-numerical row uses only three gates, $H_1$, $R_z(0.31)$ on qubit 1, and $\mathrm{CX}_{1,n}$ in that order, with output observable $Z_n$; it materializes the full-register matrix of the backward image $X_1Z_n$.
Thus this last row exercises dense matrix growth even though only two qubits participate in gates.

\begin{table}[!htbp]
\caption{All declared synthetic feasibility cases, grouped by computation and outcome.
$n$ counts register qubits, including idle qubits; $r$ counts independent rotations in the non-Clifford sweep, and a dash in these columns means that parameter does not apply.
Bold $n$ (or $r$ for the rotation sweep) identifies the completed case supplying the median time in seconds and peak worker memory in MiB.
For stopped cases, MiB is the largest sampled RSS rather than a completed worker's peak; these cases have no completed analysis timing.}
\label{tab:scalability-probes}
\centering\small
\renewcommand{\arraystretch}{1.12}
\begin{tabularx}{\linewidth}{@{}>{\raggedright\arraybackslash}p{0.24\linewidth}>{\raggedright\arraybackslash}Xllrr@{}}
\toprule
\textbf{Computation} & $n$ & $r$ & \textbf{Outcome} & \textbf{Seconds} & \textbf{MiB} \\
\midrule
Structural cone & 64, 256, 1024, \textbf{4096} & --- & Completed & $1.26\times10^{-1}$ & 676.0 \\
 & 8192 & --- & Memory stop & --- & 1028.8 \\
\addlinespace
Clifford Pauli propagation & 64, 256, 1024, 4096, \textbf{8192} & --- & Completed & $3.55\times10^{+0}$ & 124.5 \\
\addlinespace
Exact local channel proof & 64, 256, 1024, 4096, \textbf{8192} & --- & Completed & $3.30\times10^{-2}$ & 126.9 \\
\addlinespace
Compiler filter & 64, 256, 1024, 4096, \textbf{8192} & --- & Completed & $1.61\times10^{-1}$ & 213.9 \\
\addlinespace
Non-Clifford propagation & 64 & 4, 8, 12, \textbf{24} & Completed & $1.63\times10^{-1}$ & 113.0 \\
\addlinespace
Statevector planner & 12, 16, 20, \textbf{22} & --- & Completed & $5.57\times10^{-1}$ & 367.7 \\
 & 24 & --- & Memory stop & --- & 879.0 \\
\addlinespace
Dense numerical observable & 4, 6, 8, \textbf{10} & --- & Completed & $2.36\times10^{-1}$ & 285.3 \\
 & 12 & --- & Memory stop & --- & 1263.5 \\
\addlinespace
Dense symbolic channel proof & 4, \textbf{6} & --- & Completed & $1.21\times10^{+0}$ & 112.2 \\
 & 8 & --- & Timed out & --- & 116.5 \\
 & 10 & --- & Timed out & --- & 199.3 \\
 & 12 & --- & Memory stop & --- & 1138.2 \\
\bottomrule
\end{tabularx}
\end{table}

The 8,192-qubit Pauli and local-proof cases invoke sparse library helpers without a state input; they do not demonstrate an end-to-end planner accepting such a register.
The 22-qubit planner case supplies a dense basis-state vector, raises the state limit explicitly, disables equivalence filtering and pilots, and computes one known GHZ response with dense operator fallback disabled.

The stored-cone computation completed through 4,096 qubits with a peak of 676.0\,MiB.
The four memory-stopped cases exceeded the 768\,MiB threshold, with sampled RSS values of 879.0--1,263.5\,MiB.
The watchdog polls approximately every 0.2 seconds, allowing overshoot; sampled maxima are lower bounds on the worker peaks.
These observations establish threshold crossings, not operating-system out-of-memory failures or the memory required to finish.
Full-register symbolic comparisons completed at six qubits after raising the default limit; eight- and ten-qubit workers exceeded the total budget, which includes setup and all three repetitions rather than timing one call alone.

The compiler filter completed at every shared width through 8,192 qubits but declined the known-equivalent SWAP/three-CX pair even though the separate local proof succeeded, illustrating the compiled-match precondition's incompleteness.
For independent non-Clifford rotations, a 256-term cap kept the computation small while discarding enough coefficient mass to make the general gap-error allowance uninformative for the selected Pauli test.
The calculated shadow budget likewise rises from 13,248 shots at weight one to 357,696 at weight four for 256 targets, $\zeta=0.25$, and per-call failure probability $0.05$; these are formula evaluations, not collected measurements.
Together, these checks distinguish compact representation, successful certification, useful precision, and end-to-end state access; none can be inferred from register width alone.

\subsection{Choosing Rules for a Circuit}
\label{sec:rule-selection}

Practical use does not require every analysis rule to run on every circuit.
The applicable subset depends on the oracle, gate structure, mutation support, observable representation, available state representation, and resource budget.
The complexity bounds in \cref{sec:rule-scalability} identify rules whose computations remain tractable under suitable structure: scalar threshold ceilings, structural dependency checks, sparse Clifford propagation, and exact comparisons on small changed supports.
The completed cases in \cref{tab:scalability-probes}, including Pauli propagation and local channel proofs at 8,192 qubits, demonstrate that several implemented computations can operate well beyond the default four-qubit symbolic limit.
Their usefulness on a particular circuit still depends on the stated structural assumptions and on whether they supply a decisive certificate or bound.

A planner should omit full-register dense or symbolic calculations when their exponentially sized representations exceed its time or memory budget.
For a non-Clifford suffix, it should bound Pauli-term growth and retain the discarded coefficient mass in the error allowance; an allowance too large to resolve the threshold cannot justify a decision.
Pilots may likewise be disabled when observable weight makes their shot requirement exceed the budget.
State-dependent rules require an affordable way to obtain the needed expectations; a Clifford suffix alone does not make an arbitrary input state tractable.
Compiler filtering is also optional when its cost cannot be justified by reuse across candidates or tests.

The reference implementation provides equivalence and pilot switches, width limits for local, dense, symbolic, and state calculations, and a retained-term cap; it does not automatically choose these settings from circuit characteristics.
An evaluation should fix the selection policy and resource limits before inspecting mutation effects and report which analyses were omitted or stopped and why.
Disabling or exhausting a check supplies no evidence of equivalence or test survival.
The planner must continue with the remaining applicable rules and, if none decides the outcome, retain the candidate's population weight as unresolved under \cref{thm:planner-soundness}.

\subsection{Empirical Validation Still Required}
\label{sec:required-validation}

Empirical evaluation must distinguish how often the survival mechanisms occur from whether using them improves testing.
An exact census would fix the programs, mutant populations, test suites, and classification rules before inspecting mutation effects, then report how often each mechanism explains survival in that population.
Results must be reported separately for each combination of mutant population and test suite, keeping tool-generated mutants distinct from author-defined mutants and native assertions distinct from companion differential tests.
Algebraically certified zeros, high-precision numerical compatibility with zero, and tolerance-based zero classifications must remain distinguishable.
Such a census describes the declared population; execution savings and improved test adequacy require separate evaluations.
Execution-planning and pilot-estimation benefits require comparisons under a common resource budget that includes preprocessing, measurement, and execution costs.
Test augmentation must be evaluated on mutants not used to construct the added tests.

\subsection{Model and Interpretation Limits}
\label{sec:interpretation-limits}

Seven boundaries govern interpretation.
First, the formal results assume finite-dimensional noiseless unitary circuits and do not cover mid-circuit measurement, reset, leakage, or device noise.
Second, $\Delta_m$ measures a difference of means; for observables with more than two outcomes, equal means do not imply equal outcome distributions.
The population threshold rule concerns reference-based differential tests, and neither a state difference outside the permitted interface nor a gap between two accepted values establishes a native assertion failure.
Third, for generators not covered by an exact structured identity, first-order sensitivity does not predict the leading nonzero effect when its coefficient vanishes.
Fourth, transpiler equality is sound only under the stated semantic-preservation and reproducibility assumptions, while approximate redundancy and bucket selection can omit distinguishable behavior.
Fifth, the formal decomposition represents a single gate or subcircuit replacement and does not by itself analyze interacting higher-order mutants or the mapping from one source edit to several transpiled gate changes.
Sixth, Euler-coordinate offsets are representation dependent, and a simultaneous parameter path need not retain a fixed effective axis; consequently, the fixed-axis zero-set and maximizing-angle formulas apply to each reduced replacement, and applying them to an entire parameter path requires a fixed-axis family.
Seventh, the trace-norm-optimal observable is an information-theoretic benchmark unless its diagonalizing measurement circuit is available and its depth, routing, calibration, and noise costs are included.

Ideal-equivalence certificates do not establish hardware indistinguishability.
Hardware evaluation additionally requires a noise and drift model and measured calibration, routing, and mitigation costs; the noiseless analysis supplies none of these.

\section{Related Work}
\label{sec:related-work}

Prior work supplies mutation generators, test oracles, equivalence checks, and methods for computing quantum responses.
\Cref{tab:contribution-delta} compares the closest approaches with our decisions about mutant visibility, execution, and score accounting.

\subsection{Mutation Generation and Observed Survival}

Classical mutation testing uses seeded program variants to assess and improve test suites, but equivalent mutants, redundant mutants, cost, and the relationship between mutants and real faults are longstanding validity concerns \cite{jia2011survey,papadakis2019advances,just2014mutants}.
Mutation-based assessments can also produce invalid conclusions when their mutant dependencies and selection choices are not controlled \cite{papadakis2016threats}.

Quantum mutation tools generate gate insertions, deletions, and replacements, with some also changing parameters, qubit locations, controls, initialization, or measurements \cite{fortunato2022qmutpy_tool,fortunato2022mutation_qiskit,mendiluze2021muskit,garcia2024qcrmut,polo2025mutation_based,shao2026quanforge}.
Our earlier unit-testing study uses small-angle rotational mutants and compares several test oracles \cite{miranskyy2025feasibility}.
\Cref{app:mutation-tool-coverage} records the individual operator families, test rules, and applicability of our local-unitary analysis; the fraction of generated mutants covered remains unmeasured.

QMutBench and the empirical mutant-survivability study provide candidate validation corpora \cite{mendiluze2026qmutbench,mendiluze2025empirical}, while Bugs4Q and studies of quantum-platform bugs provide evidence for investigating fault plausibility \cite{zhao2023bugs4q,paltenghi2022platformbugs}.

\subsection{Equivalence and Test-Relative Detection}

Classical mutation research already distinguishes equivalent mutants from \emph{stubborn} non-equivalent mutants that survive reasonably thorough testing \cite{yao2014stubborn}.
Our contribution is a circuit-level explanation of some such survivors, including changes hidden by the input or measurement and cancellations at a particular strength; these mechanisms need not be mutually exclusive.
Coverage-based work uses changes in execution coverage to help classify equivalent mutants \cite{schuler2013equivalent}.
That empirical indication differs from a proof of semantic equivalence or an exact zero-effect certificate for a specified quantum test.
The structural backward light cone also follows the program-slicing idea of tracing the dependencies of selected outputs \cite{weiser1984slicing}; exact operator and state calculations refine that dependency test.

Compiler-based equivalent-mutant detection identifies some equivalent or duplicate mutants through optimized representations under explicit compiler assumptions \cite{papadakis2015tce}.
Quantum transpiler-based equivalence, decision diagrams, and ZX-calculus supply complementary semantic checks \cite{campos2026tbe,viamontes2003improving,miller2006qmdd,coecke2008interacting,coecke2011interacting,burgholzer2021advanced,peham2022zx}.
ZX-calculus represents quantum operations as diagrams and uses semantics-preserving rewrite rules to simplify or compare them.
Our analysis also asks why a mutant that changes program behavior can remain invisible to a particular input and measurement.
Information-theoretic testability research connects such fault masking and failed propagation to information loss \cite{clark2012squeeziness,androutsopoulos2014conditional,patel2022testability}.
Dynamic subsumption and dominator-mutant research compare detection across a test collection, which cannot be replaced by a single scalar mutant score \cite{ammann2014minimal,kurtz2016dominator}.

Predictive mutation testing learns from mutant and test features to predict kill or survival without executing each new mutant, trading prediction error against execution cost \cite{zhang2016predictive}.
The present planner addresses a related cost question through exact certificates or bounds with stated simultaneous coverage, retaining unresolved outcomes when neither decision is supported.
Its accounting guarantee therefore depends on valid evidence for the declared oracle, rather than on a classifier's empirical prediction accuracy.

Mutation-based search, QuSBT, and Quito generate tests using mutation, search, or coverage objectives, while QuraTest varies input magnitude, phase, and entanglement \cite{wang2022muttg,wang2022qusbt,wang2021quito,ye2023quratest}.
These methods complement our diagnosis of whether a surviving mutant requires a different input or measurement.
Statistical and projection-based assertions, property-based testing, and metamorphic testing provide other oracle contracts \cite{huang2019statistical,li2020proq,honarvar2020property,paltenghi2023morphq}.
QDiff uses differential comparisons and semantics-preserving transformations to test quantum software stacks \cite{wang2021qdiff}.
Our reference-based oracle measures the effect of a declared program mutation; a difference from the original does not by itself violate a native acceptance assertion.

\subsection{Established Mathematical Methods}

The analysis uses standard results in trace duality, circuit propagation, stabilizer simulation, and quantum statistical decision theory \cite{watrous2018theory,gottesman1999heisenberg,helstrom1976quantum,holevo1973statistical}.
Two-eigenvalue response identities, parameter-shift rules, and Fourier reconstruction supply established methods for differentiating or reconstructing expectations \cite{mitarai2018learning,schuld2019evaluating,wierichs2022general}.
Three-evaluation sinusoidal reconstruction also supports response reuse within a fixed-axis family and provides a competing strategy for estimating mutation effects, as reflected in \cref{tab:contribution-delta}.
Our specializations use these responses to identify mutation strengths with zero or maximal gap for a declared test, while \cref{thm:planner-soundness} composes semantic exclusions, mutant--test evidence, and population weights into a mutation-score guarantee.
Randomized benchmarking and gate-set tomography characterize implemented hardware errors \cite{magesan2011rb,nielsen2021gst}; our visibility analysis concerns ideal circuits within the broader quantum-software-testing problem \cite{ali2022qse}.
Randomized benchmarking estimates average gate errors from the decay of success probabilities as randomized gate sequences grow longer.
Gate-set tomography jointly estimates models of the gates, state preparation, and measurement from experimental data, without assuming perfectly calibrated preparation or readout.

\begin{table}[!htbp]
\caption{Established methods and their testing-specific use in this framework.}
\label{tab:contribution-delta}
\centering\small
\renewcommand{\arraystretch}{1.12}
\begin{tabularx}{\linewidth}{@{}>{\raggedright\arraybackslash}p{0.20\linewidth}*{2}{>{\raggedright\arraybackslash}X}@{}}
\toprule
\textbf{Established method} & \textbf{Standard result or prior capability} & \textbf{Testing-specific use} \\
\midrule
Compiler comparison: transpiler-based equivalence (TBE) \cite{campos2026tbe}
& Certifies some mutants as equivalent when a semantics-preserving transpiler produces identical compiled circuits under fixed settings.
& Analyzes visibility to a specified test after equivalence filtering; retains unresolved cases and population weights in mutation-score bounds. \\
QUMUG \cite{mendiluze2026qumug}
& Searches for gate mutations that survive a supplied test suite.
& Identifies exactly which strengths in a declared fixed-axis family give zero gap, including cancellation; fixes the adequacy population independently of the suite. \\
Quantum Squeeziness \cite{bensoussan2026squeeziness}
& Characterizes information loss under state perturbations and reports an empirical association with mutation scores.
& Computes the signed expectation change for a specified mutant and test, distinguishes causes of zero gap, and applies the detection threshold with explicit evidence requirements. \\
Rotosolve / NFT \cite{ostaszewski2021structure,nakanishi2020sequential}
& Reconstructs and optimizes a sinusoidal expectation using three evaluations.
& Uses the same response to find every strength that preserves the original expectation and the largest absolute change; keeps the adequacy population fixed. \\
Trace distance and Helstrom discrimination \cite{watrous2018theory}
& Bound changes under measurement and characterize optimal state discrimination.
& Give angle-only expectation and TV skip rules, and an unrestricted measurement benchmark for a fixed input. \\
Stabilizer simulation \cite{aaronson2004stabilizer}
& Computes Pauli propagation and stabilizer-state expectations efficiently.
& Classifies the three possible gap curves under the stated Pauli and Clifford assumptions. \\
Simultaneous coverage and weighted counting
& A union bound controls failures; valid indicators bound a weighted total.
& Compose oracle-specific certificates into a mutation-score interval without dropping survivors or unresolved weight. \\
\bottomrule
\end{tabularx}
\end{table}

\FloatBarrier
\section{Conclusion}
\label{sec:conclusion}

The engineering problem addressed here is how to act on a surviving quantum mutant without confusing program equivalence, test-relative invisibility, and insufficient evidence.
The three objectives stated in the introduction connect this problem to concrete testing decisions.

We addressed the first objective, diagnosing surviving mutants, by combining semantic checks, circuit-dependency analysis, and exact local-effect formulas for supported unitary replacements.
These checks distinguish semantic equivalence, input-dependent invisibility, measurement-dependent invisibility, and finite-strength cancellation, while keeping insufficient evidence separate from a certified zero effect.
For the supported fixed-axis rotations, the response classification identifies every zero-gap strength; under the stated stabilizer assumptions, the response curves can be classified in polynomial time.
The practical consequence is an explanation of survival that separates a proved zero gap from a nonzero effect below the detection threshold or a gap that remains unresolved.

We addressed the second objective, planning executions with sound reporting, by specifying which decisions certificates, bounds, and estimates justify.
Under its stated assumptions, the planning guarantee preserves resolved classifications and bounds the mutation score when analysis or execution is incomplete, retaining the population weights of survivors and unresolved mutants.
Skipping an execution therefore does not remove a non-equivalent mutant from adequacy accounting.
The resource comparison includes preprocessing, pilot measurements, and execution, together with simpler estimation alternatives, so an exact formula or reusable pilot does not by itself establish a saving.

We addressed the third objective, guiding test-oracle improvement, by using the state difference to determine whether any observable can expose a mutant on the fixed input.
The unrestricted optimal observable supplies a mathematical benchmark, while restricted measurement screens and explicit measurement circuits connect that benchmark to permitted tests and their implementation costs.
In the running sign-error example, an allowed basis change exposes a state difference that the original output-bit assertion cannot observe.
Under state-transforming semantics, the mutant remains a survivor in the original suite's adequacy denominator; under a fixed computational-basis interface on that input, the programs are observationally equivalent and the added measurement changes the contract being assessed.
Test improvement must therefore respect the declared interface, meet the chosen detection rule, and be evaluated independently of the mutants used to construct it.

Together, these results provide a basis for deciding what to execute, what to change in a test, and what remains unknown.
The contribution is the testing-specific analysis and composed guarantees built from established mathematical tools, supported by proofs, worked examples, and a reference implementation under ideal unitary semantics.
How often the mechanisms explain survival, whether planning reduces total cost, and whether augmentation improves adequacy on independent mutants remain questions for empirical validation.

\section*{Acknowledgments}

This work was supported by the Natural Sciences and Engineering Research Council of Canada (NSERC) under grants RGPIN-2022-03886 and ALLRP 613846-25.
It was also funded by the FCT~-~Fundação para a Ciência e a Tecnologia, I.P., under the projects ref.\ 2025.00002.QUANTUMVLAB.DEUCALION and 2025.08074.CPCA.A1, and under the LASIGE Research Unit, ref.\ UID/00408/2025, DOI \href{https://doi.org/10.54499/UID/00408/2025}{10.54499/UID/00408/2025}.

\printbibliography

\clearpage
\appendix

\crefalias{section}{appendix}
\crefalias{subsection}{appendix}
\crefalias{subsubsection}{appendix}
\crefname{appendix}{Appendix}{Appendices}
\Crefname{appendix}{Appendix}{Appendices}

\section{Detailed Statistical and Local Framework}
\label{app:detailed-framework}
This appendix uses the circuit, oracle, and local mutation notation of \cref{sec:framework}.
It supplies the statistical bounds, local derivations, and implementation algorithms supporting the main text.

\subsection{Planning Soundness: Assumptions and Proof}
\label{sec:planner-soundness-proof}

\Cref{thm:planner-soundness} composes valid evidence into population-level classifications and score bounds under the following conditions.
\begin{enumerate}[label=(\roman*)]
\item Semantic exclusions are sound for the declared state-transforming semantics on $\mathcal D_{\mathrm{in}}$, use a prespecified policy, and leave the target population $\mathcal M^\star$ consisting of non-equivalent mutants.
Any transpiler certificate therefore requires semantic preservation and equality of the complete compiled representation under the pinned settings.
\item Every certificate proving an exactly zero gap is valid, and every deterministic gap interval contains the true gap, including propagation, perturbative-remainder, and numerical errors.
Ordinary floating-point evaluation of an ideal formula alone does not meet this condition at an uncertified threshold comparison.
\item Pilot target sets and error allocations are prespecified with total failure probability at most $\alpha_{\mathrm{pilot}}$ across all sites, tests, and candidates.
Reused samples cover the union of their declared targets simultaneously.
Observed-test intervals have total failure probability at most $\alpha_{\mathrm{test}}$, using fixed allocations and sample sizes or a separately valid sequential procedure.
Data-dependent execution choices preserve these guarantees, for example by using fresh observed samples with conditional coverage given the planning data.
\item Representative classes, if used, have certified identical population detection indicators for every test in $\mathcal T$, and retain their original positive integer multiplicities $w_c$.
Singleton classes are always permitted.
For example, semantics-preserving pinned transpilation can certify equality between two mutants, which implies equal detection indicators across the suite.
The reference implementation does not construct such mutant--mutant classes; approximate buckets do not establish this condition.
\end{enumerate}
Apply the threshold rule only to evidence meeting these conditions.
An exact gap supplies equal lower and upper bounds; an exact zero is a below-threshold certificate because the threshold is positive.
For statistical intervals, the certified classification carries the stated simultaneous confidence guarantee.

\begin{proof}[Proof of \Cref{thm:planner-soundness}]
Take the intersection of every declared pilot-coverage and observed-test-coverage event.
The union bound gives probability at least $1-\alpha_{\mathrm{pilot}}-\alpha_{\mathrm{test}}$ without requiring independence between intervals.
Conditional coverage for fresh observed samples implies the required unconditional bound even when the planner chooses which declared comparisons to execute.
On this event all statistical and deterministic intervals contain their true gaps, so an upper endpoint below $\gamma_r$ implies $\Delta_m(\rho_r,O_r)<\gamma_r$, while a lower endpoint at least $\gamma_r$ implies detection.
The same implications hold for exact gaps and certified zeros.
Thus one detected pair establishes a detected mutant, whereas survival requires a below-threshold conclusion for every test.
A budget certificate concerns independent-copy discrimination and supplies neither implication.
Sound semantic exclusions remove only mutants outside the declared non-equivalent population, and certified class indicators permit expansion back to every original member using $w_c$.
The detected classes therefore contribute exactly $W_{\mathrm D}$ detections, survivor classes contribute zero, and unresolved classes contribute between zero and $W_{\mathrm U}$.
Dividing by $W$ proves \cref{eq:planner-score-bounds}.
\end{proof}

The guarantee requires neither execution completeness nor optimal scheduling, and specifies no particular observed kill count.
Equality of detection indicators is sufficient for the population score; transferring full outcome distributions or sample-level behavior requires the stronger corresponding equivalence.
If semantic equivalence remains unresolved, one may report bounds for the retained candidate population, but cannot silently identify its denominator with the fully non-equivalent population in \cref{eq:suite-mutation-score}.
An empty eligible population has no defined mutation score under that equation.

\subsection{Hoeffding-Based Estimation Efficiency}
\label{sec:estimation-efficiency}

In the noiseless expectation-estimation protocol considered here, the statistical uncertainty comes from finite sampling.
Write $R_O=\lambda_{\max}(O)-\lambda_{\min}(O)\geq0$ for the spectral range of $O$.
If $R_O=0$, then $O$ is proportional to the identity on the modeled register, $\Delta_m=0$ for every unitary mutant, and the observable should be discarded as non-informative.
Suppose that the original and mutant means are estimated from $n_{\mathrm{shots}}$ shots of each circuit, with all $2n_{\mathrm{shots}}$ observations independent within this comparison.
Write $d_m=\mu-\mu_m$ and $\widehat d_m=\widehat\mu-\widehat\mu_m$.
The signed estimator is a sum of $2n_{\mathrm{shots}}$ independent scaled observations, each of range $R_O/n_{\mathrm{shots}}$, so the sum of squared ranges is $2R_O^2/n_{\mathrm{shots}}$.
Applying Hoeffding's bounded-variable inequality directly to this sum and combining its two tails gives \cite[Theorem~2 and Eq.~(2.6)]{hoeffding1963probability}
\begin{equation}
\label{eq:two-sample-hoeffding}
\Pr\!\left(|\widehat d_m-d_m|\geq t\right)\leq2\exp\!\left(-\frac{n_{\mathrm{shots}}t^2}{R_O^2}\right),\qquad t>0.
\end{equation}
For a fixed-separation decision problem in which the alternative satisfies $|d_m|\geq\Delta_{\min}>0$, taking $t=\Delta_{\min}/2$ shows that
\begin{equation*}
n_{\mathrm{shots}}\geq\frac{4R_O^2}{\Delta_{\min}^2}\log\!\left(\frac{2}{\alpha_{\mathrm{test},m}}\right)
\end{equation*}
shots per circuit suffice to separate equality from that alternative with error probability at most $\alpha_{\mathrm{test},m}$ by thresholding $|\widehat d_m|$ at $\Delta_{\min}/2$.
This distribution-independent upper bound gives a sufficient shot count; establishing a lower bound for particular distributions requires a separate argument.
Variance-aware estimators or direct hypothesis tests can require fewer samples, especially when one outcome distribution is nearly deterministic.

\subsubsection{Classification at a Practical Threshold}
\label{sec:finite-shot-threshold}
For $n_{\mathrm{shots}}$ independent shots of each circuit and prespecified failure probability $\alpha_{\mathrm{test},m}\in(0,1)$, put $\widehat\Delta_m=|\widehat\mu-\widehat\mu_m|$ and
\begin{equation}
\label{eq:finite-shot-gap-interval}
r_m=R_O\sqrt{\frac{\log(2/\alpha_{\mathrm{test},m})}{n_{\mathrm{shots}}}},
\qquad
I_m^{\mathrm{stat}}=[\max(0,\widehat\Delta_m-r_m),\;\min(R_O,\widehat\Delta_m+r_m)].
\end{equation}
\Cref{eq:two-sample-hoeffding} and the reverse triangle inequality give $\Pr(\Delta_m\in I_m^{\mathrm{stat}})\geq1-\alpha_{\mathrm{test},m}$.
Declare below threshold if the upper endpoint is strictly below $\gamma$, detected if the lower endpoint reaches $\gamma$, and unresolved otherwise.
This is a fixed-shot decision; repeatedly checking the same fixed-time interval requires a separate sequential-validity argument or prespecified error spending.
Away from the boundary, sufficient classification cost scales with $R_O^2/|\Delta_m-\gamma|^2$; the $R_O^2/\Delta_m^2$ scaling concerns resolution relative to zero.
There is no finite fixed-shot procedure with uniformly nontrivial error control that resolves every gap arbitrarily close to both sides of $\gamma$.
Consequently, detecting a prespecified difference from zero, classifying a practical threshold, and ranking by $E_m$ remain different tasks.
For multiple decisions, choose allocations with $\sum_m\alpha_{\mathrm{test},m}\leq\alpha_{\mathrm{test}}$; the union bound does not require independence between the resulting decisions.
Original-circuit samples may consequently be reused across mutants, provided that the samples within each original--mutant comparison remain independent and the target comparisons and allocations are prespecified.
The reference function \texttt{classify\_fixed\_shot} implements this rule from equal-length outcome samples and known outcome bounds; its observed classifications are separate from planner statuses.

Let $C(\rho,O,m)>0$ be the cost of one paired comparison, meaning one independent sample from the original circuit and one from mutant $m$, including state preparation, circuit execution, and measurement overhead.
One concrete additive baseline is
\begin{equation}
\label{eq:paired-cost-model}
C(\rho,O,m)
=
c_{\mathrm{prep}}^{(0)}+c_{\mathrm{exec}}^{(0)}+c_{\mathrm{meas}}^{(0)}
+c_{\mathrm{prep}}^{(m)}+c_{\mathrm{exec}}^{(m)}+c_{\mathrm{meas}}^{(m)}.
\end{equation}
The examples below set these components equal across mutants, whereas an empirical cost-efficiency claim must measure or otherwise prespecify their variation.
For equal shot counts and a fixed confidence level, define the optional estimation-efficiency score
\begin{equation}
\label{eq:core-efficiency}
E_m(\rho,O)=\frac{\Delta_m(\rho,O)^2}{R_O^2C(\rho,O,m)}=\frac{\delta_m^2}{C(\rho,O,m)},
\qquad \delta_m=\Delta_m/R_O.
\end{equation}
This score ranks signal relative to a sufficient mean-estimation cost; it does not decide a practical detection threshold or measure suite adequacy.
The range normalization is invariant under observable rescaling; unequal shot allocation or reused reference estimates require corresponding cost accounting.

The quadratic dependence of $E_m$ on the gap follows from applying a uniform additive-error estimation bound for bounded observables.
Identifying the detectable signal size with
$$
\epsilon \sim \Delta_m(\rho,O),
$$
gives the sufficient worst-case shot bound
$$
N_m^{\mathrm{mean}}(\rho,O;\alpha_{\mathrm{test},m})
=
\mathcal{O}\!\left(
\frac{R_O^2\log(2/\alpha_{\mathrm{test},m})}{\Delta_m(\rho,O)^2}
\right),
$$
where $N_m^{\mathrm{mean}}$ denotes the number of paired comparisons used to resolve a prespecified nonzero separation by this variance-blind protocol.
The substitution $\Delta_{\min}\leftarrow\Delta_m$ changes the decision problem from a common prespecified alternative to a mutant-specific retrospective resolution calculation.
Therefore, comparing $E_m$ across mutants compares their signal-to-cost values under one protocol; comparing the power of identical tests against one fixed alternative requires retaining a common threshold $\Delta_{\min}$.
When the scientific question is whether every mutant clears the same practical detection boundary, the analysis should retain that common boundary and report power or confidence relative to it.

The per-decision error level also does not automatically extend to a collection of tests and mutants.
For $N_{\mathrm{test}}$ tests and $M$ mutants, the union bound controls family-wise error at $\alpha_{\mathrm{test}}$ if each of the $N_{\mathrm{test}}M$ decisions is controlled at $\alpha_{\mathrm{test}}/(N_{\mathrm{test}}M)$, regardless of dependence.
Reusing original-circuit shots induces correlation among comparisons, so less conservative simultaneous inference must model that shared randomness explicitly.

Thus, $E_m$ is a standardized worst-case ranking score for expectation estimation; optimal hypothesis-testing costs depend on the outcome regime.\footnote{
This observable-gap scaling should be distinguished from fidelity-based state-comparison tests, such as the inverse or swap test.
For pure-state comparisons, the inverse test reverses the reference preparation on the tested output and checks for return to the initial state; the swap test uses an auxiliary control qubit and controlled exchanges of two registers to estimate their state overlap.
Here $F$ denotes squared fidelity, a similarity measure ranging from zero for states with orthogonal supports to one for identical states; for pure states it is the squared magnitude of their inner product.
When the comparison has a nonzero failure probability, the shots needed to observe that failure with constant confidence scale as $\mathcal{O}(1/(1-F))$ near $F=1$ \cite{miranskyy2025feasibility,miranskyy2025cost}.
The two views complement one another by quantifying distinct aspects of state and observable differences.
For a small coherent perturbation with strength $\eta$, the leading nonzero infidelity is typically quadratic: $1-F = \eta^2\,\mathrm{Var}_\psi(H) + \mathcal{O}(\eta^4)$ for pure states under unitary perturbations, unless the variance term vanishes.
Indeed, expanding $F=|\bra{\psi}e^{-i\eta H}\ket{\psi}|^2$ gives the quadratic coefficient $\bra{\psi}H^2\ket{\psi}-(\bra{\psi}H\ket{\psi})^2$ in $1-F$, and $F(-\eta)=F(\eta)$ removes odd powers.
When the chosen observable is first-order visible, $\Tr([H_m,O_j]\tau_j)\neq 0$, the gap is $\Theta(|\eta|)$, and the uniform additive-mean bound is $\mathcal{O}(1/\eta^2)$.
For a general generator, when the first-order coefficient vanishes and no exact structured identity is available, one can conclude from the perturbative coefficient only that $\Delta_m=\mathcal{O}(\!\eta^2)$; if the quadratic coefficient is nonzero, the same uniform bound becomes $\mathcal{O}(1/\eta^4)$.
For nontrivial involutory-axis rotations, \cref{cor:test-independent-rotation-ceiling} sharpens this comparison: an optimized bounded-observable test matches the independent-copy lower bound only in its small-angle exponent; the constants and decision objectives remain separate.
Outside that structured class, no unconditional asymptotic ordering follows from these bounds alone, because direct or variance-aware detection can be faster and because either the fidelity or observable coefficient may vanish.
The light-cone and commutator analysis below characterizes when first-order visibility holds.
}

\subsection{Exact Local Replacement Form}
\label{sec:exact-local-replacement}

Before specializing to small perturbations, an arbitrary unitary gate or subcircuit replacement can be expressed exactly at the mutation site.
\begin{proof}[Proof of \Cref{thm:core-exact-local}]
The argument specializes unitary composition and trace cyclicity to the mutation site.
Embedding both operations into the union of their supports gives $\widetilde G_j=W_mG_j$ and $U_m=U_{>j}W_mG_jU_{<j}$.
With the post-gate state $\tau_j$ and propagated observable $O_j$ defined in the main text,
\begin{equation*}
\mu=\Tr(O_j\tau_j),
\qquad
\mu_m=\Tr(O_jW_m\tau_jW_m^\dagger)=\Tr(W_m^\dagger O_jW_m\tau_j).
\end{equation*}
Their difference proves \cref{eq:core-exact-local} with the declared left-multiplication convention.
For the Hermitian operator $Q_m=O_j-W_m^\dagger O_jW_m$, Schatten norm duality bounds $|\Tr(Q_m\tau)|$ by $\|Q_m\|_\infty$ for every density operator \cite[Eqs.~(1.173)--(1.174)]{watrous2018theory}.
A pure eigenstate for an eigenvalue of largest absolute value attains this bound.
Thus every state gives zero gap exactly when $Q_m=0$, equivalently $[W_m,O_j]=0$.
For one fixed $\tau_j$, the trace can instead vanish by cancellation even when $Q_m\neq0$.
\end{proof}

\subsection{State-Side Dual and Trace-Distance Interpretation}
\label{sec:state-side-dual-proof}

To prove \cref{prop:state-side-dual}, cyclicity of the trace rewrites the exact local gap as $|\Tr(O_j\Delta\tau_j)|$.
The trace-norm variational characterization \cite[Eqs.~(1.173)--(1.174)]{watrous2018theory}, specialized to the Hermitian operator $\Delta\tau_j$, gives \cref{eq:state-side-trace-norm}; the Hermitian contraction $O_j^\star=\operatorname{sgn}(\Delta\tau_j)$ is an attaining choice.
Because $O_j=U_{>j}^\dagger O U_{>j}$ is a norm-preserving bijection on Hermitian observables, the site optimizer corresponds to the final observable $O^\star=U_{>j}O_j^\star U_{>j}^\dagger$.
The equality $\Delta\tau_j=0$ is equivalent to $W_m\tau_jW_m^\dagger=\tau_j$ and hence to $[W_m,\tau_j]=0$.

For the Pauli-rotation specialization, decomposing $\tau_j$ into its commuting and anticommuting parts gives \cref{eq:state-side-pauli-difference}.
The operator $-i\tau_{j,-}^{(A)}A$ is Hermitian, anticommutes with $\tau_{j,-}^{(A)}$, and has the same square.
Consequently,
\begin{equation*}
(\Delta\tau_j)^2
=
4\sin^2(\eta/2)\left(\tau_{j,-}^{(A)}\right)^2.
\end{equation*}
Thus, $|\Delta\tau_j|=2|\sin(\eta/2)|\,|\tau_{j,-}^{(A)}|$, and taking the trace proves \cref{eq:state-side-pauli-trace-norm}.

Let $\rho_{\mathrm{out}}=U\rho U^\dagger$ and $\rho_{\mathrm{out}}^{(m)}=U_m\rho U_m^\dagger$.
Unitary invariance under the suffix gives
\begin{equation}
\label{eq:mutation-trace-distance}
D_{\mathrm{tr}}\!\left(\rho_{\mathrm{out}},\rho_{\mathrm{out}}^{(m)}\right)
=
\frac{1}{2}\left\|\rho_{\mathrm{out}}-\rho_{\mathrm{out}}^{(m)}\right\|_1
=
\frac{1}{2}\|\Delta\tau_j\|_1.
\end{equation}
For any nonconstant Hermitian observable, set $c_O=(\lambda_{\max}(O)+\lambda_{\min}(O))/2$, so the centered and rescaled observable $2(O_j-c_OI)/R_O$ is feasible in \cref{eq:state-side-trace-norm}.
Because $\Tr(\Delta\tau_j)=0$, applying that variational characterization proves \cref{eq:normalized-gap-trace-distance}.
For nonzero $\Delta\tau_j$, the tracelessness of the state difference ensures that $O_j^\star=\operatorname{sgn}(\Delta\tau_j)$ has spectral range two and attains equality.
If $\Delta\tau_j=0$, any nonconstant Hermitian observable attains equality with both sides zero, without using the constant sign observable in a range-normalized ratio.
The corresponding equal-prior single-shot discrimination success probability is $(1+D_{\mathrm{tr}})/2$ by the Holevo--Helstrom theorem \cites[Sec.~2.1, especially Eqs.~(2.5)--(2.9)]{helstrom1971quantumdetection}[Theorem~3.4 and Eq.~(3.11)]{watrous2018theory}.
If $C(\rho,O,m)=C_m$ is held fixed across observables, the same optimum gives
\begin{equation}
\label{eq:maximum-efficiency-trace-distance}
\sup_{O:R_O>0}E_m(\rho,O)
=
\frac{D_{\mathrm{tr}}(\rho_{\mathrm{out}},\rho_{\mathrm{out}}^{(m)})^2}{C_m}.
\end{equation}

For the involutory-axis rotation in \cref{cor:rotation-state-side}, the anticommuting state component satisfies
\begin{equation*}
\left\|\tau_{j,-}^{(A)}\right\|_1
=
\frac{1}{2}\|\tau_j-A\tau_jA\|_1
\leq
1,
\end{equation*}
because $\tau_j$ and $A\tau_jA$ are density operators and the trace distance between two states is at most one.
Combining this inequality with \cref{eq:state-side-pauli-trace-norm,eq:mutation-trace-distance,eq:normalized-gap-trace-distance} proves \cref{eq:test-independent-rotation-ceiling}.
The bound is tight whenever both eigenspaces of $A$ are nonzero: choose eigenvectors $\ket{a_+}$ and $\ket{a_-}$ with eigenvalues $+1$ and $-1$, respectively, and set $\ket{\psi}=(\ket{a_+}+\ket{a_-})/\sqrt{2}$.
Then $\ket{\psi}$ and $A\ket{\psi}$ are orthogonal, so $\|\tau_j-A\tau_jA\|_1=2$ for $\tau_j=\ket{\psi}\bra{\psi}$, and the Helstrom observable attains the normalized-gap ceiling.
Because the prefix through $G_j$ is unitary, this mutation-site state is attainable by a circuit input.
Together with \cref{eq:maximum-efficiency-trace-distance}, it also proves the tight optional score ceiling $\sup_{\rho,O:R_O>0}E_m(\rho,O)=\sin^2(\eta/2)/C_m$ when the comparison cost $C_m$ is test independent.

Extending the one-copy trace-distance ceiling to the independent-copy lower bound requires an additional fidelity argument.
Use the squared fidelity
\begin{equation*}
F(\sigma,\omega)
=
\left\|\sqrt{\sigma}\sqrt{\omega}\right\|_1^2.
\end{equation*}
A \emph{purification} represents a possibly mixed state as part of a larger pure state: ignoring the added reference system recovers the original density operator.
For any purification $\ket{\Psi}$ of $\tau_j$, the state $(W_m\otimes I)\ket{\Psi}$ purifies $W_m\tau_jW_m^\dagger$.
Uhlmann's purification characterization of fidelity therefore gives \cite[Eqs.~(4) and (23)]{uhlmann1976transition}
\begin{align*}
F(\tau_j,W_m\tau_jW_m^\dagger)
&\geq
\left|\Tr(\tau_jW_m)\right|^2
\\
&=
\cos^2(\eta/2)
+
\sin^2(\eta/2)\left[\Tr(A\tau_j)\right]^2
\\
&\geq
\cos^2(\eta/2).
\end{align*}
The suffix is unitary, so the same lower bound holds for the two output states associated with any one circuit input.
This argument also holds for a joint mutation-site state $\tau_{j,S\mathsf R}$ and an untouched reference $\mathsf R$: replace $W_m$ by $W_m\otimes I_{\mathsf R}$ and use $\Tr[\tau_{j,S\mathsf R}(W_m\otimes I_{\mathsf R})]=\Tr(\tau_{j,S}W_m)$ in the same purification bound.
The one-copy trace-distance ceiling likewise applies on $S\mathsf R$, since $A\otimes I_{\mathsf R}$ is an involution.
References may thus be entangled with their own circuit input, while the joint input--reference states must remain products across executions.
For predetermined inputs indexed by $k=1,\ldots,N_{\mathrm{copy}}$, let $\rho_{\mathrm{out},k}$ and $\rho_{\mathrm{out},k}^{(m)}$ be the corresponding original and mutant output states.
Each pair obeys $F(\rho_{\mathrm{out},k},\rho_{\mathrm{out},k}^{(m)})\geq\cos^2(\eta/2)$ independently of the chosen input.
Fidelity is multiplicative over the resulting product states \cite[Eq.~(31)]{fuchs1999cryptographic}.
With the squared-fidelity convention above, the Fuchs--van de Graaf inequality gives \cite[Theorem~1, Eq.~(44)]{fuchs1999cryptographic}
\begin{equation}
\label{eq:independent-copy-trace-distance}
D_{\mathrm{tr}}\!\left(\bigotimes_{k=1}^{N_{\mathrm{copy}}}\rho_{\mathrm{out},k},\bigotimes_{k=1}^{N_{\mathrm{copy}}}\rho_{\mathrm{out},k}^{(m)}\right)
\leq
\sqrt{1-\cos^{2N_{\mathrm{copy}}}(\eta/2)}.
\end{equation}
If an equal-prior discrimination procedure has success probability at least $1/2+\epsilon$ for fixed $\epsilon\in(0,1/2)$, Helstrom's formula requires the left-hand side of \cref{eq:independent-copy-trace-distance} to be at least $2\epsilon$.
Consequently, for $0<|\eta|<\pi$,
\begin{equation}
\label{eq:independent-copy-lower-bound}
N_{\mathrm{copy}}
\geq
\frac{\log\!\left(1/(1-4\epsilon^2)\right)}
{-\log\!\left(\cos^2(\eta/2)\right)}
=
\Omega\!\left(\frac{1}{\sin^2(\eta/2)}\right)
\qquad
(\eta\to0).
\end{equation}
Equivalently, a budget of $N_{\mathrm{copy}}$ independent copies can attain advantage $\epsilon$ only if $|\sin(\eta/2)|\geq\sqrt{1-(1-4\epsilon^2)^{1/N_{\mathrm{copy}}}}$, which is a test-independent necessary affordability condition.
Setting $\epsilon=1/2-\delta_{\mathrm{disc}}$ gives $1-4\epsilon^2=4\delta_{\mathrm{disc}}(1-\delta_{\mathrm{disc}})$ and proves \cref{eq:practitioner-minimum-strength}; for fixed $\delta_{\mathrm{disc}}$, its leading behavior is $2\sqrt{\log(1/[4\delta_{\mathrm{disc}}(1-\delta_{\mathrm{disc}})])/N_{\mathrm{copy}}}+\mathcal{O}(N_{\mathrm{copy}}^{-3/2})$.
The tight input and Helstrom observable above give a bounded two-outcome expectation test with gap proportional to $|\sin(\eta/2)|$, so its Hoeffding upper bound matches this lower bound in its small-angle dependence.
Here ``predetermined'' means that the input to each execution is fixed independently of outcomes from earlier executions, which is exactly what makes the two hypotheses fixed product states; adaptive input selection is outside this argument.
With inputs entangled across parallel uses, the product-input argument no longer applies.
For an involutory axis with both eigenspaces nonzero and canonical angle $|\eta|\in(0,\pi]$, an entangled multi-use construction achieves perfect discrimination with $\lceil\pi/|\eta|\rceil$ channel uses, yielding $\Theta(1/|\eta|)$ scaling compared with the independent-execution $\Theta(1/\eta^2)$ scaling \cite[pp.~177901-3--177901-4, Eqs.~(16), (20), and (23)--(24), including the arbitrary-dimension discussion]{acin2001statistical}.

\subsubsection{Optional Unknown-Source Discrimination Policy}
\label{sec:optional-discrimination-policy}
Ordinary mutation testing executes known program versions and evaluates an assertion.
A different problem is to identify an unknown original or mutant source, given $N_{\mathrm{copy}}$ predetermined independent output copies and equal prior probabilities.
Under the assumptions of \cref{eq:independent-copy-lower-bound}, error at most $\delta_{\mathrm{disc}}\in(0,1/2)$ requires $|\sin(\eta/2)|\geq\sin(\eta_{\mathrm{IT}}/2)$, where
\begin{equation}
\label{eq:practitioner-minimum-strength}
\eta_{\mathrm{IT}}(N_{\mathrm{copy}},\delta_{\mathrm{disc}})=2\arcsin\!\sqrt{1-\left[4\delta_{\mathrm{disc}}(1-\delta_{\mathrm{disc}})\right]^{1/N_{\mathrm{copy}}}}.
\end{equation}
For $N_{\mathrm{copy}}=10^4$ and $\delta_{\mathrm{disc}}=0.05$, this gives approximately $0.0258$ radians.
This necessary condition neither decides whether a differential gap exceeds $\gamma$ nor rules out classical analysis.
The reference implementation disables this policy by default; \texttt{enable\_discrimination\_budget=True} inserts it after the rotation threshold ceiling and records a separate \emph{skip--below budget floor} disposition.
Its error target is separate from the pilot and observed-test failure allocations, and any skipped pair remains unresolved for population accounting unless other evidence settles it.

\subsection{Exact Structured Rotation and Clifford Forms}
\label{sec:exact-structured-mutations}

The exact local replacement form becomes especially explicit when the relative mutation acts structurally on a Pauli expansion.
First consider
$$
W_m
=
e^{-i\eta A/2},
$$
where $A=A^\dagger$ and $A^2=I$.
Define
\begin{equation*}
\begin{aligned}
O_{j,+}
&=
\frac{O_j+AO_jA}{2},
&
O_{j,-}
&=
\frac{O_j-AO_jA}{2}.
\end{aligned}
\end{equation*}
These components satisfy $[A,O_{j,+}]=0$ and $\{A,O_{j,-}\}=0$, where $\{B,C\}=BC+CB$ denotes the anticommutator.
Using the elementary involutory exponential identity,
$$
e^{i\eta A/2}
=
\cos(\eta/2)I
+
i\sin(\eta/2)A
$$
and multiplying the commuting and anticommuting components separately gives
$$
e^{i\eta A/2}O_je^{-i\eta A/2}
=
O_{j,+}
+
\cos(\eta)O_{j,-}
-
i\sin(\eta)O_{j,-}A.
$$
Subtracting this conjugated observable from $O_j=O_{j,+}+O_{j,-}$ and substituting into \cref{eq:core-exact-local} proves \cref{thm:core-exact-pauli-rotation}.

For $H_m=A/2$, anticommutation gives
$$
[A/2,O_j]
=
AO_{j,-}.
$$
Left multiplication by the unitary $A$ preserves the operator norm, so
$$
S_m^{\mathrm{op}}
=
\|O_{j,-}\|_\infty.
$$
For a unit-norm Pauli observable propagated to a single Pauli string, this value is $1$ when the string anticommutes with $A$ and $0$ otherwise, so it supplies no ranking resolution beyond the commutation test in that regime.
More generally, unitary invariance gives $\|O_j\|_\infty=\|O\|_\infty$ at every site, so backward propagation supplies no site-dependent norm scale.
The commutator proxy can vary only through the orientation and anticommuting content of $O_j$ relative to the local generator, with the norm of $O$ preserved by the suffix.
The corresponding first-order state sensitivity is
$$
S_m^{\mathrm{state}}
=
\left|
i\Tr(O_{j,-}A\tau_j)
\right|.
$$
The general definition writes $i\Tr(AO_{j,-}\tau_j)$ instead; anticommutation changes its sign, which disappears under the modulus.
It is the coefficient obtained by expanding the exact sine term; the full finite-angle expression remains available.

For the Pauli expansion $O_j=\sum_Pc_PP$, retain only the Hermitian Pauli strings that anticommute with the Hermitian Pauli axis $A$.
For every such string, $AP$ is an anti-Hermitian Pauli operator, meaning its conjugate transpose is its negative, so define the Hermitian Pauli string $P^*$ and sign $\epsilon(P)\in\{\pm1\}$ by
$$
AP
=
i\epsilon(P)P^*.
$$
Equivalently,
$$
[A/2,P]
=
i\epsilon(P)P^*.
$$
The exact finite-angle gap is then
\begin{equation}
\label{eq:exact-pauli-signed-expansion}
\Delta_m(\rho,O)
=
\left|
\sum_{P:\{P,A\}=0}
c_P
\left[
(1-\cos\eta)\Tr(P\tau_j)
+
\sin\eta\,\epsilon(P)\Tr(P^*\tau_j)
\right]
\right|.
\end{equation}
This expression uses the same propagated Pauli strings as the first-order coefficient but preserves the exact trigonometric dependence and the additional state expectation that becomes relevant when the sine contribution cancels.

To prove \cref{cor:finite-strength-visibility}, set
$$
a_m
=
\Tr(O_{j,-}\tau_j),
\qquad
b_m
=
i\Tr(O_{j,-}A\tau_j).
$$
Both coefficients are real because $O_{j,-}$ and $iO_{j,-}A$ are Hermitian.
Then \cref{eq:core-exact-pauli-gap} gives
$$
\Delta_m(\eta)
=
\left|(1-\cos\eta)a_m+\sin\eta\,b_m\right|.
$$
For $R_m=\sqrt{a_m^2+b_m^2}>0$, choose $\phi_m=\operatorname{atan2}(a_m,b_m)$.
The identities $R_m\sin\phi_m=a_m$ and $R_m\cos\phi_m=b_m$ give
\begin{equation}
\label{eq:finite-strength-phase-form}
(1-\cos\eta)a_m+\sin\eta\,b_m
=
a_m+R_m\sin(\eta-\phi_m).
\end{equation}
The sine ranges over $[-1,1]$, so the maximum absolute value is $|a_m|+R_m$ and is attained by aligning the sine term with the sign of $a_m$.
When $a_m\neq0$, one maximizing choice is $\eta_\star=\phi_m+\operatorname{sgn}(a_m)\pi/2$ modulo $2\pi$.
When $a_m=0$ and $b_m\neq0$, either sign of the quarter-period choice is maximizing.
For a restricted admissible set of strengths, the maximum follows by evaluating \cref{eq:finite-strength-gap} on that set.
The expression vanishes for every $\eta$ only if its sine and cosine coefficients vanish, which is equivalent to $a_m=b_m=0$.
For the finite-strength zero set, factor the signed shift as
\begin{equation*}
(1-\cos\eta)a_m+\sin\eta\,b_m
=
2\sin(\eta/2)
\left[
a_m\sin(\eta/2)+b_m\cos(\eta/2)
\right].
\end{equation*}
The first factor gives $\eta\equiv0\pmod{2\pi}$.
When $(a_m,b_m)\neq(0,0)$, the second factor gives $\eta\equiv2\operatorname{atan2}(-b_m,a_m)\pmod{2\pi}$, which proves \cref{eq:finite-strength-zero-set} without a division-by-zero convention.
This second residue coincides with zero exactly when $b_m=0$.
In the signed Pauli expansion, the same coefficients are
$$
a_m
=
\sum_{P:\{P,A\}=0}c_P\Tr(P\tau_j),
\qquad
b_m
=
\sum_{P:\{P,A\}=0}c_P\epsilon(P)\Tr(P^*\tau_j).
$$

The scope includes common same-axis mutation operators directly.
Writing $R_A(\theta)=e^{-i\theta A/2}$ gives
$$
R_A(\theta')R_A(\theta)^\dagger
=
e^{-i(\theta'-\theta)A/2}.
$$
Gate removal has $\theta'=0$, sign reversal has $\theta'=-\theta$, and a fixed offset has $\theta'=\theta+\eta$.
No step of this derivation assumes that $A$ has weight one.

Now suppose that $W_m$ is Clifford on the mutated subsystem.
The normalizer property used here is the standard Clifford action on Pauli operators \cite[Sec.~3 and Table~1]{gottesman1999heisenberg}.
For every Pauli string $P$, Clifford conjugation yields
$$
W_m^\dagger P W_m
=
s_PP'_P,
$$
where $s_P\in\{-1,1\}$ and $P'_P$ is a Pauli string.
Consequently, for $O_j=\sum_Pc_PP$,
\begin{equation}
\label{eq:exact-clifford-observable-difference}
O_j-W_m^\dagger O_jW_m
=
\sum_Pc_P(P-s_PP'_P).
\end{equation}
Like terms should be combined before either exact expectation evaluation or pilot scoring.
If $W_m=e^{i\phi}B$ for a Pauli string $B$, define
$$
O_{j,-}^{(B)}
=
\frac{O_j-BO_jB}{2}.
$$
Then \cref{eq:exact-clifford-observable-difference} simplifies to
$$
O_j-W_m^\dagger O_jW_m
=
2O_{j,-}^{(B)}.
$$
Thus, a Pauli relative mutation changes exactly the anticommuting component, while a general Clifford relative mutation may also permute Pauli strings.
In either case, constructing the observable difference is exact, but evaluating its expectation still requires state information.

Finally, suppose a relative mutation has the form $W_m=e^{-i\eta H}$ and the spectral decomposition $H=\sum_k\lambda_k\Pi_k$ is available.
Then
$$
W_m
=
\sum_k e^{-i\eta\lambda_k}\Pi_k,
\qquad
W_m^\dagger
=
\sum_k e^{i\eta\lambda_k}\Pi_k,
$$
so direct multiplication gives
\begin{equation}
\label{eq:exact-spectral-conjugation}
e^{i\eta H}O_je^{-i\eta H}
=
\sum_{k,\ell}e^{i\eta(\lambda_k-\lambda_\ell)}\Pi_kO_j\Pi_\ell.
\end{equation}
If $H$ has two distinct eigenvalues $\lambda_+$ and $\lambda_-$, define
$$
\alpha
=
\frac{\lambda_++\lambda_-}{2},
\qquad
\beta
=
\frac{\lambda_+-\lambda_-}{2},
\qquad
A
=
\frac{H-\alpha I}{\beta}.
$$
Then $A=A^\dagger$, $A^2=I$, and $e^{-i\eta H}$ differs by a global phase from $e^{-i(2\beta\eta)A/2}$.
Thus every two-eigenvalue generator reduces exactly to the involutory-axis rotation algebra, although evaluating the projector blocks or state expectations may still be computationally expensive.

For completeness, every fixed one-qubit relative unitary admits the constructive reduction used in \cref{cor:fixed-one-qubit-replacement}.
Here $U(2)$ is the set of two-by-two unitary matrices, and $SU(2)$ is its subset with determinant one, obtained by removing a scalar phase.
Let $W_m\in U(2)$ and choose one square root of $\det(W_m)$ so that
\begin{equation*}
V_m
=
\frac{W_m}{\sqrt{\det(W_m)}}
\in
SU(2)
\end{equation*}
has positive trace whenever its trace is nonzero; the two square-root choices differ only by the globally irrelevant sign of $V_m$.
Every $SU(2)$ matrix has the expansion
\begin{equation}
\label{eq:one-qubit-su2-expansion}
V_m
=
cI-i(b_xX+b_yY+b_zZ),
\qquad
c
=
\frac{1}{2}\Tr(V_m),
\qquad
b_k
=
\frac{i}{2}\Tr(\sigma_kV_m),
\end{equation}
where $(\sigma_x,\sigma_y,\sigma_z)=(X,Y,Z)$, the coefficients $c,b_x,b_y,b_z$ are real, and $c^2+\|\bm{b}\|_2^2=1$.
When $c=0$, fix the remaining sign ambiguity by requiring the first nonzero entry in the ordered tuple $(b_x,b_y,b_z)$ to be positive.
If $\|\bm{b}\|_2=0$, set $\eta_{\mathrm{eff}}=0$ and leave $A_{\mathrm{eff}}$ undefined; then $W_m$ is a global phase and the replacement is channel-equivalent to the original operation.
Otherwise, define
\begin{equation}
\label{eq:one-qubit-effective-axis-angle}
\eta_{\mathrm{eff}}
=
2\operatorname{atan2}(\|\bm{b}\|_2,c),
\qquad
\bm{n}
=
\frac{\bm{b}}{\|\bm{b}\|_2},
\qquad
A_{\mathrm{eff}}
=
\bm{n}\cdot\bm{\sigma}.
\end{equation}
The chosen sign makes $\eta_{\mathrm{eff}}\in(0,\pi]$, and \cref{eq:one-qubit-su2-expansion} becomes $V_m=e^{-i\eta_{\mathrm{eff}}A_{\mathrm{eff}}/2}$.
Multiplying back the discarded scalar proves \cref{eq:fixed-one-qubit-axis-angle}; because $A_{\mathrm{eff}}^2=I$, the exact involutory-axis formulas apply without a perturbative approximation.
This is the canonical representation used to compare fixed replacement channels.
For a declared fixed-axis family, retain separately its oriented axis and signed family coordinate $\eta$, because those conventions determine the sign of $b_m$ and the reported coordinate $\eta_{0,m}$ in \cref{eq:finite-strength-zero-set}.

\subsection{Fallback General-Generator Local Perturbation Form}
\label{sec:local-perturbation}
Lie-algebraic analyses study gate generators and their commutators to characterize variational circuits, whose gate parameters are optimized against an objective \cite{ragone2024lie}.
Such analyses examine circuit sensitivity and the concentration of gradients near zero; here we use the same generator viewpoint for one fixed mutant--test pair.
For $\widetilde G_j=e^{-i\eta H_m}G_j$, use the definitions of $S_m^{\mathrm{state}}$ and $S_m^{\mathrm{op}}$ in \cref{eq:core-state-sensitivity,eq:core-operator-sensitivity} and the remainder derivation in \cref{app:local-perturbation-derivation}.
Schatten H\"older duality gives $S_m^{\mathrm{state}}\leq S_m^{\mathrm{op}}$ \cite[Eq.~(1.174)]{watrous2018theory}.
The operator quantity therefore bounds the first-order response across states; the realized gap requires the selected state's expectations and finite-strength effects.
The deterministic interval in \cref{eq:perturbative-gap-interval} remains valid even when first-order dominance fails, although it may then leave the threshold unresolved.
A right-multiplied perturbation $G_je^{-i\eta H_m}$ instead uses the pre-gate state, or equivalently the conjugated generator $G_jH_mG_j^\dagger$ on $\tau_j$.
Keeping the left-multiplication convention is essential for the signed shift and for the state at which expectations are evaluated.
If $[H_m,O_j]=0$ is certified algebraically, exponentiation proves exact invisibility at every strength; $S_m^{\mathrm{state}}=0$ alone does not.

\subsection{Backward Light Cone Criterion}
\label{sec:backward-light-cone-proof}
Two related notions of backward support must be distinguished.
The exact algebraic support at location $j$ is $\supp(O_j)$, the smallest qubit set outside which $O_j$ acts as the identity.
Computing this exact support can itself be expensive, so the scalable filter uses the structural backward light cone obtained from the circuit dependency graph.
Applying the standard circuit support relation \cite[Eq.~(22)]{bravyi2026identity} to the reversed suffix $U_{>j}^\dagger$ gives
$$
\supp(O_j)
\subseteq
\mathrm{BLC}_j(O).
$$
If $\supp(W_m)\cap\mathrm{BLC}_j(O)=\varnothing$, this inclusion implies $\supp(W_m)\cap\supp(O_j)=\varnothing$.
Operators with disjoint support commute, so $[W_m,O_j]=0$ and hence $W_m^\dagger O_jW_m=O_j$.
Substitution into the exact formula in \cref{thm:core-exact-local} gives $\Delta_m(\rho,O)=0$ for every input state.
For the perturbative model, $W_m=e^{-i\eta H_m}$, so disjoint support of $H_m$ and $O_j$ is a readily checked sufficient condition for the same conclusion.
The structural inclusion may be strict because a gate can commute with the current operator or different algebraic terms can cancel, neither of which a connectivity-only traversal detects.
Therefore, membership in the structural cone is necessary but not sufficient for visibility.
Such a mutation may still be semantically real for another observable or input test.

\subsection{Circuit Characteristics That Predict Efficient Mutants}
\label{sec:predict-efficient-mutants}

The preceding quantities suggest several structural features that can be used to prioritize mutants before execution.
Some of these features give exact zero-effect conditions in the ideal circuit model, while others are practical heuristics for estimating which nonzero effects are likely to be detectable or classically tractable to approximate.

\paragraph{1. Light-cone membership.}
A candidate mutation must lie inside the structural backward light cone of the observable.
If its relative mutation support is outside this cone, the propagated observable has no support on the changed subsystem, and the mutation has zero effect on that observable in the ideal circuit model.
Membership is only a necessary condition for visibility; exact propagated support, commutation, or state-dependent cancellation can still make the effect zero.

\paragraph{2. Noncommutation with the propagated observable.}
For a gate at location $j$, compute or approximate the observable $O_j$ propagated backward to the mutation location.
For a Pauli-axis rotation, the exact effect is controlled by the anticommuting component $O_{j,-}$ in \cref{eq:core-exact-pauli-gap}; for a general small perturbation generated by $H_m$, the first-order effect is controlled by $[H_m,O_j]$.
The operator-level sensitivity in \cref{eq:core-operator-sensitivity} bounds the first-order response by its largest possible value over states; evaluating the circuit's actual response requires its particular state.
For a single Pauli propagated through a Clifford suffix, this proxy is binary and should not be expected to resolve candidates that all pass the noncommutation test.

\paragraph{3. State support.}
Noncommutation is not sufficient for detectability.
Even when $[H_m,O_j]\neq 0$, the particular state reaching the mutation location may have no support on the component of the commutator that changes the expectation value.
The corresponding state-dependent first-order sensitivity is given in \cref{eq:core-state-sensitivity}.
For a Pauli-axis rotation, the exact finite-angle response instead uses both state expectations in \cref{eq:exact-pauli-signed-expansion}.

\paragraph{4. Distance to measurement.}
Among mutations that satisfy the preceding visibility conditions, mutations closer to the measured observable are often easier to analyze.
Distance to measurement serves as a tractability heuristic; a zero-effect claim still requires one of the exact conditions above.
As $O$ is propagated backward through additional entangling gates, its support can grow, and $O_j$ can become a sum of many Pauli terms.
The state-side trend need not match the operator-side trend: a later mutation has a shorter suffix but a longer prefix, so its $O_j$ may be simpler while its $\tau_j$ is more entangled.
Consequently, distance to measurement alone does not predict either mutation effect or total estimation cost.
Mutations near the measurement layer often have more localized propagated observables, but the tractability of the state-dependent sensitivity still depends on both $O_j$ and $\tau_j$.

\paragraph{5. Observable locality.}
The locality of $O$ determines the initial size of the backward light cone.
If $O$ is local, then only mutations in its causal past can affect its expectation value, and this causal past may remain small in shallow or geometrically local circuits.
If $O$ is global, then many more mutation sites may be visible because the observable initially has support on many qubits.
However, global observables also tend to make both $O_j$ and the corresponding state-dependent sensitivity harder to reason about, especially after propagation through entangling layers.

An execution-planning algorithm for a fixed test is given in \cref{alg:efficient-mutant-selection}.
Its effect-evaluation subprocedure is given separately in \cref{alg:candidate-effect-evaluation} so that the two stages can flow across pages without compressing the decision logic.
Returned intervals are deterministic enclosures or simultaneous statistical intervals.
It predicts or labels outcomes for that test; adequacy evaluation requires an independently constructed mutant set.
In the algorithm, $C_m$ abbreviates $C(\rho,O,m)$ and $R_O$ is the fixed outcome range defined above.
Its threshold rule assigns \emph{skip--certified zero} only when exact mathematical identities or circuit structure prove that the gap is exactly zero.
It assigns \emph{skip--below test threshold} to a gap below $\gamma$ or a deterministic upper bound strictly below $\gamma$, and \emph{skip--below test threshold at declared confidence} when a simultaneous statistical interval lies wholly below $\gamma$.
A tolerance-compatible numerical zero retains numerical provenance and uses the below-threshold rule; the certified-zero label requires the exact proof specified above.
It assigns \emph{execute--predicted above threshold} when a gap or the lower endpoint of either interval reaches $\gamma$, and \emph{execute--threshold unresolved} otherwise.
An unavailable gap leaves the outcome unresolved; an exact observable-side update still requires state information for a threshold decision.
Every candidate record stores this execution-decision status separately from one of the evidence types in \cref{sec:claim-hierarchy}, and exact evidence also records its certificate or gap-computation provenance.
A \emph{skip--below budget floor} label certifies that the specified equal-prior discrimination success cannot be attained within the declared independent-copy budget; establishing $\Delta_m<\gamma$ requires separate test-specific evidence.

\begin{algorithm}[!ht]
\caption{Test-conditional mutant execution planning}
\label{alg:efficient-mutant-selection}
\small
\KwIn{Circuit $U$, test $(\rho,O,\gamma)$, candidate mutations $\mathcal{M}$, optional discrimination policy with budget $(N_{\mathrm{copy}},\delta_{\mathrm{disc}})$, pilot failure budget $\alpha_{\mathrm{pilot}}$, and pinned equivalence settings}
\KwOut{An execution-decision status, score, evidence type, and any exact-evidence provenance for every candidate}
\SetKw{Continue}{continue}
Initialize the decision records $\mathcal R\gets\emptyset$\;
Initialize $\mathcal M_{\mathrm{active}}\gets\mathcal M$; retain the declared population $\mathcal M$ unchanged\;
Allocate total pilot failure budget $\alpha_{\mathrm{pilot}}/|\mathcal M|$ per declared candidate before filtering (no allocation is needed if $\mathcal M=\varnothing$)\;
Verify pinned transpiler matches by exact channel comparison; record certified active candidates in $\mathcal R$ as \emph{skip--semantically equivalent} with \emph{exact} evidence and \emph{pinned transpiler match verified by exact channel comparison} provenance; remove them from $\mathcal M_{\mathrm{active}}$\;
Record active candidates satisfying \cref{eq:rotation-threshold-gate} in $\mathcal R$ as \emph{skip--below test threshold} with \emph{exact} bound evidence, interval $[0,R_O|\sin(\eta/2)|]$, and \emph{rotation threshold ceiling} provenance; remove them from $\mathcal M_{\mathrm{active}}$\;
If the separate discrimination policy is enabled, record active involutory-axis rotations below the periodic strength floor in \cref{eq:practitioner-minimum-strength} in $\mathcal R$ as \emph{skip--below budget floor} with \emph{exact} evidence and \emph{strength-bound certificate} provenance; remove them from $\mathcal M_{\mathrm{active}}$\;
Compute the structural backward light cone of $O$ by graph traversal\;

\ForEach{mutation $m \in \mathcal M_{\mathrm{active}}$}{
    Let $j$ be the mutation site of $m$\;

    \If{the subsystem changed by $m$ is outside the structural backward light cone at site $j$}{
        record $m$ in $\mathcal R$ as \emph{skip--certified zero} with \emph{exact} evidence and \emph{structural zero certificate} provenance for this test\;
        \Continue\;
    }

    Propagate $O$ backward to the mutation location, obtaining $O_j$\;
    Use \cref{alg:candidate-effect-evaluation} to obtain the available gap evidence, score, evidence type, and any exact-evidence provenance\;
    Apply the preceding threshold rule, record the execution-decision status in $\mathcal R$, and attach the returned score, evidence type, and provenance\;
}

Rank only the candidates labelled for execution, within comparable score types and according to the declared execution objective\;

\Return{$\mathcal R$}\;
\end{algorithm}

\begin{algorithm}[!htp]
\caption{Candidate effect and score evaluation}
\label{alg:candidate-effect-evaluation}
\small
\KwIn{Mutation $m$, propagated observable $O_j$, local test data, cost $C_m$, outcome range $R_O$, and available exact, state, or pilot information}
\KwOut{Gap, interval, or unavailable; score, evidence type, and provenance}

\If{a supplied Hermitian generator satisfies an algebraically certified $[H_m,O_j]=0$}{
    \Return{$\Delta_m=0$, exact zero score, \emph{exact} evidence, and \emph{generator commutation certificate} provenance}\;
}

\uIf{$m$ is a same-axis Pauli rotation change, including removal or sign reversal}{
    Construct $Q_m=D_{A,j}(\eta_m)$ from \cref{eq:structured-observable-difference}\;
}
\uElseIf{the relative mutation $W_m$ is Clifford}{
    Construct the exact observable difference $Q_m$ term by term using \cref{eq:exact-clifford-observable-difference}\;
}
\uElseIf{$m$ has a tractable spectral decomposition}{
    Construct $Q_m=O_j-W_m^\dagger O_jW_m$ using \cref{eq:exact-spectral-conjugation}\;
}
\uElseIf{$m$ has a general small Hermitian generator $H_m$}{
    \eIf{the local propagated state is available}{
        Compute $s_m=|\eta_m|S_m^{\mathrm{state}}$ and $B_m=\eta_m^2\|[H_m,[H_m,O_j]]\|_\infty/2$\;
        Form $I_m^{\mathrm{det}}$ from \cref{eq:perturbative-gap-interval}; attach score $s_m^2/(R_O^2C_m)$ only as a first-order ranking estimate\;
        \Return{the deterministic interval, remainder bound, perturbative score, and \emph{state-dependent perturbative} evidence}\;
    }{
        Compute $S_m^{\mathrm{op}}$ using \cref{eq:core-operator-sensitivity}\;
        \Return{the unavailable gap marker, operator-proxy score, and \emph{operator proxy} evidence}\;
    }
}
\Else{
    \eIf{exact simulation is tractable}{
        Compute $\Delta_m$ and $E_m=\Delta_m^2/(R_O^2C_m)$\;
        \Return{the exact gap, exact score, \emph{exact} evidence, and \emph{simulated gap} provenance}\;
    }{
        \Return{the unavailable gap marker, unavailable score, and \emph{unavailable} evidence}\;
    }
}

\If{an exact structured observable difference $Q_m$ was constructed}{
    \If{$Q_m=0$ was certified algebraically}{
        \Return{$\Delta_m=0$, exact zero score, \emph{exact} evidence, and \emph{algebraic zero certificate} provenance}\;
    }
    \If{$\Tr(Q_m\tau_j)=0$ was certified algebraically for the declared input}{
        \Return{$\Delta_m=0$, exact zero score, \emph{exact} evidence, and \emph{state-dependent cancellation certificate} provenance}\;
    }
    \uIf{the required expectations in $\tau_j$ are available}{
        Compute $\Delta_m=|\Tr(Q_m\tau_j)|$ and $E_m=\Delta_m^2/(R_O^2C_m)$\;
        \Return{the exact gap, exact score, \emph{exact} evidence, and \emph{structured gap} provenance}\;
    }
    \uElseIf{$m$ is a Pauli rotation or relative Clifford replacement with a tractable retained Pauli expansion and an executable pilot}{
        Apply the pilot procedure in \cref{sec:scalable} to obtain the signed estimate, simultaneous confidence interval $I_m=[\underline{\Delta}_m,\overline{\Delta}_m]$, and cancellation diagnostic\;
        \Return{the simultaneous interval, pilot score, and \emph{pilot-estimated} evidence}\;
    }
    \uElseIf{exact simulation is tractable}{
        Compute $\Delta_m$ and $E_m=\Delta_m^2/(R_O^2C_m)$\;
        \Return{the exact gap, exact score, \emph{exact} evidence, and \emph{simulated gap} provenance}\;
    }
    \Else{
        \Return{the unavailable gap marker, unavailable score, \emph{exact observable-side update only} evidence, and \emph{structured observable-difference} provenance}\;
    }
}
\end{algorithm}

\FloatBarrier
A fixed one-qubit replacement not handled by the preceding rotation or Clifford branches enters the tractable spectral branch of \cref{alg:candidate-effect-evaluation}.
The axis--angle reduction in \cref{cor:fixed-one-qubit-replacement,eq:one-qubit-effective-axis-angle} gives an equivalent construction of its exact observable difference through the involutory-axis formula, including non-Pauli axes; a scalar relative unitary gives $Q_m=0$ without assigning an axis.
For a simultaneous Euler-parameter path, this construction is applied separately to each fixed candidate and does not assume a common axis along the path.
Approximate propagation supports pilot intervals only when all discarded coefficient mass is bounded; it never supplies a proof that the full gap is exactly zero.
For small circuits with explicitly exact gates and inputs, a separate symbolic path constructs $U^\dagger OU-U_m^\dagger OU_m$ from the complete circuits and proves either observable invariance or a zero expectation in $\rho$.
This is equivalent to the local zero conditions by \cref{thm:core-exact-local} and provides a check independent of the implemented local-response formulas.
The exact-input path accepts recognized discrete Clifford gates, explicitly exact one-qubit or multiqubit unitary gates and Pauli rotations, and explicitly exact pure or mixed states, including computational-basis inputs.
Angles must be supplied as symbolic expressions such as rational multiples of $\pi$; rounded angles are not promoted to exact declarations.
The default symbolic width limit is four qubits, and supplied floating-point observable coefficients retain their binary-rational values.
Unsupported data or an exceeded resource limit leaves a certificate unavailable, preserving the numerical or pilot fallback.
The semantic filter verifies each pinned transpiler match independently of the selected input and observable.
It first checks common prefix and suffix instructions using exact gate semantics and ordered operands, then proves that the relative unitary on the remaining changed support is a scalar of unit modulus times the identity.
This local proof uses the symbolic width limit on the changed support and can therefore certify supported circuits wider than four qubits.
If it is unavailable, a full-register symbolic comparison is attempted only within both the symbolic and dense width limits.
The implementation propagates discrete Clifford gates directly on Pauli strings and evaluates general local conjugations on their bounded gate supports.
Its relative-Clifford branch still uses numerical recognition, so it expands the actual local matrix rather than substituting a rounded Clifford tableau; exact zero labels require a separate algebraic proof.
A complete numerical Pauli expansion is used for state expectations when it fits the term cap; otherwise the planner attempts bounded dense evaluation, a pilot with a truncation allowance, or bounded direct simulation.
Exhausted resources leave the candidate unresolved, and a truncated observable alone never receives complete-observable evidence.
The built-in pilot samples a simulated mutation-site state separately for each candidate and allocates $\alpha_{\mathrm{pilot}}/|\mathcal M|$ before filtering, without recycling unused allocations.
It estimates the fixed-strength difference; reusable coefficient pairs are outside the current implementation.
The remaining dense gaps and deterministic bounds use ordinary floating-point arithmetic to evaluate ideal formulas, with numerical errors left uncertified.
Results therefore retain numerical provenance and precision; an application requiring rigor at a numerical threshold boundary must use validated arithmetic or leave that comparison unresolved.
The operator score is a state-independent first-order upper-bound proxy; estimating realized efficiency requires the test's state-dependent gap.
The pilot estimators in \cref{sec:scalable} use truncated observable-side updates for Pauli rotations and relative Clifford replacements, together with measured signed state expectations and explicit propagation and sampling errors.
They must not be applied to other replacement classes without a mutation-specific observable construction and error analysis.
When proxy, pilot, and realized scores are mixed, they must remain separately labeled, and rankings should be formed only within comparable score types.

\subsection{Interpretation and Consequence}
\label{sec:framework-interpretation}

In the noiseless setting, mutation testing asks whether the selected tests extract distinguishing information from the quantum program; hardware robustness requires a separate device model.
For structured mutations, a promising candidate lies in the backward light cone of the measured observable, changes the propagated observable under the exact relative-unitary action, has state support on that exact difference, produces an expectation-value shift large enough to be detected under the selected protocol, and does not require expensive state-preparation or measurement changes.
For a general generated perturbation, the commutator and state-dependent sensitivity provide the corresponding first-order fallback.

The central test-specific signal is the observable gap $\Delta_m(\rho,O)$.
When exact expectation evaluation is infeasible, a truncated structured observable difference can be paired with pilot measurements of the required signed expectations; when no exact structured form is available, first-order behavior can be approximated using state-dependent sensitivities, while operator-norm quantities bound the worst-case first-order response across states.
Thus, in the noiseless setting, circuit-aware analysis can reduce the cost of executing a fixed suite and can stratify mutants by expected difficulty.
It does not justify defining the mutant population by the same suite-specific gaps used to score that suite.

\section{Extended Examples and Pilot-Estimator Details}
\label{app:extended-examples}
The examples instantiate the observable gap in \cref{eq:core-observable-gap} and optional estimation-efficiency score in \cref{eq:core-efficiency}; the later sections develop pilot estimators for wider circuits.

\subsection{Concrete Examples of Mutation-Efficiency Scores}
For simplicity, all examples use pure input states and Pauli observables, so the cost term $C_m$ is treated as constant.
Within each fixed observable, the range $R_O$ is also constant, so the displayed proportional efficiency values suppress the common factor $1/(R_O^2C_m)$ and rank mutants by $\Delta_m^2$.
\begin{example}[Angle perturbation]
Consider
$$
U = R_y(\theta),
\qquad \rho = \ket{0}\bra{0},
\qquad O=Z.
$$
The original expectation is $\langle Z\rangle_U = \cos\theta$.
Mutate $R_y(\theta) \rightarrow R_y(\theta+\eta)$.
Then
$$
\Delta_m
=
|\cos\theta-\cos(\theta+\eta)|.
$$
For small $|\eta|$, this becomes $\Delta_m\approx |\eta\sin\theta|$ and $E_m \propto \eta^2\sin^2\theta$.
For $\theta=\pi/4$ and $\eta=0.05$, the exact values are $\Delta_m\approx0.036224$ and $E_m\propto1.312\times10^{-3}$, whereas the first-order values are approximately $0.035355$ and $1.25\times10^{-3}$.
At that small fixed offset the mutant is more efficient when $\theta$ is near $\pi/2$ and inefficient when $\theta$ is near $0$ or $\pi$; the exact finite-strength dependence is periodic and is governed by \cref{eq:finite-strength-gap}.
\end{example}
\begin{example}[Gate removal]
Again consider $U = R_y(\theta)$, $\rho = \ket{0}\bra{0}$, and $O=Z$.
Remove the gate: $R_y(\theta)\rightarrow I$.
This is exactly the Pauli-axis relative rotation in \cref{thm:core-exact-pauli-rotation} with $\eta=-\theta$, so the same exact formulas apply.
The original expectation is $\cos\theta$, whereas the mutant expectation is $1$.
Therefore,
$$
\Delta_m = |\cos\theta-1|.
$$
If $\theta=\pi/4$, then $\Delta_m\approx 0.293$ and $E_m\propto 8.6\times 10^{-2}$.
If $\theta=0.05$, then $\Delta_m\approx 0.00125$, so removing a near-identity gate can be highly inefficient.
\end{example}
\begin{example}[Pauli gate replacement]
Consider
$$
U=X,
\qquad \rho=\ket{0}\bra{0},
\qquad O=Z.
$$
The original circuit gives $\langle Z\rangle_U=-1$.
Mutating $X\rightarrow I$ gives $\langle Z\rangle_{U_m}=1$, so $\Delta_m=2$ and $E_m\propto 4$.
Mutating $X\rightarrow Y$, however, gives $Y\ket{0}=i\ket{1}$, which differs only by a global phase for this input and observable.
Thus $\Delta_m=0$.
For this input, the two output density operators are identical, so the replacement is output-state equivalent and therefore invisible to every observable, even though $X\neq Y$ syntactically.
The relative mutation is $W_m=YX=-iZ$, so the exact Pauli replacement formula also shows that every propagated observable commuting with $Z$ is unchanged for every state at the mutation site.
\end{example}
\begin{example}[Controlled-gate mutation]
Consider $U=\mathrm{CNOT}_{c,t}$, input $\ket{\psi}=\ket{+}_c\ket{0}_t$, and observable $O=Z_t$.
The original output is
$$
\frac{\ket{00}+\ket{11}}{\sqrt{2}},
$$
so $\langle Z_t\rangle_U=0$.
Removing the control gives $I_c\otimes X_t$, which maps the input to $\ket{+}\ket{1}$, so $\langle Z_t\rangle_{U_m}=-1$.
Hence $\Delta_m=1$ and $E_m\propto 1$.
If the input is instead $\ket{1}_c\ket{0}_t$, both the original CNOT and the mutated unconditional $X$ produce $\ket{1}_c\ket{1}_t$, so $\Delta_m=0$.
\end{example}
\begin{example}[Flipped control condition]
Again use $U=\mathrm{CNOT}_{c,t}$, $\ket{\psi}=\ket{+}_c\ket{0}_t$, and $O=Z_t$.
Flipping the control condition produces
$$
\frac{\ket{01}+\ket{10}}{\sqrt{2}}.
$$
Both original and mutant states have $\langle Z_t\rangle=0$, so $\Delta_m=0$ for $O=Z_t$.
But for $O=Z_cZ_t$, the original state has expectation $1$, whereas the mutant state has expectation $-1$.
Thus $\Delta_m=2$ and $E_m\propto 4$.
This example shows that mutant efficiency is observable-dependent.
\end{example}
\begin{example}[Entangling-gate removal]
Consider
$$
U=\mathrm{CNOT}_{c,t}(H_c\otimes I_t),
\qquad
\ket{\psi}=\ket{00},
\qquad
O=Z_cZ_t.
$$
The original circuit prepares the Bell state $(\ket{00}+\ket{11})/\sqrt{2}$, so $\langle Z_cZ_t\rangle_U=1$.
Removing the CNOT gives $\ket{+}_c\ket{0}_t$, for which $\langle Z_cZ_t\rangle=0$.
Hence $\Delta_m=1$.
For the local observable $Z_c$, both states give $\langle Z_c\rangle=0$, so the same mutation is invisible to the local marginal observable.
\end{example}
\begin{example}[Qubit permutation mutation]
Let $Z_1$ denote the Pauli-$Z$ observable on the first tensor factor.
Consider $\ket{\psi}=\ket{01}$ and $O=Z_1$.
Without mutation, $\langle Z_1\rangle=1$.
A swap mutation produces $\ket{10}$, so $\langle Z_1\rangle_{U_m}=-1$, $\Delta_m=2$, and $E_m\propto 4$.
But if the observable is symmetric, $O=Z_1+Z_2$, the expectation is unchanged by the swap and $\Delta_m=0$.
\end{example}
\begin{example}[Global phase mutation]
Let $U_m=e^{i\phi}U$.
Then
$$
U_m\rho U_m^\dagger=U\rho U^\dagger.
$$
Therefore, for every observable $O$, $\Delta_m=0$.
A global phase mutation is always equivalent and should be discarded when that equivalence is certified.
The reference implementation's sufficient transpiler-match filter can miss such a mutant because the compiled comparison retains global phase; the remaining rules then give a conservative test-relative decision rather than a semantic exclusion.
\end{example}
\begin{example}[Relative phase insertion]
Consider $U=I$, input $\ket{\psi}=\ket{+}$, and observable $O=X$.
Insert a relative phase gate $U_m=R_z(\eta)$.
The original expectation is $\langle X\rangle_U=1$, while $\langle X\rangle_{U_m}=\cos\eta$.
Thus
$$
\Delta_m=|1-\cos\eta|\approx \frac{\eta^2}{2}.
$$
For $\eta=0.05$, $\Delta_m\approx 0.00125$.
Under direct $X$-outcome detection, the original state never produces outcome $-1$, while the mutant produces it with probability $\sin^2(\eta/2)=\Theta(\eta^2)$.
After $n$ mutant shots, the probability of missing that outcome is $(1-\sin^2(\eta/2))^n\leq \exp[-n\sin^2(\eta/2)]$.
Thus, for fixed confidence, observing the rare outcome takes $\mathcal{O}(1/\eta^2)=\mathcal{O}(1/\Delta_m)$ shots, illustrating why the uniform $\mathcal{O}(1/\Delta_m^2)$ mean-estimation bound need not be tight.
However, if $O=Y$, then $\langle Y\rangle_U=0$ and $\langle Y\rangle_{U_m}=\sin\eta$, so $\Delta_m\approx |\eta|$.
Thus the same relative-phase insertion is weak under $X$ but stronger under $Y$.
\end{example}
\begin{example}[Subcircuit replacement]
Suppose a subcircuit prepares
$$
U\ket{00}
=
\frac{\ket{00}+\ket{11}}{\sqrt{2}},
$$
and the mutant subcircuit prepares
$$
U_m\ket{00}
=
\frac{\ket{01}+\ket{10}}{\sqrt{2}}.
$$
For $O=Z_1Z_2$, the original expectation is $1$, while the mutant expectation is $-1$.
Therefore, $\Delta_m=2$ and $E_m\propto 4$.
For the local observable $Z_1$, both states have expectation $0$, so $\Delta_m=0$.
Subcircuit replacement can therefore be highly efficient when the observable captures the changed semantic property, but invisible when the observable only measures an unaffected marginal.
\end{example}

\subsection{End-to-End GHZ Planning Walk-Through}
\label{sec:ghz-planning-walkthrough}

Consider the three-qubit GHZ preparation
$$
U
=
\mathrm{CNOT}_{0,2}\mathrm{CNOT}_{0,1}R_y^{(0)}(\pi/2),
$$
with input $\ket{000}$ and the single expectation test $O=Z_1$.
The eight-candidate set in \cref{tab:ghz-planning-walkthrough} combines illustrative gate deletions and insertions, as supported by existing quantum mutation-testing tools, with the structured rotation family analyzed here \cite{fortunato2022qmutpy_tool}.
These author-defined candidates form a hand-checkable planning example; outcome frequencies for tool-generated populations require a separate benchmark census.
Use the illustrative population gap threshold $\gamma=0.1$ and the default planner with the discrimination policy disabled.
The $0.01$-radian offset probes sensitivity; its representativeness as a programming fault remains unestablished.

\begin{table}[!htb]
\caption{End-to-end fixed-test planning for eight declared mutations of the three-qubit GHZ preparation.}
\label{tab:ghz-planning-walkthrough}
\centering
\small
\renewcommand{\arraystretch}{1.12}
\begin{tabular}{@{}clp{0.42\linewidth}l@{}}
\toprule
\textbf{ID} & \textbf{Mutation} & \textbf{First decisive rule} & \textbf{Exact $\Delta_m$} \\
\midrule
$m_1$
& delete $\mathrm{CNOT}_{0,2}$
& outside the backward light cone of $Z_1$
& $0$ \\
$m_2$
& append $X_2$
& outside the backward light cone of $Z_1$
& $0$ \\
$m_3$
& append $Z_1$
& $[Z_1,O]=0$
& $0$ \\
$m_4$
& append $X_1$
& noncommuting but $\Tr(Z_1\rho_{\mathrm{out}})=0$
& $0$ \\
$m_5$
& $R_y^{(0)}(\pi/2)\to R_y^{(0)}(-\pi/2)$
& $a_m=0$, $b_m=1$, and $\eta=-\pi\equiv\eta_{0,m}$
& $0$ \\
$m_6$
& angle offset $\eta=0.01$
& $2|\sin(0.01/2)|<\gamma=0.1$
& $|\sin(0.01)|$ \\
$m_7$
& angle offset $\eta=\pi/2$
& exact predicted gap above threshold; retain for execution
& $1$ \\
$m_8$
& delete $\mathrm{CNOT}_{0,1}$
& no zero-effect certificate; retain for execution
& $1$ \\
\bottomrule
\end{tabular}
\end{table}

\FloatBarrier
The scalar threshold gate removes $m_6$ before circuit analysis, the structural traversal removes $m_1$ and $m_2$, commutation removes $m_3$, and exact state-dependent evaluation labels $m_4$ and $m_5$ as \emph{skip--certified zero} at their generated strengths, leaving only $m_7$ and $m_8$ for execution.
The executable walkthrough supplies exact angles and certifies both cancellations symbolically; the corresponding floating-point inputs retain numerical zero provenance instead.
Thus, this fixed-test planner schedules two of eight declared candidates without treating any test-relative zero as program equivalence, as summarized with population outcomes in \cref{sec:ghz-planning-summary}.
In a suite-adequacy study, all six non-executed but non-equivalent mutants remain in the independently fixed denominator, and the prespecified detection rule determines their survivor or kill labels independently of the planning disposition.

\subsection{Summary of Example Scores}
\Cref{tab:example-scores} collects the representative gaps and their test-dependent interpretations.

\begin{table}[!htb]
\caption{Representative gaps from the worked examples.
Each effect is conditional on the displayed input and observable where specified.}
\label{tab:example-scores}
\centering
\small
\renewcommand{\arraystretch}{1.15}
\begin{tabular}{@{}p{0.27\linewidth}p{0.43\linewidth}p{0.23\linewidth}@{}}
\toprule
\textbf{Mutation} & \textbf{Typical effect $\Delta_m$} & \textbf{Efficiency implication} \\
\midrule
$R_y(\theta)\rightarrow R_y(\theta+\eta)$
& $|\cos\theta-\cos(\theta+\eta)|$
& Large near $\theta=\pi/2$. \\
$R_y(\theta)\rightarrow I$
& $|\cos\theta-1|$
& Large for nontrivial rotations. \\
$X\rightarrow I$
& $2$ under $(\ket{0},Z)$
& Maximal. \\
$X\rightarrow Y$
& $0$ under $(\ket{0},Z)$
& Test-equivalent. \\
$\mathrm{CNOT}\rightarrow I\otimes X$
& $1$ under $(\ket{+}\ket{0},Z_t)$
& Efficient. \\
Flip CNOT control
& $0$ under $Z_t$; $2$ under $Z_cZ_t$
& Observable-dependent. \\
$\mathrm{CNOT}\rightarrow I$
& $1$ under $Z_cZ_t$
& Correlation-sensitive. \\
$\mathrm{SWAP}$
& $2$ under asymmetric $Z_1$
& Symmetry-dependent. \\
$e^{i\phi}U$
& $0$
& Equivalent. \\
$R_z(\eta)$ insertion on $\ket{+}$
& $1-\cos\eta$ under $X$; $|\sin\eta|$ under $Y$
& Observable-dependent. \\
Subcircuit replacement
& $2$ under $Z_1Z_2$; $0$ under $Z_1$
& Property-dependent. \\
\bottomrule
\end{tabular}
\end{table}

\subsection{Simultaneous Parameter Mutations of a General One-Qubit Gate}

Consider the standard three-parameter one-qubit Euler gate, also commonly denoted $U_3$, with the phase convention
\begin{equation*}
U(\theta,\phi,\lambda)
=
e^{i(\phi+\lambda)/2}
R_z(\phi)R_y(\theta)R_z(\lambda).
\end{equation*}
Write $\doteq$ for equality up to global phase.
The absence of an explicit $R_x$ factor is only a consequence of the ZYZ Euler convention:
\begin{equation*}
R_x(\theta)
=
R_z(-\pi/2)R_y(\theta)R_z(\pi/2),
\end{equation*}
so the decomposition remains universal up to global phase.
Let one mutation change two or three parameters simultaneously:
\begin{equation*}
(\theta,\phi,\lambda)
\longrightarrow
(\theta+\delta_\theta,\phi+\delta_\phi,\lambda+\delta_\lambda).
\end{equation*}
Under the left-multiplication convention, the complete relative mutation is
\begin{equation}
\label{eq:simultaneous-u-relative-unitary}
\begin{aligned}
W_m
&=
\widetilde G_jG_j^\dagger
\\
&\doteq
R_z(\phi+\delta_\phi)
R_y(\theta+\delta_\theta)
R_z(\delta_\lambda)
R_y(-\theta)
R_z(-\phi).
\end{aligned}
\end{equation}
The original value of $\lambda$ cancels from \cref{eq:simultaneous-u-relative-unitary} because the adjacent rotations combine as $R_z(\lambda+\delta_\lambda)R_z(-\lambda)=R_z(\delta_\lambda)$; only its offset remains, while the original $\theta$ and $\phi$ determine the orientations of later relative rotations.

The common one-parameter mutation operators are fixed-axis families:
\begin{equation}
\label{eq:u-single-parameter-families}
\begin{aligned}
\delta_\theta\text{ only:}\qquad
W_m
&\doteq
R_z(\phi)R_y(\delta_\theta)R_z(-\phi),
\\
A_\theta
&=
-\sin\phi\,X+\cos\phi\,Y,
\qquad
\eta=\delta_\theta;
\\[0.4em]
\delta_\phi\text{ only:}\qquad
W_m
&\doteq
R_z(\delta_\phi),
\\
A_\phi
&=
Z,
\qquad
\eta=\delta_\phi;
\\[0.4em]
\delta_\lambda\text{ only:}\qquad
W_m
&\doteq
R_z(\phi)R_y(\theta)R_z(\delta_\lambda)R_y(-\theta)R_z(-\phi),
\\
A_\lambda
&=
\sin\theta\cos\phi\,X
+
\sin\theta\sin\phi\,Y
+
\cos\theta\,Z,
\qquad
\eta=\delta_\lambda.
\end{aligned}
\end{equation}
Each axis in \cref{eq:u-single-parameter-families} is independent of its coordinate offset.
Consequently, one coefficient pair $(a_m,b_m)$ gives the complete zero set and maximum along the corresponding one-parameter sweep, with the signed coordinate offset used directly as $\eta$.
For each generated replacement, its canonical effective angle supplies the strength used by the test-independent floor in \cref{cor:test-independent-rotation-ceiling}.

For a simultaneous mutation, the component parameter effects must not be scored and added independently, because their ordered rotations can reinforce or cancel one another.
For each fixed offset vector $\bm{\delta}=(\delta_\theta,\delta_\phi,\delta_\lambda)$, first construct the complete $W_m$ in \cref{eq:simultaneous-u-relative-unitary}, then use \cref{eq:one-qubit-effective-axis-angle} to obtain $(\eta_{\mathrm{eff}},A_{\mathrm{eff}})$, and finally evaluate the exact gap through \cref{eq:core-exact-pauli-gap}.
This reduction also detects coordinate changes that implement only a global phase, for which the mutant is semantically equivalent despite a nonzero Euler-coordinate offset.
Two singular families make these null directions explicit:
\begin{equation}
\label{eq:u-coordinate-singularities}
\begin{aligned}
\theta=0,
\quad
\delta_\theta=0:
&\qquad
W_m
\doteq
R_z(\delta_\phi+\delta_\lambda),
\\
\theta=\pi,
\quad
\delta_\theta=0:
&\qquad
W_m
\doteq
R_z(\delta_\phi-\delta_\lambda).
\end{aligned}
\end{equation}
Thus, modulo $2\pi$, $\delta_\phi=-\delta_\lambda$ is a null direction at $\theta=0$, while $\delta_\phi=\delta_\lambda$ is a null direction at $\theta=\pi$.
An individual offset by an integer multiple of $2\pi$ in any Euler coordinate likewise changes the gate by at most a global phase.
These cases are analytically certifiable as globally equivalent whenever they occur in the chosen mutant population.
Their incidence depends on that population's mutation operators and allowed offsets.

If a family is swept along $\bm{\delta}(t)=t\bm{v}$, then in general
\begin{equation*}
W_m(t)
\doteq
e^{-i\eta_{\mathrm{eff}}(t)A_{\mathrm{eff}}(t)/2}
\end{equation*}
with an axis $A_{\mathrm{eff}}(t)$ that varies with $t$.
Accordingly, one fixed coefficient pair $(a_m,b_m)$ does not generally determine the complete zero set or maximum over that path, even though every fixed candidate on it is covered exactly.
To state the local expansion as a matrix equality, remove the relative scalar phase and define
\begin{equation*}
V_m
=
e^{-i(\delta_\phi+\delta_\lambda)/2}W_m
\in SU(2).
\end{equation*}
In a neighborhood of $\bm{\delta}=0$, this is the near-identity, positive-trace branch of the determinant-normalized representative used in \cref{eq:one-qubit-su2-expansion}.
The local reason is then visible from its first-order generator:
\begin{equation*}
V_m
=
I
-
\frac{i}{2}
\left(
\delta_\theta A_\theta
+
\delta_\phi Z
+
\delta_\lambda A_\lambda
\right)
+
\mathcal{O}(\|\bm{\delta}\|_2^2).
\end{equation*}
The physical relative matrix is $W_m=e^{i(\delta_\phi+\delta_\lambda)/2}V_m$; its scalar phase cancels from conjugation, so the observable gap and canonical effective angle are unchanged.
Let $\bm{n}_\theta$, $\widehat{\bm{z}}$, and $\bm{n}_\lambda$ be the Bloch vectors of $A_\theta$, $Z$, and $A_\lambda$, respectively.
They satisfy $\bm{n}_\theta\cdot\widehat{\bm{z}}=0$, $\bm{n}_\theta\cdot\bm{n}_\lambda=0$, and $\widehat{\bm{z}}\cdot\bm{n}_\lambda=\cos\theta$, so
\begin{equation}
\label{eq:u-local-effective-strength}
\eta_{\mathrm{eff}}^2
=
\delta_\theta^2
+
\delta_\phi^2
+
\delta_\lambda^2
+
2\cos\theta\,\delta_\phi\delta_\lambda
+
\mathcal{O}(\|\bm{\delta}\|_2^3).
\end{equation}
At $\theta=\pi/2$, the three coordinate directions are locally orthogonal and the leading term is the Euclidean squared norm.
At $\theta=0$ and $\theta=\pi$, it instead becomes $\delta_\theta^2+(\delta_\phi+\delta_\lambda)^2$ and $\delta_\theta^2+(\delta_\phi-\delta_\lambda)^2$, respectively, exposing the cancellations in \cref{eq:u-coordinate-singularities}.
The Euclidean norm of the Euler-coordinate offset depends on the representation; the relative unitary up to global phase and its effective rotation angle provide invariant measures of the physical mutation.

For a generic effective axis $A_{\mathrm{eff}}=\bm{n}\cdot\bm{\sigma}$, the projection $O_{j,-}=(O_j-A_{\mathrm{eff}}O_jA_{\mathrm{eff}})/2$ is not a binary Pauli-string commutation filter.
Each retained Pauli string can contribute up to three local Pauli components to the coefficient observables $O_{j,-}$ and $iO_{j,-}A_{\mathrm{eff}}$, after like strings are combined.
This is only a constant-factor expansion at a one-qubit mutation site, but the Pauli-expansion analysis below specializes to Pauli-string axes because they permit the cheaper anticommutation test and single-partner-string update used by the pilot estimator.

\subsection{Pauli-Expansion View}
Expand the backward-propagated observable in a Pauli basis:
$$
O_j
=
\sum_{P\in\mathcal P_n}
c_P P.
$$
For the unnormalized $n$-qubit Pauli basis used here, the expansion is unique with coefficients $c_P=2^{-n}\Tr(PO_j)$, which are real because $P$ and $O_j$ are Hermitian.
Only Pauli terms that anticommute with the mutation axis $A_q$ contribute to the exact observable difference.
Write each Pauli string as
$$
P = B_q\otimes R,
$$
where $B_q\in\{I,X,Y,Z\}$ acts on the mutated qubit and $R$ acts on the remaining qubits.
Then
$$
[A_q,B_q\otimes R]=0
\quad\text{whenever}\quad
B_q=I
\quad\text{or}\quad
B_q=A_q.
$$
Therefore, the exact effect comes only from terms where
$$
B_q \in \{X,Y,Z\}\setminus\{A_q\}.
$$
For each such contributing string, the commutator with the rotation generator gives
$$
\left[\frac{A_q}{2},\, B_q\otimes R\right]
=
i\,\epsilon_{abc}\,C_q\otimes R,
$$
where $C_q$ is the third single-qubit Pauli (the one orthogonal to both $A_q$ and $B_q$) and $\epsilon_{abc}\in\{\pm 1\}$ is the Levi-Civita symbol associated with the ordered triple $(A_q,B_q,C_q)$.
Concretely, $[X,Y]/2 = iZ$, $[Y,Z]/2 = iX$, $[Z,X]/2 = iY$, and the opposite orderings give the corresponding minus signs.
Defining $P^* = C_q\otimes R$ (the same Pauli string with $B_q$ replaced by $C_q$), the exact gap is
$$
\Delta_m(\rho,O)
=
\left|
\sum_{P:\{P_q,A_q\}=0}
c_P
\left[
(1-\cos\eta)\Tr(P\tau_j)
+
\sin\eta\,\epsilon(P)\Tr(P^*\tau_j)
\right]
\right|,
$$
with $\epsilon(P)$ defined by $AP=i\epsilon(P)P^*$ in \cref{sec:exact-structured-mutations}; for this single-qubit factor it equals $\epsilon_{abc}$.
The familiar state-dependent sensitivity is the absolute coefficient of the sine term at $\eta=0$.
The exact expression also retains the $(1-\cos\eta)$ contribution, which is quadratic near zero and can remain nonzero when the first-order coefficient vanishes.
\subsection{How Surrounding Gates Amplify or Wipe Out the Mutation}
\paragraph{1. Wiping out by commutation.}
If the backward-propagated observable commutes with the mutation generator,
$$
[A_q,O_j]=0,
$$
then the left-multiplied rotational mutation has $\Delta_m=0$ exactly.
For example, a $Z$-rotation mutation followed by direct $Z$-measurement is invisible for every mutation strength because $[Z_q,Z_q]=0$.
\paragraph{2. Amplification by basis rotation.}
A later gate can rotate the measured observable into a noncommuting basis.
For example, if $U_{>j}=H_q$ and $O=Z_q$, then
$$
O_j = H_q Z_q H_q = X_q.
$$
A mutation in a $Z$-rotation has generator $Z_q$, and $[Z_q,X_q]\neq 0$.
The Hadamard therefore makes a $Z$-rotation mutation potentially visible.
The exact response depends on both $(1-\cos\eta)\Tr(X_q\tau_j)$ and $\sin\eta\Tr(Y_q\tau_j)$, with the latter giving the first-order term.
\paragraph{3. Spreading through entangling gates.}
Entangling gates can spread a local observable into a multiqubit observable.
For example, let the suffix contain $\mathrm{CNOT}_{q,r}$ and suppose the final observable is $O=Z_r$.
Then
$$
O_j
=
\mathrm{CNOT}_{q,r}^\dagger Z_r \mathrm{CNOT}_{q,r}
=
Z_qZ_r.
$$
Now an $R_y$ angle mutation on qubit $q$ has the exact gap
$$
\Delta_m
=
\left|
\left(1-\cos\eta\right)\Tr(Z_qZ_r\tau_j)
+
\sin\eta\Tr(X_qZ_r\tau_j)
\right|.
$$
Thus the effect of a one-qubit mutation can depend on two distinct two-qubit correlations, with $\Tr(X_qZ_r\tau_j)$ controlling the first-order response.
\paragraph{4. Wiping out by missing correlations.}
Even if $[A_q,O_j]\neq 0$, the mutation can still be invisible if the state has no support on the corresponding correlation.
In the previous example, $\Tr(X_qZ_r\tau_j)=0$ removes the first-order response but does not remove the possible quadratic contribution from $\Tr(Z_qZ_r\tau_j)$.
If both correlations vanish, the mutation is exactly invisible to this expectation-value test for every $\eta$.
\paragraph{5. Cancellation among Pauli components.}
If $O_j=\sum_P c_P P$, several terms may contribute with different signs:
$$
\Delta_m
=
\left|
\sum_{P:\{P_q,A_q\}=0}
c_P
\left[
(1-\cos\eta)\Tr(P\tau_j)
+
\sin\eta\,\epsilon(P)\Tr(P^*\tau_j)
\right]
\right|.
$$
The effect can be large if terms add coherently, or small if they cancel.
Thus destructive interference among propagated contributions can wipe out a mutation's observable effect even inside the light cone.
\subsection{Bounded Amplification}
\label{sec:bounded-amplification}
A suffix circuit can make a mutation visible in new observables and correlations, but it cannot make the mutation effect arbitrarily large.
Set $B_{A,j}=iO_{j,-}A$.
Anticommutation with $A$ makes $B_{A,j}$ Hermitian and gives
\begin{equation*}
\{O_{j,-},B_{A,j}\}=0,
\qquad
B_{A,j}^2=O_{j,-}^2.
\end{equation*}
Therefore, the exact structured difference satisfies
\begin{equation*}
D_{A,j}(\eta)^2
=
\left[(1-\cos\eta)^2+\sin^2\eta\right]O_{j,-}^2
=
4\sin^2(\eta/2)O_{j,-}^2.
\end{equation*}
Because $D_{A,j}(\eta)$ is Hermitian, taking operator norms gives
\begin{equation}
\label{eq:exact-rotation-operator-norm}
\|D_{A,j}(\eta)\|_\infty
=
2|\sin(\eta/2)|\,\|O_{j,-}\|_\infty.
\end{equation}
The exact local identity in \cref{thm:core-exact-local} then proves \cref{eq:exact-rotation-gap-envelope}.
The result recovers linear small-strength scaling and the global ceiling $2\|O\|_\infty$ at $\eta=\pi$.
For a unit-norm Pauli observable, the gap is therefore at most $2|\sin(\eta/2)|$.
Surrounding gates redistribute the mutation effect across observables and correlations within this bound.
\subsection{Restricted Pauli Test-Augmentation Screen}
\label{sec:restricted-pauli-augmentation}
Declare a nonempty admissible output family $\mathcal F_{\mathrm{out}}$ of Hermitian observables with $\|O\|_\infty\leq1$, together with an implementation and cost for each measurement.
For a finite candidate list, its restricted optimum is
\begin{equation}
\label{eq:output-restricted-augmentation}
V_{\mathcal F_{\mathrm{out}}}
=\max_{O\in\mathcal F_{\mathrm{out}}}|\Tr(B_O\Delta\tau_j)|,
\qquad B_O=U_{>j}^\dagger O U_{>j}.
\end{equation}
Cyclicity supplies the original-prefix identity
\begin{equation}
\label{eq:output-restricted-prefix}
\Tr(B_O\Delta\tau_j)
=\Tr\!\left[(B_O-W_m^\dagger B_OW_m)\tau_j\right].
\end{equation}
Thus $V_{\mathcal F_{\mathrm{out}}}\leq\|\Delta\tau_j\|_1$, but tractable screening additionally requires a manageable candidate list, propagated representations, and access to the required expectations.
For output Pauli measurements, $O=P_{\mathrm{out}}$ is measured by local basis changes followed by computational-basis measurements and multiplication of the outcomes on its support.
A non-Clifford suffix may nevertheless make $B_O$ expensive to represent or estimate; inexpensive output measurement does not imply inexpensive analysis.

A different restriction starts with site Pauli strings $P$ and admits output observables $U_{>j}PU_{>j}^\dagger$.
Their scores are $|\Tr(P\Delta\tau_j)|$, but their measurement implementations must be supplied separately.
One implementation applies $U_{>j}^\dagger$ to the output before measuring $P$; its inverse-suffix depth, routing, and compilation costs count in the test budget.
Executing a truncated prefix instead changes the tested interface and is not automatically an output test of the original program.
For a Clifford suffix the full Pauli family is preserved up to signs, so the two unrestricted Pauli screens coincide and local output basis changes suffice.
Weight caps or other restricted Pauli lists must still be mapped explicitly because Clifford propagation can change Pauli weight.
Accordingly, the site screen supplies an inexpensive implementable oracle only when both its analysis and the corresponding output measurement have bounded declared cost.

The Pauli restriction can be exponentially loose even with an identity suffix.
A \emph{Haar-random} orthogonal pair is sampled uniformly over orthogonal pure-state pairs, with no preferred basis.
Applying L\'evy's lemma (a concentration bound on high-dimensional spheres) \cite[Lemma~III.1]{hayden2006generic} and a union bound over the $4^n$ Pauli strings to such an $n$-qubit pair bounds the largest Pauli expectation difference by $\mathcal{O}(\sqrt{n}\,2^{-n/2})$ with high probability, even though the trace norm of the state difference is $2$.
A small retained-Pauli score therefore does not imply that the mutant is intrinsically hard to detect; its distinguishing observable may be non-Pauli or spread over exponentially many Pauli components.

For the running sign-reversal example, $\tau_{j,-}^{(Y)}=(\sin\theta\,X+\cos\theta\,Z)/2$ has trace norm one.
Applying \cref{eq:state-side-pauli-trace-norm} at $\eta=-2\theta$ gives the same optimum $2|\sin\theta|$ obtained from $\Delta\tau_j=\sin\theta\,X$ and its sign observable.
\subsection{Pilot Estimation of Mutation Effect in Wide Circuits}
\label{sec:scalable}
For wide circuits, exact evaluation of the state $\tau_j$ and the propagated observable $O_j$ can require exponential resources in the worst case, although structured circuit families may admit efficient exact methods.
The proposed extension below truncates the observable-side Pauli expansion and measures the required signed state expectations with a shared pilot experiment; its practical benefit remains unvalidated.
It therefore replaces an assumed state-support model with an empirical estimate and an explicit error budget, while unrestricted wide-circuit simulation retains its worst-case resource requirements.

\subsubsection{Computational Growth and Bounded Fallbacks}
\label{sec:planner-growth}
Let $n$ be the register width, $L_{>j}$ the number of suffix gates, $k$ the largest local gate width, $K_0$ the number of supplied observable terms, and $K_{\mathrm{ret}}$ the retained-term cap.
The scalar rotation ceiling requires constant work once the angle and outcome range are supplied.
Structural cone traversal inspects circuit operands without constructing a state or operator; the reference implementation also stores one support set per circuit boundary, requiring up to $\mathcal O(Ln)$ storage for a circuit with $L$ gates.
Compiler comparison has the cost of the pinned compilation procedure, followed by an independent exact proof; no polynomial bound for arbitrary compilation or equivalence checking is assumed.
For a verified factorization $U=VGP$ and $U_m=V\widetilde G P$, channel equivalence reduces to checking $\widetilde G G^\dagger=e^{i\phi}I$ on the changed support.
The reference local certificate obtains this factorization by comparing source instructions with exact semantics, rather than assuming it from mutation metadata.
Unsupported common context can therefore prevent this sufficient proof even when the circuits are equivalent.

An $n$-qubit Pauli string requires $\mathcal O(n)$ explicit storage, so $K$ strings require $\mathcal O(Kn)$ storage.
A Clifford suffix preserves term count, giving $\mathcal O(K_0L_{>j})$ local term updates; full-string copying and sign bookkeeping can add an $n$ factor in an implementation.
Conjugating a Pauli term by a general $k$-qubit gate produces at most $4^k$ local Pauli terms.
After each truncation, the next gate therefore generates at most $K_{\mathrm{ret}}4^k$ terms, in addition to the cost of obtaining its local conjugation rule.
The reference implementation computes and caches these rules within each gate update using $2^k\times2^k$ matrices, merges equal strings, and sorts coefficients for truncation; sorting $T$ coefficients costs $\mathcal O(T\log T)$, rather than the linear-time selection bound discussed below.
It separately retains all $K_0$ terms of a complete discrete-Clifford image for algebraic checks, including coefficients omitted by the pilot approximation.
Without truncation, general propagation can populate all $4^n$ Pauli strings; a term cap bounds the retained representation but does not guarantee a useful error bound.

A dense statevector stores $2^n$ amplitudes, and applying $L$ fixed-width gates costs $\mathcal O(L2^n)$ arithmetic operations under ordinary statevector simulation.
A density matrix or full-register operator stores $4^n$ entries; generic dense matrix multiplication costs $\mathcal O(8^n)$ scalar operations with conventional multiplication.
A pure unitary Choi ket also has $4^n$ amplitudes, while its dense Choi density matrix would have $16^n$ entries; neither is constructed by this implementation's symbolic comparison.
Symbolic unitary matrices have the same entry count as numerical unitary matrices, but exact expression growth and simplification add costs not captured by that count.
The four-qubit symbolic default is a computational safeguard, not a mathematical restriction on the identities.

The implementation declares independent limits: \texttt{local\_max\_qubits}=4 for local numerical matrices, \texttt{dense\_max\_qubits}=10 for full-register numerical operators and density-matrix evolution, and \texttt{symbolic\_max\_qubits}=4 for exact matrix comparisons.
State preparation and direct simulation have separate width limits, both ten by default; lowering the dense limit does not disable permitted statevector calculations or sparse checks.
Observable validation uses Pauli coefficients for Hermiticity and analytical spectral ranges for pairwise anticommuting or independent commuting Pauli families, or diagonalizes only the active support within the dense limit.
An unsupported over-limit range calculation or gate validation produces an explicit input-validation error rather than an unchecked contract or an unbounded allocation.
These limits constrain calculations initiated by the planner, not the storage of already supplied dense states and mutation operators.
The simulated pilot still prepares a numerical state, and its probability cache is capped at 32 measurement settings; the hardware shot bound in \cref{eq:shadow-sample-complexity} does not remove classical state-simulation costs.
The tableau-based polynomial bound in \cref{sec:stabilizer-family} applies under its separate assumptions and is not established for the complete reference implementation by these resource limits.

\subsubsection{Backward Pauli Propagation}
Assume that the final observable is supplied as an exact Pauli expansion or as an initial truncated expansion whose discarded mass is included below.
The backward propagation step follows the Heisenberg-picture Pauli propagation used by Rall et al.\ \cite[Sec.~II.C and Eqs.~(14)--(15)]{rall2019pauli}.
In the \emph{Heisenberg picture}, circuit evolution is accounted for by transforming the observable while holding the state at the chosen reference point fixed.
The coefficient-based truncation used here is closely related to later sparse-Pauli dynamics methods \cite{begusic2024sparse,begusic2025clifford}.
Starting from $O$, propagate each retained Pauli term backward through the suffix:
\begin{equation*}
O_j
=
U_{>j}^\dagger O U_{>j}.
\end{equation*}
Clifford conjugation maps each Pauli string to one Pauli string, whereas a non-Clifford gate can produce a linear combination.
For example,
\begin{equation*}
H Z H = X,
\qquad
\mathrm{CNOT}_{c,t}^\dagger Z_t \mathrm{CNOT}_{c,t} = Z_c Z_t,
\qquad
R_z(\varphi)^\dagger X R_z(\varphi)=\cos\varphi\,X-\sin\varphi\,Y.
\end{equation*}
To limit representation growth, retain the $K_{\mathrm{ret}}$ largest Pauli coefficients after each propagation step.
Rank coefficients by magnitude after combining contributions to identical Pauli strings.
If $T_s$ coefficients are presented for truncation at step $s$, linear-time selection and accumulation of the discarded coefficient mass require $\mathcal{O}(T_s)$ scalar operations, without sorting the full list.
For a list of at most $K_{\mathrm{ret}}$ retained terms, this selection cost is $\mathcal{O}(K_{\mathrm{ret}})$ per step when each term expands into at most a fixed number of terms, as for gates acting on a fixed number of qubits with supplied local conjugation rules.
Generating the expanded terms, combining identical strings, manipulating their representations, and the precision of coefficient arithmetic require separate cost accounting; the selection bound alone is not a bound on the full propagation cost.
If the terms discarded at step $s$ have total absolute coefficient mass $d_s$, then their operator norm is at most $d_s$.
Unitary conjugation preserves the operator norm, and successive discarded contributions add by the triangle inequality, so
\begin{equation}
\label{eq:pauli-propagation-error}
\left\|O_j-\widehat O_j\right\|_\infty
\leq
\varepsilon_j,
\qquad
\varepsilon_j
=
d_{\mathrm{init}}+\sum_{s>j}d_s.
\end{equation}
Here $d_{\mathrm{init}}=0$ when the supplied expansion of $O$ is exact.
For the exact observable-difference map
\begin{equation*}
\mathcal L_{W_m}(X)
=
X-W_m^\dagger XW_m,
\end{equation*}
the retained structured difference obeys
\begin{equation}
\label{eq:structured-truncation-error}
\left\|
\mathcal L_{W_m}(O_j)
-
\mathcal L_{W_m}(\widehat O_j)
\right\|_\infty
\leq
2\varepsilon_j.
\end{equation}
For a Pauli-axis relative rotation, let $\delta_j=O_j-\widehat O_j$ and $\delta_{j,-}=(\delta_j-A\delta_jA)/2$.
The exact norm identity in \cref{eq:exact-rotation-operator-norm} sharpens the general certificate to
\begin{equation}
\label{eq:rotation-truncation-error}
\left\|
\mathcal L_{W_m}(O_j)
-
\mathcal L_{W_m}(\widehat O_j)
\right\|_\infty
=
2|\sin(\eta/2)|\,\|\delta_{j,-}\|_\infty
\leq
2|\sin(\eta/2)|\varepsilon_j.
\end{equation}
Unlike the mutation-independent certificate, this rotation-specific error shrinks at the same linear rate as the structured signal near $\eta=0$.
Exact zero-effect labels must still use the structural light cone because terms discarded from $\widehat O_j$ may carry omitted support.

\subsubsection{Pilot Measurements of Signed State Support}
Prepare $\tau_j$ by executing the prefix $U_{<j}$ followed by the original gate $G_j$, and collect randomized single-qubit Pauli measurements.
Median-of-means classical-shadow post-processing gives simultaneous signed estimates
\begin{equation*}
\widehat r_j(Q)
\approx
\Tr(Q\tau_j)
\end{equation*}
for every retained Pauli string $Q$ required by the candidates at site $j$ \cite[Theorem~1 and Eq.~(S13)]{huang2020classicalshadows}.
For each target string, median-of-means post-processing partitions the independent shots into batches, averages the single-shot estimates within each batch, and takes the median of those averages.
A single-shot classical-shadow estimator is unbiased; the median of independent batch means need not be unbiased and can lie outside the physical interval $[-1,1]$.
Throughout this section, $\widehat r_j(Q)$ denotes that median-of-means estimate projected onto $[-1,1]$ before scoring.
Clipping can alter its bias, but it cannot increase pointwise absolute error because the true Pauli expectation lies in the interval, and it therefore preserves the simultaneous error guarantee below.
Independently of clipping, using the same expectation estimates in the signed sum and magnitude-only envelope makes the triangle inequality ensure $\widehat\kappa_m\in[0,1]$ whenever its denominator is nonzero.
For $M_{\mathrm{P}}$ target Pauli strings of weight at most $w$, substituting the weight-$w$ Pauli shadow norm into the simultaneous-estimation theorem gives \cite[Theorem~1 and Eq.~(S13); Lemma~3 and Eq.~(S50)]{huang2020classicalshadows}
\begin{equation}
\label{eq:shadow-sample-complexity}
N_{\mathrm{pilot}}
=
\mathcal{O}\!\left(
\frac{3^w\log(M_{\mathrm{P}}/\alpha_{\mathrm{pilot},j})}{\zeta^2}
\right)
\end{equation}
samples to estimate all targets to additive error $\zeta$ with simultaneous failure probability at most $\alpha_{\mathrm{pilot},j}$ for this fixed site and test.
The squared \emph{shadow norm} bounds the single-shot estimator's worst-case second moment; for a weight-$w$ Pauli string under these local measurements it is $3^w$, producing the weight dependence above.
The dependence is independent of the total register size $n$ for fixed $w$, but it is exponential in the retained Pauli weight, so truncation controls both classical propagation and pilot-measurement cost.
The same site-$j$ pilot data can be reused across all candidate mutations whose scores use those target strings.
This guarantee covers the declared target set in one pilot; coverage across sites and tests requires the allocation below.
For a study-wide failure target $\alpha_{\mathrm{pilot}}$, prespecify pilot allocations $\alpha_{\mathrm{pilot},j,r}$ with $\sum_{j,r}\alpha_{\mathrm{pilot},j,r}\leq\alpha_{\mathrm{pilot}}$ and include every reused target in its pilot's simultaneous guarantee.
The reference implementation collects a separate pilot per candidate, with its allocation policy documented in the implementation README \cite{quantummutantplanner}.
If $\tau_j$ belongs to an analytically tractable family, such as a known stabilizer or product state, its Pauli expectations may instead be evaluated exactly and no pilot error is incurred.

\subsubsection{Pauli-Rotation Coefficient Estimates}
Let
\begin{equation*}
\widehat O_j
=
\sum_{k=1}^{K_{\mathrm{ret}}}c_kP_k
\end{equation*}
be the retained propagated observable for a relative rotation $W_m=e^{-i\eta A/2}$.
Estimate the two strength-independent coefficients separately:
\begin{equation}
\label{eq:pilot-pauli-coefficients}
\widehat a_m
=
\sum_{\substack{1\leq k\leq K_{\mathrm{ret}}\\\{P_k,A\}=0}}
c_k\widehat r_j(P_k),
\qquad
\widehat b_m
=
\sum_{\substack{1\leq k\leq K_{\mathrm{ret}}\\\{P_k,A\}=0}}
c_k\epsilon(P_k)\widehat r_j(P_k^*),
\end{equation}
where $[A/2,P_k]=i\epsilon(P_k)P_k^*$.
This pair supports every admissible mutation strength at the site without collecting a new pilot sample.
Define
\begin{equation*}
L_m
=
\sum_{\substack{1\leq k\leq K_{\mathrm{ret}}\\\{P_k,A\}=0}}|c_k|.
\end{equation*}
On the simultaneous event $|\widehat r_j(Q)-\Tr(Q\tau_j)|\leq\zeta$ for every target string, projection onto the anticommuting component and \cref{eq:pauli-propagation-error} give
\begin{equation}
\label{eq:pilot-coefficient-errors}
|\widehat a_m-a_m|
\leq
\varepsilon_j+\zeta L_m,
\qquad
|\widehat b_m-b_m|
\leq
\varepsilon_j+\zeta L_m.
\end{equation}
The fixed-strength signed shift, gap estimate, and efficiency estimate are
\begin{equation}
\label{eq:pilot-fixed-strength-score}
\widehat d_m(\eta)
=
(1-\cos\eta)\widehat a_m+\sin\eta\,\widehat b_m,
\qquad
\widehat\Delta_m(\eta)
=
|\widehat d_m(\eta)|,
\qquad
\widehat E_m(\eta)
=
\frac{\widehat\Delta_m(\eta)^2}{R_O^2C_m}.
\end{equation}
Combining the rotation-specific truncation certificate with the pilot error gives
\begin{equation}
\label{eq:pilot-fixed-strength-error}
|\widehat d_m(\eta)-d_m(\eta)|
\leq
2|\sin(\eta/2)|\varepsilon_j
+
\left(
|1-\cos\eta|+|\sin\eta|
\right)\zeta L_m.
\end{equation}
The same right-hand side bounds $|\widehat\Delta_m(\eta)-\Delta_m(\eta)|$.
The strength-independent pilot visibility score and nontrivial zero estimate are
\begin{equation}
\label{eq:pilot-all-strength-diagnostics}
\widehat V_m
=
|\widehat a_m|+\sqrt{\widehat a_m^2+\widehat b_m^2},
\qquad
\widehat\eta_{0,m}
=
2\operatorname{atan2}(-\widehat b_m,\widehat a_m)
\pmod{2\pi}.
\end{equation}
Finite pilot data cannot certify the exact equality $a_m=b_m=0$.
Instead, the joint confidence region implied by \cref{eq:pilot-coefficient-errors} reports whether all-strength invisibility is compatible with the selected tolerance; $\widehat\eta_{0,m}$ should be reported with propagated uncertainty and omitted when that region contains the origin.

The magnitude-only companion
\begin{equation}
\label{eq:pilot-magnitude-envelope}
\widehat G_{m,\mathrm{abs}}^{\mathrm{rot}}(\eta)
=
\sum_{\substack{1\leq k\leq K_{\mathrm{ret}}\\\{P_k,A\}=0}}
|c_k|
\left[
|1-\cos\eta|\,|\widehat r_j(P_k)|
+
|\sin\eta|\,|\widehat r_j(P_k^*)|
\right]
\end{equation}
deliberately removes cancellation to provide an envelope of the estimated contributions.
When $\widehat G_{m,\mathrm{abs}}^{\mathrm{rot}}>0$, define
\begin{equation}
\label{eq:cancellation-risk}
\widehat\kappa_m^{\mathrm{rot}}
=
1-
\frac{\widehat\Delta_m(\eta)}
{\widehat G_{m,\mathrm{abs}}^{\mathrm{rot}}(\eta)}.
\end{equation}
A value near one flags retained contributions that are individually large but nearly cancel in the signed estimate.
The sampling error and the $2|\sin(\eta/2)|\varepsilon_j$ truncation allowance must be included before the magnitude envelope is treated as an upper confidence bound on the true gap.

\subsubsection{Relative Clifford Replacements}
For a relative Clifford mutation, conjugate every retained Pauli string exactly, combine like terms, and write
\begin{equation*}
\widehat Q_m
=
\widehat O_j-W_m^\dagger\widehat O_jW_m
=
\sum_{\ell}d_\ell Q_\ell.
\end{equation*}
The same site-$j$ pilot sample gives
\begin{equation}
\label{eq:pilot-clifford-score}
\widehat d_m^{\mathrm{Cl}}
=
\sum_\ell d_\ell\widehat r_j(Q_\ell),
\qquad
\widehat\Delta_m^{\mathrm{Cl}}
=
\left|\widehat d_m^{\mathrm{Cl}}\right|,
\qquad
\widehat G_{m,\mathrm{abs}}^{\mathrm{Cl}}
=
\sum_\ell|d_\ell|\,|\widehat r_j(Q_\ell)|.
\end{equation}
With $L_m^{\mathrm{Cl}}=\sum_\ell|d_\ell|$, the simultaneous pilot event gives
\begin{equation}
\label{eq:pilot-clifford-error}
\left|
\widehat d_m^{\mathrm{Cl}}
-
\Tr\!\left[
\left(O_j-W_m^\dagger O_jW_m\right)\tau_j
\right]
\right|
\leq
2\varepsilon_j+\zeta L_m^{\mathrm{Cl}}.
\end{equation}
The squared signed gap divided by $R_O^2C_m$ is the pilot efficiency estimate, while the relative difference between the signed and magnitude-only values is the Clifford cancellation diagnostic.
For a Pauli relative mutation, $\widehat Q_m$ is twice the retained component that anticommutes with that Pauli; for a general Clifford, combining permuted strings before measurement preserves observable-side cancellation.

\subsubsection{Cost, Amortization, and Interpretation}
\label{sec:pilot-cost}
Pilot preparation, execution, measurement, and classical post-processing must be charged to the total selection budget.
The weight dependence in \cref{eq:shadow-sample-complexity} supplies a measurement cost, without an additional ad hoc geometric correlation penalty.
The resulting quantities estimate retained signed expectations from pilot data with declared confidence; establishing a candidate's efficiency also requires the declared error and cost analysis.
They become exact only when $\varepsilon_j=0$ and the required state expectations are evaluated exactly.
The strength-independent pair $(\widehat a_m,\widehat b_m)$ is preferable to screening only at one chosen angle because it exposes possible all-strength invisibility, the maximum signal, and the nontrivial cancellation strength from the same pilot data.
This reuse can save work relative to separate candidate executions, but three-point interpolation and targeted coefficient measurements also reuse the response across strengths.
The comparison in \cref{sec:measurement-baselines} is therefore needed before attributing an economic advantage specifically to shared shadows.
At fixed $\eta$ and $\zeta$, the certified rotation error is
\begin{equation*}
2|\sin(\eta/2)|\varepsilon_j(K_{\mathrm{ret}})
+
\left(|1-\cos\eta|+|\sin\eta|\right)\zeta L_m(K_{\mathrm{ret}}).
\end{equation*}
Because
\begin{equation*}
|1-\cos\eta|+|\sin\eta|
=
2|\sin(\eta/2)|
\left(
|\sin(\eta/2)|+|\cos(\eta/2)|
\right),
\end{equation*}
both the truncation and pilot-sampling terms vanish at the same linear rate in $|\sin(\eta/2)|$ near zero.
Thus, decreasing the mutation strength does not by itself improve the relative precision of the certified interval.

A target error specifies estimation precision; a population-boundary decision uses the full interval and can remain unresolved when the gap lies near the detection threshold $\gamma$.
For a fixed test and site $j$, let $\mathcal C_j$ be the rotation candidates and admissible strengths sharing one pilot state, let $M_{\mathrm{P},j}$ be the number of target Pauli strings of maximum weight $w_j$, and define
\begin{equation*}
g(\eta)
=
|1-\cos\eta|+|\sin\eta|,
\qquad
B_j
=
\max_{m\in\mathcal C_j}g(\eta_m)L_m.
\end{equation*}
When $B_j>0$, keeping the pilot-sampling contribution in \cref{eq:pilot-fixed-strength-error} below a fraction $q\in(0,1)$ of a common target estimation error $\varepsilon_{\mathrm{target}}$ for every candidate is ensured by
\begin{equation}
\label{eq:pilot-break-even-accuracy}
\zeta
\leq
\frac{q\varepsilon_{\mathrm{target}}}{B_j}.
\end{equation}
Meeting that error target also requires the retained propagation to satisfy
\begin{equation}
\label{eq:pilot-break-even-truncation}
2|\sin(\eta_m/2)|\varepsilon_j
\leq
(1-q)\varepsilon_{\mathrm{target}}
\qquad
\text{for every }m\in\mathcal C_j;
\end{equation}
otherwise, increasing the pilot sample cannot compensate for the truncation interval.
The resulting pilot and direct paired-comparison shot counts scale as
\begin{equation}
\label{eq:pilot-versus-direct-shots}
N_j^{\mathrm{pilot}}
=
\mathcal{O}\!\left(
\frac{3^{w_j}B_j^2\log(M_{\mathrm{P},j}/\alpha_{\mathrm{pilot},j})}{q^2\varepsilon_{\mathrm{target}}^2}
\right),
\qquad
N^{\mathrm{pair}}
=
\mathcal{O}\!\left(
\frac{R_O^2\log(2/\alpha_{\mathrm{test},m})}{\varepsilon_{\mathrm{target}}^2}
\right).
\end{equation}
If $c_j^{\mathrm{pilot}}$ is the cost of one prefix pilot shot, $c_j^{\mathrm{pair}}$ is the cost of one paired original--mutant sample, and $s_j\in(0,1]$ is the fraction of candidate comparisons that the pilot allows the planner to avoid, the rough break-even scaling can be written
\begin{equation}
\label{eq:pilot-break-even-condition}
s_j|\mathcal C_j|
\geq
\kappa_{\mathrm{be}}
\frac{c_j^{\mathrm{pilot}}}{c_j^{\mathrm{pair}}}
\frac{3^{w_j}B_j^2}{q^2R_O^2}
\frac{\log(M_{\mathrm{P},j}/\alpha_{\mathrm{pilot},j})}{\log(2/\alpha_{\mathrm{test},m})}.
\end{equation}
Here $\kappa_{\mathrm{be}}>0$ absorbs constants suppressed by the big-$\mathcal{O}$ bounds, and $s_j$ is the realized avoided-comparison fraction.
This is a retrospective scaling comparison: a forecast requires calibrated constants and a pre-execution estimate or conservative bound for $s_j$ from disjoint calibration circuits.
The shared precision factor cancels, but coefficient mass, Pauli weight, truncation, and measurement costs remain; changing $K_{\mathrm{ret}}$ can improve one contribution while worsening another.
Strength reuse alone establishes no advantage over the reconstruction baselines below.

\subsubsection{Precision-Matched Analytical Baselines}
\label{sec:measurement-baselines}
Consider one fixed input, output observable, and mutation site with $N_{\mathrm{axis}}$ declared Pauli axes and $M$ declared axis--strength candidates.
All families share the original expectation $f(0)$, and $f_\ell(\eta)$ denotes the mutated expectation for axis $\ell$.
Assume that the full circuits at the interpolation strengths can be executed, the original prefix can be prepared for Pauli measurements, and all target families are fixed before sampling.
For an arbitrary Hermitian involution the interpolation baseline remains valid, but the Pauli-measurement baselines below require the declared tractable Pauli representations.
Make the comparisons at simultaneous additive gap error at most $\epsilon>0$ and a common failure budget $\beta\in(0,1)$, keeping threshold decisions and observed kill counts separate from this estimation criterion.
An estimate used for planning charges its failure probability to $\alpha_{\mathrm{pilot}}$; fresh observed classifications use their separate $\alpha_{\mathrm{test}}$ allocation as in \cref{thm:planner-soundness}.

The established three-point reconstruction \cite{ostaszewski2021structure,nakanishi2020sequential} gives, in the present signed convention,
\begin{equation}
\label{eq:three-point-coefficients}
a_\ell=f(0)-\frac{f_\ell(\pi/2)+f_\ell(-\pi/2)}{2},
\qquad b_\ell=\frac{f_\ell(-\pi/2)-f_\ell(\pi/2)}{2}.
\end{equation}
To propagate errors without allocating another failure probability to each strength, put $h=1-\cos\eta$ and $s=\sin\eta$.
The signed shift is
\begin{equation}
\label{eq:three-point-shift}
d_\ell(\eta)=h f(0)-\frac{h+s}{2}f_\ell(\pi/2)-\frac{h-s}{2}f_\ell(-\pi/2).
\end{equation}
If every measured configuration mean has error at most $e$, simultaneously, then
\begin{equation}
\label{eq:three-point-error}
|\widehat d_\ell(\eta)-d_\ell(\eta)|
\leq [h+\max(h,|s|)]e
\leq4e
\qquad\text{for every }\ell,\eta.
\end{equation}
The same bound holds for the absolute gaps by the reverse triangle inequality.
Thus $2N_{\mathrm{axis}}+1$ configurations, with $e=\epsilon/4$, suffice for every strength in every declared family, including a continuum; each sampled configuration requires enough shots to achieve that precision.
The common original estimate is measured once, and an exactly known original expectation can instead be inserted at zero sampling cost.
The auxiliary strengths need not belong to the mutant population, but their executable circuits and costs must be included.

Use a serial-time baseline in which each preparation, circuit execution, and measurement has cost $c$, and the method's classical analysis, compilation, and post-processing have total cost $T^{\mathrm{set}}$ in the same units.
Parallel execution requires a different scheduling model.
Empirical comparisons should also report quantum duration, classical time, and peak memory separately.
For independent outcomes of range $R$, Hoeffding's inequality permits the explicit mean-estimation shot count
\begin{equation}
\label{eq:baseline-shot-count}
N(R,e,u)=\left\lceil\frac{R^2}{2e^2}\log\frac{2}{u}\right\rceil.
\end{equation}
For direct execution, estimate the original and each of the $M$ mutants to error $\epsilon/2$, reusing the original estimate.
For interpolation, estimate the $2N_{\mathrm{axis}}+1$ configuration means to error $\epsilon/4$.
Let $S$ be the number of distinct nonidentity Pauli targets in the retained coefficient expansions, with maximum weight $w$.
For either prefix method, require the same certified truncation bound from \cref{eq:pilot-break-even-truncation} and put $\zeta=q\epsilon/B_j$ with $B_j>0$ as defined in \cref{sec:pilot-cost}.
Then simultaneous Pauli errors at most $\zeta$ keep every sampling contribution below $q\epsilon$; if $B_j=0$, no retained-coefficient sampling is needed and only the truncation interval remains.
The three elementary shot counts are
\begin{equation}
\label{eq:baseline-shot-allocations}
\begin{aligned}
n_{\mathrm d}&=N(R_O,\epsilon/2,\beta/(M+1)),\\
n_{\mathrm i}&=N(R_O,\epsilon/4,\beta/(2N_{\mathrm{axis}}+1)),\\
n_{\mathrm p}&=N(2,\zeta,\beta/S).
\end{aligned}
\end{equation}
For an explicit conservative shadow count matching the reference implementation, use $k=\lceil4\,3^w/\zeta^2\rceil$ shots per batch and the smallest positive odd integer $b\geq8\log(S/\beta)$ batches, so $N_{\mathrm{sh}}=kb$.
Indeed, a local-Pauli shadow estimator for a weight-at-most-$w$ string has variance at most $3^w$; Chebyshev's inequality bounds a batch's failure probability by $1/4$.
Hoeffding's inequality across independent batches bounds a median failure by $\exp(-b/8)$, and a union bound over $S$ targets gives failure at most $\beta$.
Clipping the medians preserves this guarantee.
For $w=1$, $S=2$, $\zeta=0.25$, and $\beta=0.05$, this gives $k=192$, $b=31$, and $N_{\mathrm{sh}}=5952$; a planner pilot uses its allocated per-candidate failure budget in place of $\beta$.
These counts suffice for the specified estimators; establishing optimal sample-complexity lower bounds requires a separate analysis.

Writing $c_0,c_m,c_{\ell,+},c_{\ell,-}$ for full-circuit shot costs, $c_{j,Q}$ for a targeted prefix Pauli shot, and $c_{j,\mathrm{sh}}$ for a randomized prefix shot, the comparable total costs are
\begin{equation}
\label{eq:four-baseline-costs}
\begin{aligned}
T_{\mathrm{direct}}&=T_{\mathrm{direct}}^{\mathrm{set}}+n_{\mathrm d}\left(c_0+\sum_{m=1}^{M}c_m\right),\\
T_{\mathrm{interp}}&=T_{\mathrm{interp}}^{\mathrm{set}}+n_{\mathrm i}\left(c_0+\sum_{\ell=1}^{N_{\mathrm{axis}}}(c_{\ell,+}+c_{\ell,-})\right),\\
T_{\mathrm{target}}&=T_{\mathrm{target}}^{\mathrm{set}}+n_{\mathrm p}\sum_{Q=1}^{S}c_{j,Q},\\
T_{\mathrm{shadow}}&=T_{\mathrm{shadow}}^{\mathrm{set}}+N_{\mathrm{sh}}c_{j,\mathrm{sh}}.
\end{aligned}
\end{equation}
The prefix setup costs include propagation and every discarded-coefficient certificate; all setup costs include the processing required to return the declared $M$ estimates.
The targeted baseline measures strings separately; grouping compatible measurements can improve it and should be included when available.
High Pauli weight does not introduce a $3^w$ shot penalty for a directly measured Pauli product with outcomes $\pm1$, although it can affect basis-change, readout, and per-shot costs.
With a single propagated Pauli and one noncommuting Pauli axis, $S\leq2$, so targeted measurement is the natural prefix competitor.

The conditional regimes follow directly from \cref{eq:four-baseline-costs}.
Direct execution is competitive when few candidates are requested and the setup or auxiliary-configuration cost exceeds the saved sampling cost.
For many strengths on a few axes, interpolation uses $2N_{\mathrm{axis}}+1$ configurations independently of $M$ and avoids a potentially intractable propagated Pauli expansion.
For a modest target list, especially at large weight, targeted prefix measurements avoid the shadow locality penalty and can also exploit a shorter prefix than the full interpolation circuits.
Shared shadows are attractive only when many targets can reuse low-weight measurements sufficiently to compensate for their sample count and setup cost: their precise comparison against separate targeted measurements is $T_{\mathrm{shadow}}^{\mathrm{set}}-T_{\mathrm{target}}^{\mathrm{set}}+N_{\mathrm{sh}}c_{j,\mathrm{sh}}<n_{\mathrm p}\sum_Q c_{j,Q}$.
Reuse across additional observables can enlarge that target union, but all targets and costs must enter its simultaneous guarantee.
For any proposed deployment, these comparisons must account for calibrated costs and compatible truncation error alongside candidate count.
All four methods can still leave threshold classifications unresolved near $\gamma$.

\subsubsection{Light-Cone Filtering}
Use \cref{eq:core-light-cone-discard} before collecting pilot data; a retained Pauli approximation cannot prove that the full mutation effect is exactly zero.

\subsubsection{Algorithm}
\begin{enumerate}
\item Expand $O$ in the Pauli basis and initialize the retained coefficient--Pauli list:
$$
T
\leftarrow
\{(c_P,P):c_P\neq 0\}.
$$
If this initial list is truncated, record its discarded coefficient mass as $d_{\mathrm{init}}$.
\item For gates from the end of the circuit to the beginning, conjugate each Pauli term through the gate.
If non-Clifford expansion occurs, combine contributions to identical Pauli strings and retain at most $K_{\mathrm{ret}}$ terms with largest coefficient magnitudes.
Record every discarded coefficient mass and compute $\varepsilon_j$ from \cref{eq:pauli-propagation-error}.
\item At each mutation site $j$, collect the union of retained Pauli strings required by its Pauli-rotation and relative Clifford candidates.
\item Execute the original prefix through $G_j$, collect one randomized single-qubit pilot sample, estimate all required signed expectations with a simultaneous accuracy target $(\zeta,\alpha_{\mathrm{pilot},j})$, and clip the raw estimates to $[-1,1]$.
\item For each Pauli-axis candidate, compute $(\widehat a_m,\widehat b_m)$ from \cref{eq:pilot-pauli-coefficients}, the fixed-strength score and error bar from \cref{eq:pilot-fixed-strength-score,eq:pilot-fixed-strength-error}, and the all-strength diagnostics from \cref{eq:pilot-all-strength-diagnostics}.
\item Separately compute the magnitude-only envelope in \cref{eq:pilot-magnitude-envelope} and report $\widehat\kappa_m^{\mathrm{rot}}$ as a cancellation diagnostic.
\item For each relative Clifford candidate, combine the retained observable-difference terms and compute \cref{eq:pilot-clifford-score,eq:pilot-clifford-error} from the same site-level pilot data.
\item Select $K_{\mathrm{ret}}$ against combined truncation, pilot error, and total cost; charge propagation and pilot costs to the execution-planning budget; record the evidence type separately from the execution-decision status; and report the result as a truncated estimate with separate truncation and pilot-sampling uncertainties.
\end{enumerate}
\subsubsection{Interpretation}
The pilot estimators preserve the signs needed to represent state orthogonality, Pauli-term interference, and cancellation between the sine and cosine contributions.
They also expose the actual resource limitation: high-weight retained strings require exponentially more local-Pauli shadow samples even though the bound is independent of the total register size at fixed weight.
For a single Pauli observable propagated through a Clifford suffix, a noncommuting finite-strength rotation generally requires the two signed expectations $\Tr(P\tau_j)$ and $\Tr(iPA\tau_j)$.
One suffices only when a coefficient is known or its strength multiplier vanishes; a known stabilizer tableau permits both to be computed without measurements.
Comparisons of execution-planning benefit should include the precision-matched interpolation and targeted-measurement baselines in \cref{sec:measurement-baselines}, alongside structural and operator-based heuristics and random selection, and must charge pilot costs to the execution budget.
\section{Detailed Pre-Execution Filters}
\label{app:filtering}
The formal score $E_m$ ranks normalized signals per comparison cost, but computing it may itself be costly.
Before estimating $\Delta_m$, inexpensive analysis can remove globally equivalent mutants or label test-relative zero effects without executing them.
Apply the population-accounting rule following \cref{eq:suite-mutation-score} to these exclusions and skips.
\subsection{Backward-Light-Cone Filter}
Apply the structural rule in \cref{eq:core-light-cone-discard}, whose proof is given in \cref{sec:backward-light-cone-proof}.
It uses only exact topology and gate supports, independently of the input state and of any truncated approximation to $O_j$.

\subsection{Transpiler-Based Equivalence Filter}
\label{app:transpiler-equivalence}
A test-independent filter uses circuit optimization.
Let $\mathsf T$ denote a semantics-preserving transpilation or optimization pass.
If a deterministic, semantics-preserving optimization pipeline maps $U$ and $U_m$ to the same canonical circuit representation, then the mutant can be treated as equivalent under that pipeline.
This gives a sufficient practical test:
$$
\mathsf T(U)=\mathsf T(U_m)
\quad\Rightarrow\quad
\text{discard } m.
$$
For example, a mutation that replaces a gate by another gate that is optimized away, merged with a neighbor, or reduced to the same basis-gate sequence need not be executed.
This specializes the established idea of compiler-based equivalent-mutant detection to quantum-circuit compilation \cite{papadakis2015tce,campos2026tbe}.
The implication requires equality of the complete compiled representation, including gate parameters, qubit and classical-bit mappings, and any output convention relevant to the test.
It also requires the same transpiler version, options, target, and random seed, or another mechanism that makes the transformation deterministic.
Most importantly, the filter assumes that the transpiler preserves the implemented channel for both inputs.
If the purpose of the analysis is to test the transpiler itself, this assumption is inappropriate.
Tolerance-based gate removal, numerical synthesis, or approximate circuit comparison does not supply the exact premise, even when synthesis approximation is nominally disabled.
The reference implementation therefore treats Qiskit matches only as candidates for an independent exact channel comparison of the source circuits, with a default four-qubit symbolic limit.
Only a successful proof supplies the semantic-exclusion certificate; unsupported or unproved matches continue through the remaining planning rules, and failure to prove equivalence does not prove non-equivalence.
\subsection{Circuit-Equivalence Filter}
\label{app:circuit-equivalence}
Transpiled text equality guarantees equivalence under the stated assumptions.
Two circuits may have different textual representations but the same physical action on states, possibly up to global phase.
Therefore, stronger filters may compare circuit objects, canonicalized DAGs, decision diagrams, stabilizer representations, or other circuit-equivalence abstractions; decision-diagram methods provide an established exact approach for relevant circuit families \cite{burgholzer2021advanced}.
\Cref{sec:semantic-contract,tab:equivalence-hierarchy} distinguish channel equivalence, state equivalence on a declared input domain, interface-relative observational equivalence, and equality of a selected expectation.
An adequacy population must state which semantics defines its equivalence exclusions; the unrestricted channel results continue to quantify over every circuit input.
For small unitary circuits, comparing the full unitaries up to global phase provides circuit-equivalence ground truth, whereas state-vector comparison for one input provides only output-state ground truth for that input.
For general quantum circuits, the standard promise version of non-identity checking up to global phase is QMA-complete, and the corresponding non-equivalence problem inherits that worst-case hardness \cite[Sec.~4, especially Definition~3, Theorem~1, and the remark following Definition~3 on p.~469]{janzing2005nonidentity}.
In QMA (Quantum Merlin--Arthur), a polynomial-time quantum verifier checks a polynomial-size quantum witness with bounded error.
Canonicalization, stabilizer methods, decision diagrams, and other specialized equivalence tools can nevertheless remove many mutants in tractable circuit families without solving the unrestricted problem.
\subsection{Mutant--Mutant Redundancy Filter}
The same filtering idea applies to both original-versus-mutant and mutant-versus-mutant comparisons.
If two mutants produce identical optimized circuits,
$$
\mathsf T(U_{m_1})=\mathsf T(U_{m_2}),
$$
then they are redundant under the same transpiler assumptions stated above.
For expectation-value behavior, define
$$
\mu_m(\rho,O)
=
\Tr\!\left(OU_m\rho U_m^\dagger\right).
$$
Two mutants are exactly redundant under the fixed expectation-value test when $\mu_{m_1}(\rho,O)=\mu_{m_2}(\rho,O)$.
Equality of the absolute original--mutant gaps, $\Delta_{m_1}=\Delta_{m_2}$, is not sufficient because the signed expectation shifts can point in opposite directions.
For a distribution-level test, redundancy instead requires equality of the relevant outcome distributions, and equality of their means is generally insufficient.
Two mutants that are redundant under one test may differ under another, so the classification must name the input, observable or measurement, and equivalence tolerance.
Certified equality of population detection indicators across the declared suite permits representative classification as specified in \cref{thm:planner-soundness}; transferring full distributions or sample-level outcomes requires the stronger corresponding equivalence.
For a fixed adequacy population, weight resolved class indicators by their original multiplicities: if class $c$ has size $w_c$ and common detection indicator $d_c$, the score is $\sum_c w_cd_c/\sum_c w_c$.
Unresolved classes instead contribute to the bounds in \cref{eq:planner-score-bounds}.
An unweighted score over representatives instead measures a different, class-level population and must be named as such.
Near-equivalence inferred from finite samples or approximate output profiles is only a budgeted clustering heuristic and can discard behavior that a stronger test would distinguish.
This supports a budgeted mutation strategy:
$$
\begin{aligned}
&\text{generate many candidates}
\rightarrow
\text{filter equivalent mutants}
\\
&\rightarrow
\text{cluster redundant mutants}
\rightarrow
\text{sample representatives}.
\end{aligned}
$$
\subsection{Subsumption and Mutation Buckets}
Subsumption is defined relative to a test collection and a detection rule.
Using $K_{\mathcal T}(m)$ from \cref{eq:suite-detection-set}, mutant $m_1$ dynamically subsumes $m_2$ when $K_{\mathcal T}(m_1)$ is nonempty and $K_{\mathcal T}(m_1)\subseteq K_{\mathcal T}(m_2)$, because any test that detects $m_1$ also detects $m_2$.
This framework uses the established test-collection-relative notion of dynamic subsumption \cite{ammann2014minimal,kurtz2016dominator}.
This relation cannot be inferred from a single value of $\Delta_m$ or from syntactic mutation type alone.
The population-level relation uses exact detection sets; subsumption inferred from finite-shot kills is a noisy estimate whose errors are correlated across mutants when original-circuit shots are reused.

Many mutation operators produce similar behavioral effects.
For example, replacing $X$ by $Y$ and replacing $X$ by $Z$ may both strongly disturb the same tests.
Executing both may provide little additional information.
We therefore group candidate mutants into buckets according to expected effect profile.
A bucket may be defined by gate type, mutation type, local topology, estimated $\Delta_m$, or empirical output-distance profile.
Under a fixed budget, we select at most $K_{\mathrm{bucket}}$ representatives per bucket.
For smoke testing or debugging, a practical rule is to prefer diverse high-effect mutants over many similar high-effect mutants.
For comparing suite adequacy, representative selection should instead preserve the declared population or stratification and include difficult mutants near the detection boundary; maximizing the evaluated suite's own gaps would make the assessment circular.
Unlike exact equivalence or established subsumption, bucket membership provides no semantic guarantee and must be evaluated by the coverage or fault-detection behavior lost when representatives are omitted.
\section{Local Perturbation Remainder}
\label{app:local-perturbation-derivation}

The first-order generator derivative is standard \cite[Sec.~III, especially Eqs.~(3), (5), and (6)]{schuld2019evaluating}; this appendix records only its local specialization and the nonasymptotic remainder used in \cref{eq:core-perturbative-bound}.
\begin{proof}[Proof of \Cref{eq:core-perturbative-bound}]
By \cref{eq:core-exact-local}, the signed local expectation difference is $\Tr((O_j-O_j(\eta))\tau_j)$ for
$$
O_j(t)
=
e^{itH_m}O_je^{-itH_m}.
$$
The standard generator derivative and one further differentiation give
$$
O_j'(t)
=
i e^{itH_m}[H_m,O_j]e^{-itH_m}
$$
and
$$
O_j''(t)
=
-e^{itH_m}[H_m,[H_m,O_j]]e^{-itH_m}.
$$
Taylor's theorem with an integral remainder and unitary invariance of the operator norm give
$$
O_j(\eta)
=
O_j+i\eta[H_m,O_j]+R_2(\eta),
\qquad
\|R_2(\eta)\|_\infty
\leq
\frac{\eta^2}{2}
\|[H_m,[H_m,O_j]]\|_\infty.
$$
Because $\|\tau_j\|_1=1$, H\"older's inequality bounds $|\Tr(R_2(\eta)\tau_j)|$ by $\|R_2(\eta)\|_\infty$.
The reverse triangle inequality then yields \cref{eq:core-perturbative-bound} with first-order magnitude $|\eta|S_m^{\mathrm{state}}$.
\end{proof}

\section{Mutation Operators, Test Oracles, and Applicability}
\label{app:mutation-tool-coverage}

\Cref{tab:tool-coverage} records what each tool or study changes, how its tests judge the resulting mutant, and which parts of our analysis can handle the edit.
The final column describes our proposed application of the analysis to the published operators; the cited tools supply the mutations and tests in the preceding columns.
The inventory records applicability by operator family; the fraction of generated mutants covered remains unmeasured.
QuanForge additionally applies gate and parameter mutations to quantum neural networks, parameterized quantum circuits trained as learning models \cite{shao2026quanforge}.

For each unitary edit, form $W_m=\widetilde G_jG_j^\dagger$ on the union of the original and replacement supports before choosing an analysis.
The table uses \emph{rotation} for a fixed involutory axis, \emph{one-qubit} for every fixed one-qubit relative unitary, \emph{Clifford} when $W_m$ is Clifford, and \emph{spectral} for other tractable local replacements.
These cases overlap; deletion uses $\widetilde G_j=I$ and insertion uses $G_j=I$.
A changed control or wire can enlarge the union support even when both operations have small arity.
A tractable observable update still requires state expectations to determine the test's gap.
Large-support non-involutory updates may instead need the bounded perturbative fallback; measurement, reset, and classical-control edits require a different model.

Passing one of these tests is distinct from establishing semantic equivalence.
QMutPy's syntactically equivalent gates and QCRMut's gate classes group compatible signatures; channel equivalence requires a separate semantic check.
QCRMut rejects unchanged circuit objects, while QuMu leaves equivalent-mutant identification to future work.
Our earlier unit-testing study removes matching output statevectors at tolerance $10^{-10}$ for the tested input \cite[Secs.~III-B, IV-A--IV-C]{miranskyy2025feasibility}.
Such input-specific numerical filtering must be distinguished from an algebraic certificate of channel equivalence.
An absence of semantic certification is not evidence that every generated mutant is non-equivalent.

The Rotation Gate Insertion (RGI) operator in that study is $W_m=\exp[-i\theta(\cos\phi\,X+\sin\phi\,Y)/2]$ on the inserted gate's qubit \cite[Sec.~IV-A]{miranskyy2025feasibility}.
For fixed $\phi$, its axis $A=\cos\phi\,X+\sin\phi\,Y$ satisfies $A^2=I$, so \cref{thm:core-exact-pauli-rotation,cor:fixed-one-qubit-replacement} apply even though $A$ need not be a Pauli string.
This establishes applicability to the published rotational operator; evaluating the present planner and establishing the fault realism of our illustrative $0.01$-radian offset require further evidence beyond that study.

\begin{table}[!htbp]
\caption{Published mutation operators and tests, mapped to the analyses available in this paper.
The final column uses the overlapping cases defined in \cref{app:mutation-tool-coverage}; equivalence checks are discussed separately in the text.}
\label{tab:tool-coverage}
\centering\small
\renewcommand{\arraystretch}{1.12}
\begin{tabularx}{\linewidth}{@{}>{\raggedright\arraybackslash}p{0.12\linewidth}*{3}{>{\raggedright\arraybackslash}X}@{}}
\toprule
\textbf{Tool / study} & \textbf{Mutations generated} & \textbf{Test used in the cited work} & \textbf{Analysis available here} \\
\midrule
QMutPy \cite[Sec.~2.2]{fortunato2022qmutpy_tool}
& Gate replacement, deletion, insertion; measurement insertion and deletion.
& Existing Python test assertions.
& Unitary edits: rotation, one-qubit, Clifford, or spectral according to $W_m$; measurement edits outside scope. \\
Muskit \cite{mendiluze2021muskit}
& Add, remove, and replace gates.
& Specified valid outputs and probabilities; wrong-output and statistical oracles.
& Same unitary cases; gate names alone do not establish a rotation family. \\
QCRMut \cite[Secs.~2.3--2.4, 3.2]{garcia2024qcrmut}
& Insert, delete, change gate name, or change attributes (parameters and qubit locations); immutable positions.
& Experiment uses specified outputs and chi-square comparison.
& Parameter changes may be rotations; other changes use one-qubit, Clifford, or spectral analysis on the support union. \\
QuMu \cite[Secs.~4--5]{polo2025mutation_based}
& Initialization errors, gate swaps, control-gate faults, gate duplication and deletion.
& Sampled output-histogram distance with tolerance.
& Unitary gate and control edits use the matching case; a changed input is outside the fixed-input model unless represented as a unitary preparation edit. \\
Unit-testing study \cite[Secs.~III-B, IV-A--IV-C]{miranskyy2025feasibility}
& Gate replacement, deletion, insertion, and Rotation Gate Insertion (RGI) with $\theta=\phi=\pi/180$.
& Statistical, Swap, Statevector, and Inverse tests.
& RGI has one-qubit support and a fixed involutory axis; other gate edits use the case determined by $W_m$. \\
\bottomrule
\end{tabularx}
\end{table}

\end{document}